\documentclass[12pt,one column]{IEEEtran}
\usepackage{setspace}
\pdfoutput=1
\usepackage{amsmath,mathtools}
\usepackage{amssymb}
\usepackage{amsthm}
\usepackage{mathabx}
\usepackage{bbm}
\usepackage{dsfont}
\usepackage{cite}
\usepackage{color}
\usepackage{epsfig}
\usepackage{epsf}
\usepackage{rotating}
\usepackage{mathrsfs}
\usepackage{epsfig}
\usepackage{graphics}
\usepackage{enumerate}
\usepackage{euscript}
\usepackage{accents}
\DeclareMathAccent{\wtilde}{\mathord}{largesymbols}{"65}
\usepackage[T3,T1]{fontenc}
\DeclareSymbolFont{tipa}{T3}{cmr}{m}{n}
\DeclareMathAccent{\invbreve}{\mathalpha}{tipa}{16}
\usepackage[multiple]{footmisc}
\usepackage{anyfontsize}
\usepackage{t1enc}
\usepackage{amssymb}
\usepackage{cite}
\usepackage{color}
\usepackage{epsfig}
\usepackage{epsf}
\usepackage{rotating}
\usepackage{mathrsfs}
\usepackage{epsfig}
\usepackage{graphics}
\usepackage{subfigure}

\theoremstyle{plain}

\newtheorem{lem}{Lemma}

\newtheorem{corollary}{Corollary}
\newtheorem{proposition}{Proposition}

\begin{document}
\vskip 1cm

\thispagestyle{empty} \vskip 1cm


\title{{Asynchronous Data Transmission over \\Gaussian Interference Channels with \\Stochastic Data Arrival }}
\author{ Kamyar Moshksar\\
\small Department of Pure Mathematics \\University of Waterloo\\ Waterloo, ON, Canada, N2L 3G1\\e-mail: kmoshksar@uwaterloo.ca} \maketitle

\begin{abstract}
This paper addresses a Gaussian interference channel with two transmitter-receiver~(Tx-Rx) pairs under stochastic data arrival~(GIC-SDA). Information bits arrive at the transmitters according to independent and asynchronous Bernoulli processes~(Tx-Tx~asynchrony). Each information source turns off after generating a given total number of bits. The transmissions are \textit{asynchronous} (Tx-Rx~asynchrony) in the sense that each Tx sends a codeword to its Rx immediately after there are enough bits available in its buffer. Such asynchronous style of transmission is shown to significantly reduce the transmission delay in comparison with the existing Tx-Rx synchronous transmission schemes. The receivers learn the activity frames of both transmitters by employing sequential joint-typicality detection. As a consequence, the GIC-SDA under Tx-Rx asynchrony is represented by a standard GIC with state known at the receivers. The cardinality of the state space is $\binom{2N_1+2N_2}{2N_2}$ in which $N_1, N_2$ are the numbers of transmitted codewords by the two transmitters. Each realization of the state imposes two sets of constraints on $N_1, N_2$ referred to as the geometric and reliability constraints.  In a scenario where the transmitters are only aware of the statistics of Tx-Tx~asynchrony, it is shown how one designs $N_1,N_2$ to achieve target transmission rates for both users and minimize the probability of unsuccessful decoding. An achievable region is characterized for the codebook~rates in a two-user GIC-SDA under the requirements that the transmissions be immediate and the receivers treat interference as noise. This region is described as the union of uncountably many polyhedrons and is in general disconnected and non-convex due to infeasibility of time~sharing. Special attention is given to the symmetric case where closed-form expressions are developed for the achievable codebook~rates.   
\end{abstract}
\section{Introduction}
\subsection{Summary of prior art}
The interference channel is the basic building block in modelling ad hoc wireless networks of separate Tx-Rx pairs. The Shannon capacity region of interference channels has been unknown for more than thirty years. It is shown in \cite{et1} that the classical random coding scheme developed by Han and Kobayashi~\cite{HK} achieves within one bit of the capacity region of the two-user Gaussian interference channel~(GIC) for all ranges of channel coefficients and signal-to-noise ratio~(SNR) values. 

One key assumption made in \cite{et1} and the references therein is that data is constantly available at the encoders. This is in contrast with the reality of communication networks where the incoming bit streams at the transmitters are burst-like in nature. A detailed account about why information theory has not yet left a mark in communication networks and several areas of common interest in both fields is given in~\cite{eph-haj}.  Recently, the authors in~\cite{sim} study a cognitive interference channel~\cite{dev} where the data arrival is burst-like at both the primary and secondary transmitters. The secondary transmitter regulates its average transmission power in order to maximize its throughput subject to two conditions, namely, the throughput of the primary user does not fall below a given threshold and the queues of both transmitters are stable. Reference~\cite{dash} studies a cognitive interference channel where the packets arrive at the primary transmitter according to a Poisson process. A protocol called sense-and-send is developed~in~\cite{dash} where the secondary transmitter breaks its message into several bursts and switches between sensing and sending while transmitting these bursts. It is shown that for high rates of data arrival at the primary transmitter, a decaying power profile at the secondary transmitter outperforms a strategy with constant power level during the transmission mode. 

The authors in~\cite{khude} address the burst-like~nature of interference using the concept of degraded message sets in multiuser information theory. In this setting, two sets of messages are mapped into the set of codewords at the transmitter side. If interference is present, the receiver decodes only one of the messages and if there is no interference, both messages are decoded. A similar idea of transmitting information in layers is adopted in~\cite{minero} in the context of the random access channel.

A two-user GIC is studied in~\cite{wang1} where the signal at each transmitter is randomly present or absent from time~slot to time~slot according to stationary and ergodic Bernoulli processes. This setup is referred to as a GIC with bursty traffic for which the capacity region with or without feedback is characterized within a constant gap. Reference~\cite{kim} assumes the same model for a GIC with burst-like traffic as in~\cite{wang1} and demonstrates the benefits in utilizing an in-band relay and multiple antennas at the transmitters and receivers.

 Another major assumption in \cite{et1} is that both transmitters are block-synchronous, i.e., they start to transmit their codewords in the same symbol interval or time slot in a time-slotted channel. This assumption is not necessarily valid in practice due to the fact that different transmitters do not become active simultaneously. Information theoretic studies on a network of block-asynchronous users is investigated in \cite{cov,hui} in the context of the Multiple Access Channel~(MAC) with two transmitters. In case the amount of mutual delay between the transmitted codewords by the two transmitters in negligible compared to the length of codewords, it is shown in \cite{cov} that the capacity region of a block-asynchronous~MAC coincides with the capacity region of a block-synchronous~MAC. If the amount of delay is comparable with the block-length, the authors in \cite{hui} show that the capacity region is still similar to the capacity region of a block-synchronous~MAC except that time-sharing is no longer feasible.  The collision MAC without feedback is introduced in \cite{massey1} where it is shown how transmitters can jointly design the so-called protocol sequences in order to reliably communicate in the presence of unknown delays between the blocks of any two transmitters. An approach is taken in \cite{calvo} to study a two-user interference channel with block-asynchronous users. Using a general formula for capacity in \cite{verdu}, the authors derive an expression for the capacity region of such channels. This expression is not a single-letter formulation of the capacity region and hence, it is not computable. By deriving a single-letter inner bound, the authors present an analysis of the block-asynchronous two-user interference channel where the main focus is to show that using Gaussian codewords is in general suboptimal. 

 \begin{figure}[t]
  \centering
  \includegraphics[scale=.7] {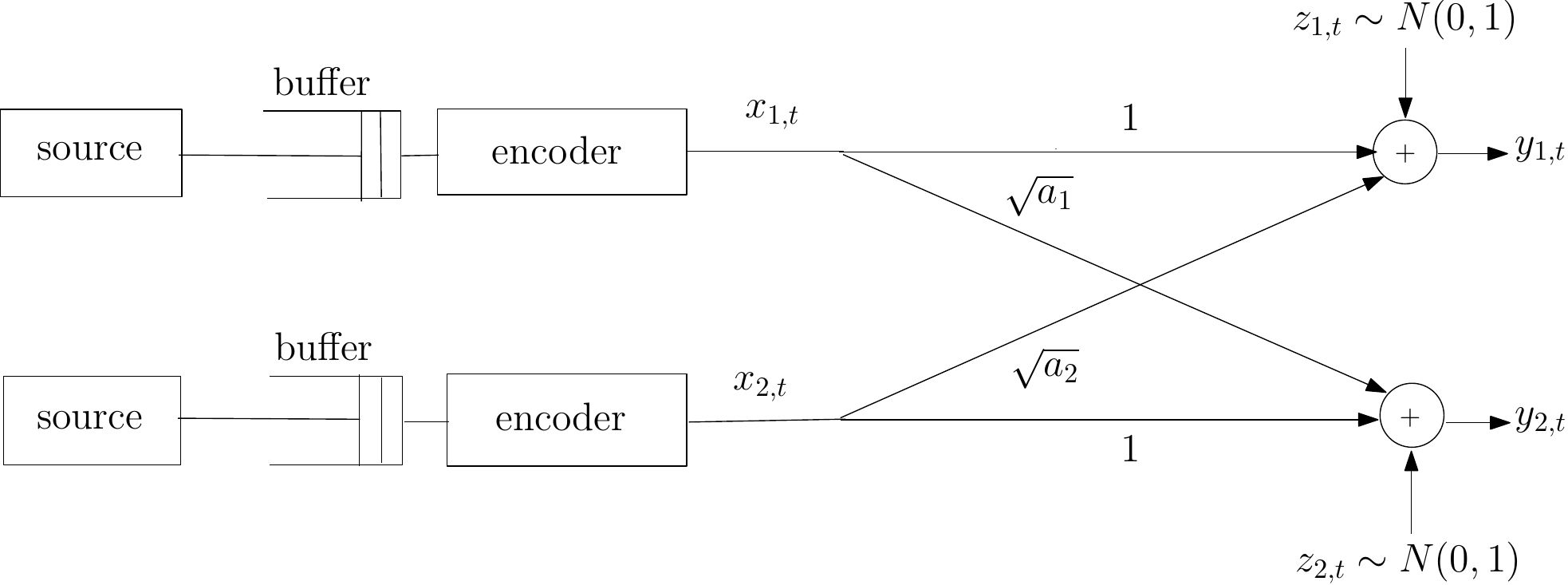}
  \caption{This figure shows a two-user GIC with stochastic data arrival~(GIC-SDA). The source of Tx~$i$ generates $k_i$ bits per time slot with a probability of $q_i$ and turns off after a total of $k_in$ bits are generated. The links from each transmitter to each receiver are modelled by static and non-frequency selective coefficients. The signals at the transmitters are subject to an average power constraint and the noise at each receiver is an AWGN process with unit variance.}
  \label{fig1}
 \end{figure} 
 
\subsection{Contributions}
We study a two-user GIC with stochastic data arrival (SDA) shown in Fig.~\ref{fig1}. The input bit streams at the transmitters are independent and asynchronous Bernoulli processes. The information source  at each transmitter turns off after randomly generating a given total number of bits. Let us consider two transmission schemes: 

1- Each transmitter begins to send a codeword only at a predetermined set of time slots $t_1<t_2<\cdots$ agreed upon between both Tx-Rx pairs. Transmission of a codeword at time instant $t_1$ is subject to availability of enough data bits to represent a codeword. If the number of available information bits is not enough, the transmitter waits until the earliest time slot $t_m$ for $m\geq 2$ when enough data is gathered in its buffer. This scheme is introduced in \cite{gamkim} which we refer to as the Tx-Rx synchronous scheme. 

2- Each transmitter begins to send a codeword immediately when there are enough bits gathered in its buffer. In this case, each receiver does not know a priori the time slots when the codewords are dispatched by the transmitters. We refer to this style of transmission immediate or Tx-Rx asynchronous. 

Under both the Tx-Rx synchronous and the Tx-Rx asynchronous schemes, the data sent by both transmitters can potentially look like intermittent bursts along the time axis. In the synchronous scheme, all symbols in a transmitted codeword are received either in the absence of interference or in the presence of interference. In the asynchronous scheme, however, a number of symbols per transmitted codeword may be received interference-free while the rest are received in the presence of interference. 

We compare both schemes in terms of the underlying \textit{relative delay} induced in the transmission process. More precisely, if the first symbol of the $j^{th}$ transmitted codeword is sent at time slots $T^{(j)}_{\mathrm{synch}}$ and $T^{(j)}_{\mathrm{asynch}}$  under the synchronous and asynchronous schemes, respectively, then there exists a $\delta>0$ such that the probability of $\frac{T^{(j)}_{\mathrm{synch}}-T^{(j)}_{\mathrm{asynch}}}{T^{(j)}_{\mathrm{asynch}}}>\delta$ occurring grows to $1$ in the asymptote of large codeword length. 

Applying sequential joint typicality decoding~\cite{Chandar11},  the receivers estimate and learn about the positions of the transmitted bursts along the time axis. We study fundamental constraints on the codebook~rates in order to guarantee immediate transmission at the transmitters and successful decoding at the receivers. For simplicity of presentation, we employ random Gaussian codebooks and assume all receivers treat interference as noise. The achievable region for codebook~rates is characterized as the union of uncountably many polyhedrons which is in general non-convex and disconnected due to infeasibility of time sharing. In a setup where the exact asynchrony between the input bit streams is unknown to the transmitters, the number of transmitted codewords at each transmitter is optimized to achieve a target transmission rate and minimize the probability of unsuccessful decoding at the receivers. 

Our analysis directly incorporates the burst-like nature of incoming data in the standard information-theoretic framework for reliable communications.     

\subsection{Notation}
Random quantities are shown in bold such that $\boldsymbol{x}$ with realizations $x$. Sets and in particular, events are shown using capital calligraphic or cursive letters such that $\mathcal{A}$ or $\mathscr{A}$. The set difference for two sets $\mathcal{A}$ and $\mathcal{B}$ is denoted by $\mathcal{A}\setminus\mathcal{B}$. The underlying probability measure and the expectation operator are denoted by $\mathbb{P}(\cdot)$ and $\mathbb{E}[\cdot]$, respectively. For a real number $x$, the floor of $x$ is $\lfloor x\rfloor$ and the ceiling of $x$ is $\lceil x\rceil$. A binomial random variable with parameters $n$ (number of trials) and $p$ (probability of success) is shown by $\mathrm{Bin}(n,p)$. A negative binomial random variable with parameters $k$ and $p$, denoted by $\mathrm{NB}(k,p)$, is defined to be the number of trials until $k$ successes are observed where $p$ is the probability of success. The probability density function~(PDF) of a Gaussian random variable with zero mean and variance $\sigma^2$ is shown by $\mathrm{g}(x;\sigma^2):=\frac{1}{\sqrt{2\pi}\sigma}e^{-\frac{x^2}{2\sigma^2}}$. The differential entropy of a continuous random variable $\boldsymbol{x}$ with PDF $p(\cdot)$ is shown by $h(\boldsymbol{x})$ or $h(p)$. For two functions $f$ and $g$ of a real variable $x$, we write $f=\Theta(g)$ if there is $x_0$ and constants $c_1$ and $c_2$ such that $c_1g(x)\leq f(x)\leq c_2g(x)$ for $x\ge x_0$. We define 
\begin{eqnarray}
\label{cap_fun11}
\mathsf{C}(x):=\frac{1}{2}\log(1+x). 
\end{eqnarray}
 All logarithms have base~2. Throughout the paper, 
 \begin{itemize}
  \item Any equality or inequality involving random variables is understood in the ``almost sure'' sense unless otherwise stated. We avoid repeating ``almost surely'' throughout the paper. 
  \item ``SLLN'' stands for ``the strong law of large numbers''. 
  \item The symbol ``$:=$'' means ``is defined by''
\end{itemize}

 \section{System Model}
 \subsection{Signalling and channel model}

We consider a GIC with two users of separate~Tx-Rx pairs shown in Fig.\!~\ref{fig1}. The channel from Tx~$i$ to Rx~$j$ is modelled by a static and non-frequency selective coefficient $h_{i,j}$ where $h_{1,1}=h_{2,2}=1$, $h_{2,1}=\sqrt{a_2}$ and $h_{1,2}=\sqrt{a_1}$.  The channel from each transmitter to each receiver is slotted in time and the time slots on any of the four channels from different transmitters to different receivers coincide. Therefore, the two users are symbol-synchronous. Throughout the paper, we show the time slots using the index $t=1,2,\cdots$. If we are describing a property for user~$i$, the index $i'$ refers to the other user, i.e., $i'=3-i$ for $i=1,2$. Denoting the signal at Tx~$i$ in time slot $t$ by $x_{i,t}$, we impose the average power constraint
\begin{equation}
\label{power_11}
Q_i:=\frac{1}{|\mathcal{T}_i|}\sum_{t\in\mathcal{T}_i}x_{i,t}^2\leq P_i,
\end{equation}
where $\mathcal{T}_i$ is the communication period of interest for Tx~$i$ and $|\mathcal{T}_i|$ denotes the length of $\mathcal{T}_i$. The signal $y_{i,t}$ received at Rx~$i$ in time slot $t$ is given by 
\begin{equation}
\label{system}
      y_{i,t}=x_{i,t}+\sqrt{a_{i'}}\,x_{i',t}+z_{i,t},\,\,\,\,i=1,2,            
\end{equation}
where $z_{i,t}$ is the additive noise at Rx~$i$ in time slot $t$. The noise at each receiver is an additive white Gaussian noise~(AWGN) process with unit variance.  
\subsection{Data arrival}
 Each transmitter is connected to an information source through a buffer as shown in Fig.~\ref{fig1}. At the ``beginning'' of time slot $t=1$, the buffers are empty.  At the ``end'' of each time slot, a number of $k_i$ bits arrive at the buffer of Tx~$i$ with a probability of $q_i$ or no bits arrive with a probability of $1-q_i$.\footnote{All results in the paper are still valid as long as the incoming bit stream is a random process with independent and identically distributed inter-arrival periods with finite mean value.}   The rate of data arrival at Tx~$i$ is denoted by 
 \begin{eqnarray}
\lambda_i:= k_iq_i. 
\end{eqnarray}
 The bit streams entering the buffers of the two users are independent processes. Source~$i$ is turned off permanently after it generates a total number of $k_in$ bits where $n$ runs in the set of positive integers. To transmit its data,  Tx~$i$ employs a codebook consisting of $2^{\lfloor n\eta_i \rfloor}$ codewords of length 
 \begin{equation}
\label{theta_11}
n_i:=\lfloor n\theta_i\rfloor,
\end{equation}
 where $\eta_i,\theta_i>0$. Note that the codebook~rate for Tx~$i$ is $\frac{\lfloor n\eta_i\rfloor}{\lfloor n\theta_i\rfloor}$. Assuming $\eta_i$ has the particular expression 
 \begin{equation}
\label{ }
\eta_i=\frac{k_i}{N_i},
\end{equation}
 for integers $N_1$ and $N_2$, we see that  Tx~$i$ sends a total number of $N_i$ codewords where each codeword represents $\lfloor n\eta_i\rfloor=\lfloor \frac{k_in}{N_i}\rfloor$ of the bits stored in its buffer. The number of bits that are not transmitted is equal to $k_in-N_i\lfloor \frac{k_in}{N_i}\rfloor\leq N_i$ which is negligible in the asymptote of large $n$.  Before a codeword is transmitted over the channel, each transmitter sends a preamble sequence  consisting of $n'=o(n)$  symbols.\footnote{This means $\lim_{n\to\infty}\frac{n'}{n}=0$.} Each preamble sequence enables the receivers to identify the arrival of a codeword. Details on the preamble sequences and how they are utilized are provided in Section~III.
  \begin{figure}[t]
  \centering
  \includegraphics[scale=.7] {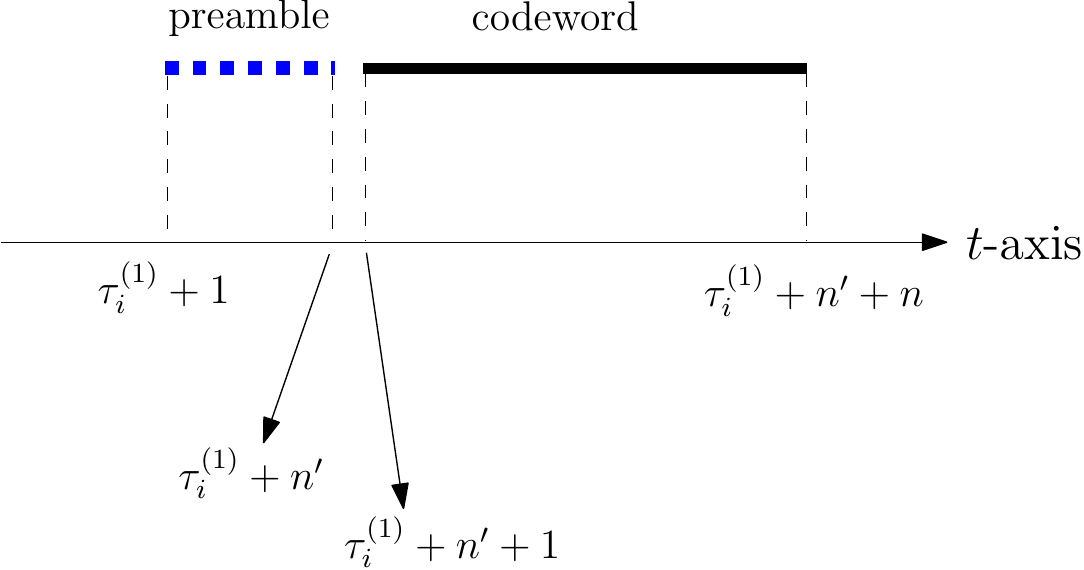}
  \caption{This figure shows the first transmission burst of Tx~$i$ along the $t$-axis. At the end of time slot $\tau_i^{(1)}$ the number of bits in the buffer of Tx~$i$ become larger than or equal to $\lfloor n\eta_i\rfloor$ for the first time. A number of $\lfloor n\eta_i\rfloor$ bits in the buffer of Tx~$i$ are represented by a codeword which together with the preamble sequence are sent during time slots $\tau^{(1)}_i+1,\cdots, \tau^{(1)}_i+n'+n_i$. }
  \label{bufpic2}
 \end{figure}

Let $b_{i,t}$ be the number of bits in the buffer of user $i$ at the ``beginning'' of time slot $t$, $b'_{i,t}$ be the number of bits entering the buffer of user $i$ at the ``end'' of  time slot $t$ and $\tau_i^{(1)}$ be the smallest index $t\geq 1$ such that $b_{i,t}+b'_{i,t}\geq \lfloor n\eta_i\rfloor$. At time slot $t=\tau^{(1)}_i+1$, a number of $\lfloor n\eta_i\rfloor$ bits in the buffer of Tx~$i$ are represented by a codeword which together with the preamble sequence are sent during time slots $\tau^{(1)}_i+1,\cdots, \tau^{(1)}_i+n'+n_i$. This is referred to as a \textit{transmission burst} or simply a burst as shown in Fig.~\ref{bufpic2}. These $\lfloor n\eta_i\rfloor$ bits are erased from the buffer of Tx~$i$, i.e., 
\begin{equation}
\label{ }
b_{i,\tau^{(1)}_i+1}=b_{i,\tau^{(1)}_i}+b'_{i,\tau^{(1)}_i}-\lfloor n\eta_i\rfloor.
\end{equation}
 Let $\tau_i^{(2)}$ be the smallest index $t\geq \tau^{(1)}_i+n'+n_i$ such that $b_{i,t}+b'_{i,t}\geq \lfloor n\eta_i\rfloor$. At time slot $t=\tau^{(2)}_i$, a second group of $\lfloor n\eta_i\rfloor$ bits in the buffer are scheduled for transmission. These bits are represented by a codeword which together with the preamble sequence are sent during time slots $\tau^{(2)}_i+1,\cdots, \tau^{(2)}_i+n'+n_i$ and we have 
\begin{equation}
\label{ }
b_{i,\tau^{(2)}_i+1}=b_{i,\tau^{(2)}_i}+b'_{i,\tau^{(2)}_i}-\lfloor n\eta_i\rfloor.
\end{equation}
 In general, $\tau_i^{(j)}$ is defined by
 \begin{equation}
\label{buffer1}
\tau_i^{(j)}=\min\left\{t\geq \tau_{i}^{(j-1)}+n'+n_i: b_{i,t}+b'_{i,t}\geq \lfloor n\eta_i\rfloor\right\}.
\end{equation}
At time slot $\tau^{(j)}_i$, a number of $\lfloor n\eta_i\rfloor$ bits in the buffer of Tx~$i$ are represented by a codeword which together with the preamble sequence are sent during time slots $\tau^{(j)}_i+1,\cdots, \tau^{(j)}_i+n'+n_i$. Moreover, 
\begin{equation}
\label{buffer2}
b_{i,\tau^{(j)}_i+1}=b_{i,\tau^{(j)}_i}+b'_{i,\tau^{(j)}_i}-\lfloor n\eta_i\rfloor.
\end{equation}
This style of transmission is \textit{Tx-Rx asynchronous} in the sense that Rx~$i$ does not know a priori the time slots $\tau_i^{(1)}+1, \tau_i^{(2)}+1,\cdots$ when Tx~$i$ begins to send its bursts.

A few remarks are in order:
\begin{enumerate}[(i)]
  \item  The Tx-Rx~asynchronous transmission considered in this paper is in contrast to the Tx-Rx~synchronous scheme\footnote{See chapter~24 on page~600.}  studied in \cite{gamkim} in the context of networking and information theory. In this scheme, Tx~$i$ sends its codewords only at time slots $mn_i+1$ where $m\geq 1$ is an integer.\footnote{The description provided here for the scheme in \cite{gamkim} is given in terms of the notations introduced in this paper. Moreover, the communication scenarios studied in \cite{gamkim} are the point~to~point channel and the multiple access channel.} The so-called augmented codebook of Tx~$i$ consists of $2^{\lfloor n\eta_i\rfloor}$ \textit{data codewords} of length $n_i$ and one additional codeword referred to as the \textit{null codeword} with the same length~$n_i$. At the ``end'' of time slot $mn_i$, if there are at least $\lfloor n\eta_i\rfloor$ bits in the buffer, a data codeword is transmitted over the channel during time slots $mn_i+1,\cdots,(m+1)n_i$. If the number of bits at the ``end'' of time slot $mn_i$ is less than $\lfloor n\eta_i\rfloor$, the null codeword is transmitted over the channel during time slots $mn_i+1,\cdots,(m+1)n_i$ and Tx~$i$ repeats this process at time slot $(m+1)n_i$. Transmission of the null codeword facilitates the synchronization between a receiver and its corresponding transmitter. Lemma~24.1 in~\cite{gamkim} guarantees that the buffer of Tx~$i$ is stable, i.e., $\mathrm{sup}_{t\geq 0}\mathbb{E}[\boldsymbol{b}_{i,t}]<\infty$, if~and~only~if 
  \begin{equation}
\label{buffstab}
\mu_i:=\frac{\eta_i}{\lambda_i}=\frac{1}{N_iq_i}>\theta_i.
\end{equation}
 In the scheme considered in this paper, stability of the buffers is not an issue because Tx~$i$ only transmits a finite number $k_in$ of bits and hence, the backlog~(buffer content) is bounded from above by $k_in$ at any time slot. However, we still impose the constraint in~(\ref{buffstab}) for $N_i>1$ because it guarantees immediate data transmission described in the next remark.  
  \item  It is desirable that the transmissions be \textit{immediate} in the following sense: 
  
  \textit{We say the transmissions of Tx~$i$ are immediate if Tx~$i$ sends a codeword immediately after there are at least $\lfloor n\eta_i\rfloor$ bits stored in its buffer.}
  
 Such immediate transmission is not possible if a previously scheduled codeword is not completely transmitted. More precisely, let 
  \begin{equation}
\label{xi1}
\widetilde{\tau}_i^{(1)}:=\min\left\{t\geq \tau^{(1)}_i+1: b_{i,t}+b'_{i,t}\geq \lfloor n\eta_i\rfloor\right\}.
\end{equation}
Then $\widetilde{\tau}^{(1)}_i$ is the earliest time slot such that the buffer of Tx~$i$ contains at least $\lfloor n\eta_i\rfloor$ bits after the transmission of the first burst begun at time slot $\tau^{(1)}_i+1$. If $\widetilde{\tau}^{(1)}_i\leq \tau^{(1)}_i+n'+n_i-1$, these $\lfloor n\eta_i\rfloor$ bits must stay in the buffer until time slot $\tau^{(1)}_i+n'+n_i$ when the transmission of the first scheduled codeword is complete. In~Appendix A it is shown that if (\ref{buffstab}) holds, then 
\begin{equation}
\label{buffer3}
\mathbb{P}\big(\boldsymbol{\widetilde{\tau}}^{(1)}_i\leq \boldsymbol{\tau}^{(1)}_i+n'+n_i-1\big)\leq e^{-c_in},
\end{equation}
 where $c_i>0$ is a constant that does not depend on $n$. In virtue of (\ref{buffer3}) and for sufficiently large $n$, the second transmission is immediate with arbitrarily large probability if $\mu_i>\theta_i$.  Next,  define 
 \begin{equation}
\label{xi2}
\widetilde{\tau}_i^{(2)}:=\min\left\{t\geq \tau^{(2)}_i+1: b_{i,t}+b'_{i,t}\geq \lfloor n\eta_i\rfloor\right\}.
\end{equation}
 Then $\widetilde{\tau}^{(2)}_i$ is the earliest time slot such that the buffer of Tx~$i$ contains at least $\lfloor n\eta_i\rfloor$ bits after the transmission of the second burst begun at time slot $\tau^{(2)}_i+1$. If $\widetilde{\tau}^{(2)}_i\leq \tau^{(2)}_i+n'+n_i-1$, then these $\lfloor n\eta_i\rfloor$ bits must stay in the buffer until time slot $\tau^{(2)}_i+n'+n_i$ when the transmission of the second scheduled codeword is complete. Similar to (\ref{buffer3}), 
\begin{equation}
\label{buffer4}
\mathbb{P}\big(\boldsymbol{\widetilde{\tau}}^{(2)}_i\leq \boldsymbol{\tau}^{(2)}_i+n'+n_i-1\big|\,\boldsymbol{\widetilde{\tau}}_i^{(1)}\geq \boldsymbol{\tau}_i^{(1)}+n'+n_i \big)\leq e^{-c_i n}
\end{equation}
holds under the condition $\mu_i>\theta_i$. By (\ref{buffer3}) and (\ref{buffer4}), the probability that both the second and third transmissions are immediate is bounded from above by $2e^{-c_in}$. Simple induction shows that the probability of all $N_i$ transmissions by Tx~$i$ being immediate is bounded from above by $N_ie^{-c_in}$.
\begin{figure}[t]
  \centering
  \includegraphics[scale=.7] {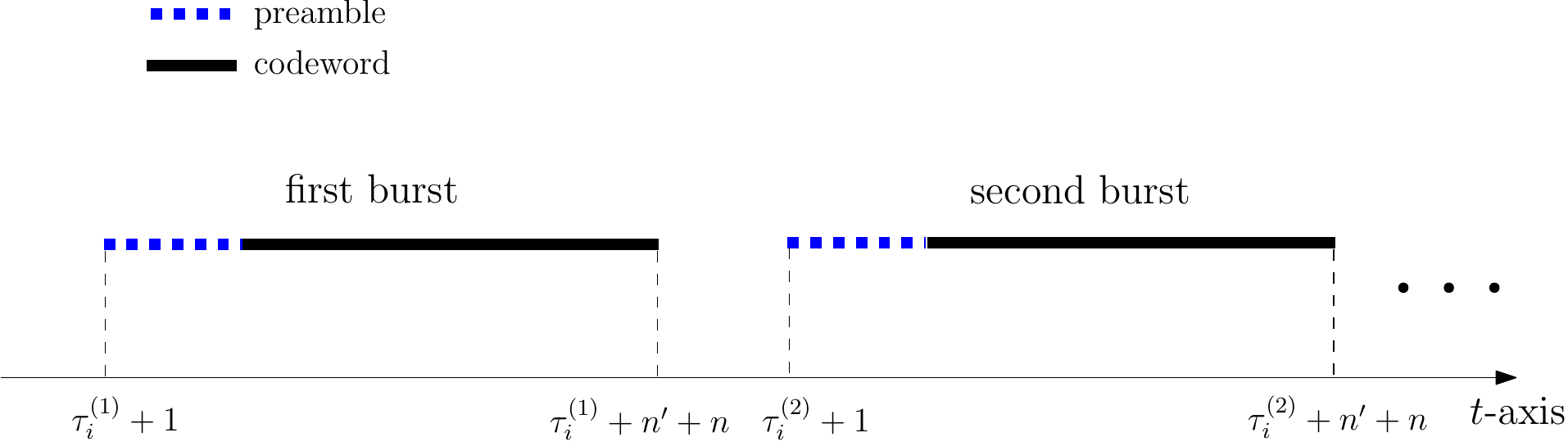}
  \caption{If $\mu_i>\theta_i$, the signals sent by Tx~$i$ look like intermittent bursts along the $t$-axis with high probability.  }
  \label{bufpic1}
 \end{figure} 
\item By the previous remark and under the constraint in (\ref{buffstab}), the signals sent by Tx~$i$ look like intermittent bursts along the $t$-axis with high probability as shown in Fig.~\ref{bufpic1}.  After sending a codeword, the transmitter must wait to receive enough bits in its buffer to transmit the next codeword. In contrast to~\cite{gamkim}, no ``null codeword'' is utilized in this paper, i.e., Tx~$i$ stays silent if it does not have enough bits in its buffer to represent a codeword. 
  \item   Since transmissions are immediate, $\boldsymbol{\tau}_{i}^{(j)}=\boldsymbol{\widetilde{\tau}}_i^{(j)}$ with high probability and one may redefine $\tau_{i}^{(j)}$ in~(\ref{buffer1})~by
 \begin{eqnarray}
 \label{buffer}
\tau_i^{(0)}:=0,\,\,\,\,\tau^{(j)}_i:=\min\left\{t\geq \tau^{(j-1)}_i+1: b_{i,t}+b'_{i,t}\ge\lfloor n\eta_i\rfloor\right\},\,\,\,\,j\ge1.
\end{eqnarray}
  Without loss of generality, let $n$ be a multiple of $N_1N_2$. Then $\lfloor n\eta_i\rfloor $ becomes divisible by $k_i$ for both $i=1,2$ and the inequality $ b_{i,t}+b'_{i,t}\ge \lfloor n\eta_i\rfloor$ in~(\ref{buffer}) can be replaced by $b_{i,t}+b'_{i,t}= \lfloor n\eta_i\rfloor$, i.e., 
    \begin{eqnarray}
 \label{buffer111}
\tau_i^{(0)}:=0,\,\,\,\,\tau^{(j)}_i:=\min\left\{t\geq \tau^{(j-1)}_i+1:  b_{i,t}+b'_{i,t}=\lfloor n\eta_i\rfloor\right\},\,\,\,\,j\ge1.
\end{eqnarray}
Moreover, the buffer dynamics can be written as 
\begin{equation}
\label{buffer321}
b_{i,t+1}=\left\{\begin{array}{cc}
      b_{i,t}+b'_{i,t} & \textrm{$b_{i,t}+b'_{i,t}<\lfloor n\eta_i\rfloor$} \\
   0   &   b_{i,t}+b'_{i,t}= \lfloor n\eta_i\rfloor 
\end{array}\right..
\end{equation} 
\end{enumerate}
The following proposition compares the times when Tx~$i$ begins to send its $j^{th}$ burst under the immediate Tx-Rx~asynchronous scheme considered in this paper and the Tx-Rx~synchronous scheme in~\cite{gamkim}:
\begin{proposition}
\label{prop_opt11}
Assume $\mu_i$ is not an integer multiple of $\theta_i,\frac{\theta_i}{2},\cdots,\frac{\theta_i}{N_i}$. Let $\boldsymbol{\varsigma}_i^{(j)}+1$ for $1\leq j\leq N_i$ be the time slot that Tx~$i$ begins to send its $j^{th}$ codeword under the Tx-Rx~synchronous scheme in~\cite{gamkim}. There exists $\delta>0$ such that 
\begin{eqnarray}
\label{sugar_11}
\lim_{n\to\infty}\mathbb{P}\big(\boldsymbol{\varsigma}_i^{(j)}>(1+\delta)\boldsymbol{\tau}_i^{(j)}\big)=1,
\end{eqnarray}
for any $1\leq j\leq N_i$.
\end{proposition}
\begin{proof}
See Appendix~B.
\end{proof} 
\textbf{Remark}- Under the assumptions in Proposition~\ref{prop_opt11}, one can prove the existence of $\delta>0$ such that $\lim_{n\to\infty}\frac{\boldsymbol{\varsigma}_i^{(j)}-\boldsymbol{\tau}_i^{(j)}}{\boldsymbol{\tau}_i^{(j)}}>\delta$ for any $1\leq j\leq N_i$ which is stronger than the statement in (\ref{sugar_11}).
\subsection{Tx-Tx asynchrony}
In the previous section the incoming bit streams at the transmitters where assumed to be synchronous in the sense that both start to run at time slot $t=1$. In practice, the activation times for these processes are different. Let Tx~$1$ and Tx~$2$ start their activity at time slots $t=\lfloor n\nu_1\rfloor$ and $t=\lfloor n\nu_2\rfloor$, respectively, where $\nu_1,\nu_2> 0$. Then (\ref{buffer111}) is rewritten as 
\begin{equation}
\label{buffer222}
\tau_i^{(0)}:=\lfloor n\nu_i\rfloor,\,\,\,\,\tau^{(j)}_i:=\min\left\{t\geq \tau^{(j-1)}_i+1:  b_{i,t}+b'_{i,t}=\lfloor n\eta_i\rfloor\right\},\,\,\,\,j\geq1.
\end{equation}
We see that $\tau^{(j)}_i$ is the smallest $t\geq \tau^{(j-1)}_i+1$ such that  Tx~$i$ receives packets of $k_i$ bits in exactly $\frac{\lfloor n\eta_i\rfloor}{k_i}$ slots among the time slots with indices $\tau^{(j-1)}_i+1,\cdots, t$. Therefore, $\boldsymbol{\tau}_i^{(j)}-\boldsymbol{\tau}_i^{(j-1)}$ is a negative binomial random variable with parameters $\frac{\lfloor n\eta_i\rfloor}{k_i}$ and $q_i$, i.e., 
\begin{equation}
\label{nb_hello}
\boldsymbol{\tau}_i^{(j)}-\boldsymbol{\tau}_i^{(j-1)}\sim\mathrm{NB}\Big(\frac{\lfloor n\eta_i\rfloor}{k_i},q_i\Big),\,\,\,\,j\ge 1.
\end{equation}
Alternatively, if $\boldsymbol{\xi}_{i,1},\cdots, \boldsymbol{\xi}_{i,j}$ is a sequence of independent $\mathrm{NB}(\frac{\lfloor n\eta_i\rfloor}{k},q)$ random variables, one can write
\begin{equation}
\label{}
\boldsymbol{\tau}_i^{(j)}=\boldsymbol{\xi}_{i,1}+\cdots+\boldsymbol{\xi}_{i,j}+\lfloor n\nu_i\rfloor-1,\,\,\,\,j\geq 1.
\end{equation}
Defining 
\begin{equation}
\label{ }
\boldsymbol{\xi}_i^{(j)}:=\boldsymbol{\xi}_{i,1}+\cdots+\boldsymbol{\xi}_{i,j},
\end{equation}
then $\boldsymbol{\xi}_i^{(j)}\sim\mathrm{NB}(\frac{j\lfloor n\eta_i\rfloor}{k_i},q_i)$ and we get our final expression for $\boldsymbol{\tau}_i^{(j)}$, i.e.,  
\begin{equation}
\label{laws11}
\boldsymbol{\tau}_i^{(j)}=\boldsymbol{\xi}_i^{(j)}+\lfloor n\nu_i\rfloor-1.
\end{equation}
We end this subsection with the following remarks:

\textbf{Remark}- Throughout the paper, $\nu_1$ and $\nu_2$ are realizations of independent and continuous random variables $\boldsymbol{\nu}_1$ and $\boldsymbol{\nu}_2$.

\textbf{Remark}- In Remark~(iv) in the previous subsection we assumed that $n$ is a multiple of $N_1N_2$. If we do not make such an assumption, each $\boldsymbol{\xi}_{i,l}$ for $1\leq l\leq j$ turns out to be a $\mathrm{NB}\big(\lfloor\frac{\lfloor n\eta_i\rfloor}{k_i}\rfloor+m_l,q_i\big)$ random variable where $m_l$ is an integer that depends on $n, k_i$ and $N_i$ and $0\leq |m_l|<k_i$. This does not affect the results in the forthcoming sections. The assumption that $n$ is a multiple of $N_1N_2$ is made only for notational simplicity. 
\subsection{The Average transmission power and the average transmission rate}
 The incoming bit stream at the buffer of Tx~$i$ starts at time slot $\lfloor n\nu_i\rfloor$ and Tx~$i$ sends the last symbol in its $N_i^{th}$ burst (last burst) at time slot $\tau_i^{(N_i)}+n'+n_i$. Therefore, the activity period $\mathcal{T}_i$ appearing in~(\ref{power_11}) is given by
 \begin{eqnarray}
 \label{period_11}
\mathcal{T}_i=\Big\{\lfloor n\nu_i\rfloor, \lfloor n\nu_i\rfloor+1,\cdots,\tau_i^{(N_i)}+n'+n_i\Big\}.
\end{eqnarray}
   The elements of each codeword  and the preamble sequence for Tx~$i$ are realizations of independent $\mathrm{N}(0,\gamma_i)$ random variables where $\gamma_i>0$ is a constant designed to ensure that the average transmission power at Tx~$i$ does not exceed $P_i$. In the following, we compute the average transmission power~$Q_i$ and the average transmission rate~$R_i$ for Tx~$i$:
\subsubsection{Average transmission power} Tx~$i$ sends out $N_i$ bursts where the $j^{th}$ burst lasts from time slot $\boldsymbol{\tau}_{i}^{(j)}+1$ to time slot $\boldsymbol{\tau}_i^{(j)}+n'+n_i$. The average transmission power $\boldsymbol{Q}_{i}$ is a random variable given by 
\begin{eqnarray}
\label{pow11}
\boldsymbol{Q}_{i}=\frac{1}{|\boldsymbol{\mathcal{T}}_i|}\sum_{t=\lfloor n\nu_i\rfloor}^{\boldsymbol{\tau}_i^{(N_i)}+n'+n_i}\boldsymbol{x}_{i,t}^2=\frac{n'+n_i}{|\boldsymbol{\mathcal{T}}_i|}\sum_{j=1}^{N_i}\frac{1}{n'+n_i}\sum_{t=\boldsymbol{\tau}_i^{(j)}+1}^{\boldsymbol{\tau}_i^{(j)}+n'+n_i}\boldsymbol{x}_{i,t}^2.
\end{eqnarray}
By SLLN,
\begin{eqnarray}
\label{pow22}
\lim_{n\to\infty}\frac{1}{n'+n_i}\sum_{l=\boldsymbol{\tau}_i^{(j)}+1}^{\boldsymbol{\tau}_i^{(j)}+n'+n_i}\boldsymbol{x}_{i,t}^2=\gamma_i,
\end{eqnarray}
for any $1\leq j\leq N_i$. Recalling the expression for $\boldsymbol{\tau}_i^{(j)}$ in (\ref{laws11}), 
\begin{eqnarray}
\label{pow33}
\frac{n'+n_i}{|\boldsymbol{\mathcal{T}}_i|}&=&\frac{n'+n_i}{\boldsymbol{\tau}_i^{(N_i)}+n'+n_i-\lfloor n\nu_i\rfloor+1}\notag\\&=&\frac{n'+n_i}{\boldsymbol{\xi}_i^{(N_i)}+n'+n_i}\notag\\
&=&\frac{n'+n_i}{\frac{N_i\lfloor n\eta_i\rfloor}{k_i}\frac{\boldsymbol{\xi}_i^{(N_i)}}{\frac{N_i\lfloor n\eta_i\rfloor}{k_i}}+n'+n_i}.
\end{eqnarray}
Since $\boldsymbol{\xi}_i^{(N_i)}\sim \mathrm{NB}(\frac{N_i\lfloor n\eta_i\rfloor}{k_i},q_i)$ is the sum of $\frac{N_i\lfloor n\eta_i\rfloor}{k_i}$ independent geometric random variables with parameter $q_i$,  we invoke SLLN one more time to write
\begin{eqnarray}
\label{pow44}
\lim_{n\to\infty}\frac{\boldsymbol{\xi}_i^{(N_i)}}{\frac{N_i\lfloor n\eta_i\rfloor}{k_i}}=\frac{1}{q_i}.
\end{eqnarray}
By (\ref{pow11}), (\ref{pow22}), (\ref{pow33}) and (\ref{pow44}), the average transmission power in the asymptote of large $n$ is given by
\begin{eqnarray}
\label{pow55}
\lim_{n\to\infty}\boldsymbol{Q}_i=\lim_{n\to\infty}\frac{(n'+n_i)N_i\gamma_i}{\frac{N_i\lfloor n\eta_i\rfloor}{k_i}\frac{1}{q_i}+n'+n_i}=\frac{\theta_iN_i\gamma_i}{\frac{N_i\eta_i}{k_iq_i}+\theta_i}=\frac{N_i\gamma_i}{1+\frac{1}{q_i\theta_i}},
\end{eqnarray}
where we replaced $\eta_i=\frac{k_i}{N_i}$ in the last step.
\subsubsection{Average transmission rate}
Tx~$i$ sends a total number of $k_in$ bits over its whole period of activity $\mathcal{T}_i$. Then the average transmission rate $\boldsymbol{R}_i$ is a random variable given by 
\begin{eqnarray}
\label{rate55}
\boldsymbol{R}_i=\frac{k_in}{|\boldsymbol{\mathcal{T}}_i|}=\frac{k_in}{\boldsymbol{\xi}_i^{(N_i)}+n'+n_i}.
\end{eqnarray}
Using (\ref{pow44}), the average transmission rate in the asymptote of large $n$ is
\begin{equation}
\label{ }
\lim_{n\to\infty}\boldsymbol{R}_i=\frac{k_i}{\frac{N_i\eta_i}{\lambda_i}+\theta_i}=\frac{\lambda_i}{1+q_i\theta_i}.
\end{equation}
We will use the expressions in (\ref{pow55}) and (\ref{rate55}) in Section~V.A where we study system design. 
\section{Estimating the arrival times and transmitter identification at the receivers} 
 Let $(s^{(j)}_{i,l})_{l=0}^{n_i-1}$ for $1\leq j\leq N_i$ be the $N_i$ codewords of length $n_i$ sent by Tx~$i$. Also, let $(s'_{i,l})_{l=0}^{n'-1}$  be the preamble sequence for user~$i$.  The signal $x_{i,t}$ in (\ref{system}) can be written as
\begin{eqnarray}
x_{i,t}=\left\{\begin{array}{cc}
         s'_{i,t-\tau_{i}^{(j)}-1}&  \tau_{i}^{(j)}+1\leq t\leq \tau_i^{(j)}+n'   \\
      s^{(j)}_{i,t-\tau_{i}^{(j)}-n'-1} & \tau_{i}^{(j)}+n'+1\leq t\leq \tau_i^{(j)}+n'+n_i   \\
         0& \textrm{otherwise}
\end{array}\right.,
\end{eqnarray}  
for $1\leq i\leq 2$ and $1\leq j\leq N_i$. The preambles $(s'_{1,l})_{l=0}^{n'-1}$  and $(s'_{2,l})_{l=0}^{n'-1}$ are revealed to both receivers. The following assumption considerably simplifies the analysis in this section.\\
\textbf{Assumption}- For any integers $1\leq j_1\leq N_1$ and $1\leq j_2\leq N_2$,
\begin{equation}
\label{res_11}
j_2\mu_2-j_1\mu_1+\nu_2-\nu_1\notin\big\{0,\theta_1,-\theta_2,\theta_1-\theta_2\big\}.
\end{equation}
\textbf{Remark}- Since we are assuming that  $\nu_1$ and $\nu_2$ are realizations of independent and continuous random variables $\boldsymbol{\nu}_1$ and $\boldsymbol{\nu}_2$, the restrictions in (\ref{res_11}) are considered ``mild'' in the sense that the probability of $j_2\mu_2-j_1\mu_1+\boldsymbol{\nu}_2-\boldsymbol{\nu}_1$ lying in $\{0,\theta_1,-\theta_2,\theta_1-\theta_2\}$ for some $j_1$ and $j_2$ is equal to zero.  

We will use the assumption in (\ref{res_11}) throughout the paper. Its first application appears in the following proposition: 
\begin{proposition}
\label{prop1}
Assuming~(\ref{res_11}) holds, the probability of Tx~$i$ starting or ending a transmission burst while Tx~$i'$ is sending a preamble sequence tends to zero as $n$ grows. 
\end{proposition}
\begin{proof}
See Appendix~C.
\end{proof}
In view of Proposition~\ref{prop1} and for given $\epsilon>0$, we assume $n$ is large enough so that the probability of Tx~$i$ starting or ending a transmission burst while Tx~$i'$ is sending a preamble sequence is less than $\epsilon$ and add $\epsilon$ to the probability of error in decoding the codewords. In other words, we assume no transmitter starts or ends a transmission burst while the other transmitter is sending a preamble sequence.

Next, we study the detection/estimation procedure at Rx~1. A similar procedure is carried out at Rx~2. Define the PDFs $p^{(1)}(\cdot, \cdot), \cdots, p^{(4)}(\cdot, \cdot)$ on $\mathbb{R}^2$ as follows:
\begin{itemize}
  \item For $1\leq j\leq N_1$ and $\tau_1^{(j)}+1\leq t\leq \tau_1^{(j)}+n'$, Tx~$1$ is sending the preamble sequence in its $j^{th}$ burst. If Tx~$2$ is not transmitting during this time interval, then $p_{\boldsymbol{x}_{1,t},\boldsymbol{y}_{1,t}}(x,y)=\mathrm{g}(x;\gamma_1)\mathrm{g}(y-x;1)$. We define 
  \begin{equation}
\label{p1}
p^{(1)}(x,y):=\mathrm{g}(x;\gamma_1)\mathrm{g}(y-x;1).
\end{equation}
  \item For $1\leq j\leq N_1$ and $\tau_1^{(j)}+1\leq t\leq \tau_1^{(j)}+n'$, Tx~$1$ is sending the preamble  sequence in its $j^{th}$ burst. If Tx~$2$ is transmitting during this time interval, then $p_{\boldsymbol{x}_{1,t},\boldsymbol{y}_{1,t}}(x,y)=\mathrm{g}(x;\gamma_1)\mathrm{g}(y-x;1+a_2\gamma_2)$. We define 
  \begin{equation}
\label{p2}
p^{(2)}(x,y):=\mathrm{g}(x;\gamma_1)\mathrm{g}(y-x;1+a_2\gamma_2).
\end{equation}
 \item For $1\leq j\leq N_2$ and $\tau_2^{(j)}+1\leq t\leq \tau_2^{(j)}+n'$, Tx~$2$ is sending the preamble sequence in its $j^{th}$ burst. If Tx~$1$ is not transmitting during this time interval, then $p_{\boldsymbol{x}_{2,t},\boldsymbol{y}_{1,t}}(x,y)=\mathrm{g}(x;\gamma_2)\mathrm{g}(y-a_2x;1)$. We define 
  \begin{equation}
\label{p3}
p^{(3)}(x,y):=\mathrm{g}(x;\gamma_2)\mathrm{g}(y-a_2x;1).
\end{equation}
\item For $1\leq j\leq N_2$ and $\tau_2^{(j)}+1\leq t\leq \tau_2^{(j)}+n'$, Tx~$2$ is sending the preamble sequence in its $j^{th}$ burst. If Tx~$1$ is transmitting during this time interval, then $p_{\boldsymbol{x}_{2,t},\boldsymbol{y}_{1,t}}(x,y)=\mathrm{g}(x;\gamma_2)\mathrm{g}(y-a_2x;1+\gamma_1)$. We define 
  \begin{equation}
\label{p4}
p^{(4)}(x,y):=\mathrm{g}(x;\gamma_2)\mathrm{g}(y-a_2x;1+\gamma_1).
\end{equation}
\end{itemize}
To identify the arrival time of a transmission burst, each receiver applies the so-called sequential joint typicality decoder~\cite{Chandar11}. Recall that for given $\epsilon>0$, $m\geq 1$ and a PDF $p(\cdot,\cdot)$ on $\mathbb{R}^2$ with marginals $p_{1}$ and $p_{2}$, the typical set $A_\epsilon^{(m)}[p]$ is the set of all pairs $(\vec{x},\vec{y}\,)$ where $\vec{x},\vec{y}\in \mathbb{R}^m$ and the three inequalities
    \begin{eqnarray}
    \label{smart0}
    \bigg| \frac{1}{m}\sum_{l=1}^m \log p_{1}(x_l)+h(p_{1})\bigg|<\epsilon,  
\end{eqnarray}
\begin{eqnarray}
\label{smart00}
 \bigg| \frac{1}{m}\sum_{l=1}^m \log p_{2}(y_l)+h(p_{2})\bigg|<\epsilon
\end{eqnarray}
and
  \begin{eqnarray}
  \label{smart1}
      \bigg| \frac{1}{m}\sum_{l=1}^m \log p(x_l,y_l)+h(p)\bigg|<\epsilon
\end{eqnarray}
hold. 
We refer to any $(\vec{x},\vec{y}\,)\in A_\epsilon^{(m)}[p]$ as an $\epsilon$-jointly typical pair with respect to $p$ \cite{16}.
\begin{figure}[t]
  \centering
  \includegraphics[scale=.7] {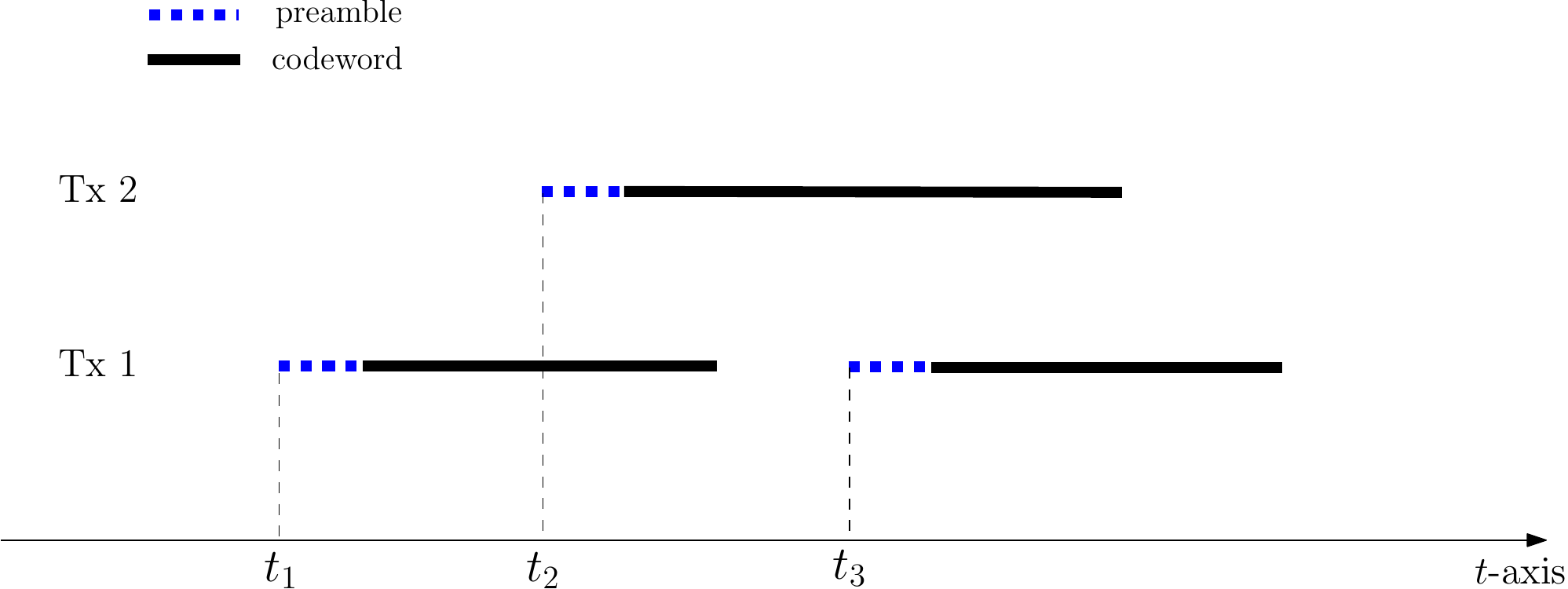}
  \caption{A scenario where where Tx~1 and Tx~2 send only two and one bursts, respectively, i.e., $N_1=2$ and $N_2=1$. It is assumed that $n_1<n_2$. For simplicity of presentation, we call $t_1:=\tau_1^{(1)}+1$, $t_2:=\tau_2^{(1)}+1$ and $t_3:=\tau_1^{(2)}+1$. }
  \label{bufpic3}
 \end{figure} 
 To describe how Rx~$1$ estimates $\tau_i^{(j)}$ for different $i$ and $j$ and without loss of generality, we find it best to consider the particular situation shown in Fig.~\ref{bufpic3} where $N_1=2$, $N_2=1$ and $n_1<n_2$. The arrival time estimation and user identification are performed in the following steps:
\begin{enumerate}[(i)]
  \item By Fig.~\ref{bufpic3}, $t_1:=\tau_{1}^{(1)}+1$ is the time slot that the first active transmitter sends the first symbol in its preamble sequence. Rx~$1$ estimates $t_1$ by 
\begin{equation}
\label{est_t1}
\hat{t}_1=\min\left\{t\geq 0: \textrm{$\big((s'_{1,l})_{l=0}^{n'-1},(y_{1,l})_{l=t}^{t+n'-1}\big)\in A^{(n')}_{\epsilon}[p^{(1)}]$ or $\big((s'_{2,l})_{l=0}^{n'-1},(y_{1,l})_{l=t}^{t+n'-1}\big)\in A^{(n')}_{\epsilon}[p^{(3)}]$} \right\},
\end{equation}
where the PDFs $p^{(1)}$ and $p^{(3)}$ are defined in (\ref{p1}) and (\ref{p3}), respectively. By Proposition~\ref{prop1}, we can assume $p_{\boldsymbol{x}_{1,l},\boldsymbol{y}_{1,l}}=p^{(1)}$ for any $t_1\leq l\leq t_1+n'-1$. Then the weak law of large numbers yields 
\begin{equation}
\label{wl1}
\lim_{n\to\infty}\mathbb{P}\Big(\big((\boldsymbol{s}'_{1,l})_{l=0}^{n'-1},(\boldsymbol{y}_{1,l})_{l=\boldsymbol{t}_1}^{\boldsymbol{t}_1+n'-1}\big)\in A^{(n')}_{\epsilon}[p^{(1)}]\Big)=1.
\end{equation}
By (\ref{wl1}), $\lim_{n\to\infty}\mathbb{P}(\,\hat{\boldsymbol{t}}_1\leq \boldsymbol{t}_1)=1$. As such, to show that $\hat{\boldsymbol{t}}_1=\boldsymbol{t}_1$ holds with high probability, it is enough to show that $\mathbb{P}(\,\hat{\boldsymbol{t}}_1<\boldsymbol{t}_1)$ is negligible for sufficiently large $n$. This is the content of the following proposition: 
\begin{proposition}
\label{prop_2}
We have 
\begin{eqnarray}
\mathbb{P}(\,\hat{\boldsymbol{t}}_1<\boldsymbol{t}_1)\leq \Theta(n)e^{-\Theta(n')}.
\end{eqnarray}
In particular, $\lim_{n\to\infty}\mathbb{P}(\,\hat{\boldsymbol{t}}_1<\boldsymbol{t}_1)=0$.
\end{proposition}
\begin{proof}
See Appendix~D.
\end{proof}
Motivated by Proposition~\ref{prop_2}, we assume Rx~$1$ knows the exact value of $t_1$. In fact, for given $\epsilon>0$, we assume $n$ is large enough so that the probability of error in estimating the first arrival time is less than $\epsilon$ and add $\epsilon$ to the probability of error in decoding the codewords. Not only does Rx~1 know the exact value of $t_1$, but also it realizes that $t_1$ is $\tau_1^{(1)}+1$ and not $\tau_2^{(1)}+1$, i.e., it knows the first arriving burst belongs to~Tx~1. This is described in the next~step. 
 \item After finding $t_1$, Rx~$1$ decides whether the first burst belongs to Tx~$1$ or Tx~$2$. Towards this goal, Rx~$1$ verifies if
  \begin{eqnarray}
  \label{wlwl1}
\big((s'_{1,l})_{l=0}^{n'-1},(y_{1,l})_{l=t_1}^{t_1+n'-1}\big)\in A^{(n')}_{\epsilon}[p^{(1)}]
\end{eqnarray}
or 
 \begin{eqnarray}
 \label{wlwl2}
\big((s'_{2,l})_{l=0}^{n'-1},(y_{1,l})_{l=t_1}^{t_1+n'-1}\big)\in A^{(n')}_{\epsilon}[p^{(3)}].
\end{eqnarray}
 If (\ref{wlwl1}) holds, the first arriving burst is assumed to belong to Tx~1. If (\ref{wlwl2}) holds, the first arriving burst is assumed to belong to Tx~2. As mentioned earlier in (\ref{wl1}), (\ref{wlwl1}) holds with high probability in the asymptote of large $n$. In Appendix~E, it is shown that
\begin{eqnarray}
\label{canned_11}
\mathbb{P}\Big(\big((\boldsymbol{s}'_{2,l})_{l=0}^{n'-1},(\boldsymbol{y}_{1,l})_{l=t_1}^{t_1+n'-1}\big)\in A^{(n')}_{\epsilon}[p^{(3)}]\Big)\leq e^{-\Theta(n')}.
\end{eqnarray}
 Therefore, (\ref{wlwl2}) holds with a probability that decays exponentially with $n'$ and hence, Rx~$1$ can identify the sender of the first burst with high probability. 
  \item Up to this point, Rx~$1$ knows that the first burst belongs to Tx~$1$ and it lasts from time slot $t_1$ to time slot $t_1+n'+n_1-1$. If another burst arrives during this period, it must belong to Tx~$2$. As shown in Fig.~\ref{bufpic3}, a burst belonging to Tx~$2$ indeed arrives at time slot $t_2:=\tau_2^{(1)}+1$ when the first burst of Tx~$1$ is still arriving. The preamble sequence in the first burst by Tx~$2$ extends from time slot $t_2$ to time slot $t_2+n'-1$. By Proposition~\ref{prop1}, $p_{\boldsymbol{x}_{2,l},\boldsymbol{y}_{1,l}}=p^{(4)}$ for any $t_2\leq l\leq t_2+n'-1$ where $p^{(4)}$ is defined in (\ref{p4}). Based on these observations, Rx~$1$ estimates $t_2$ by 
  \begin{eqnarray}
  \label{est_t2}
\hat{t}_2=\min\left\{t_1\leq t\leq t_1+n'+n_1-1: \textrm{$\big((s'_{2,l})_{l=0}^{n'-1},(y_{1,l})_{l=t}^{t+n'-1}\big)\in A^{(n')}_{\epsilon}[p^{(4)}]$} \right\}.
\end{eqnarray}
Following similar lines of reasoning in the proof of Proposition~\ref{prop_2}, one can show that $\mathbb{P}(\,\hat{\boldsymbol{t}}_2\neq \boldsymbol{t}_2)\leq \Theta(n)e^{-\Theta(n')}$. As such, we can assume that Rx~$1$ knows the exact value of $t_2$, i.e., Rx~1 knows $\tau_2^{(1)}$. 

\textbf{Remark}- If (\ref{est_t2}) fails to return an estimate for $t_2$, Rx~$1$ concludes that no burst of Tx~$2$ is received by the time Tx~$1$ finishes its first burst. As such, starting at time slot $t_1+n'+n_1$, Rx~$1$ looks for the arrival time $t^*$ of a new transmission burst that might belong to Tx~$1$ or Tx~$2$. The time slot $t^*$ is estimated similar to (\ref{est_t1}), i.e., $\hat{t}^*$ is the smallest value of $t\geq t_1+n'+n_1$ such that $\big((s'_{1,l})_{l=0}^{n'-1},(y_{1,l})_{l=t}^{t+n'-1}\big)\in A^{(n')}_{\epsilon}[p^{(1)}]$ or $\big((s'_{2,l})_{l=0}^{n'-1},(y_{1,l})_{l=t}^{t+n'-1}\big)\in A^{(n')}_{\epsilon}[p^{(3)}]$.

\item After finding $t_2$ in step~(iii), Rx~$1$ knows that the first burst of Tx~$2$ lasts from time~slot $t_2$ to time~slot $t_2+n'+n_2-1$. Since the first burst of Tx~$1$ ends at time slot $t_1+n+n_1-1$, Rx~$1$ looks for possible arrival of the second burst of Rx~$1$ during time slots $t_1+n'+n_1$ to $t_2+n'+n_2-1$. In fact, as shown in Fig.~\ref{bufpic3}, the second burst of Tx~$1$ arrives at time slot $t_3:=\tau^{(2)}_1+1$ when the first burst by Tx~$2$ is still arriving. The preamble sequence in the second burst of Tx~$1$ extends from time slot $t_3$ to $t_3+n'-1$. By Proposition~\ref{prop1}, $p_{\boldsymbol{x}_{1,l},\boldsymbol{y}_{1,l}}=p^{(2)}$ for any $t_3\leq l\leq t_3+n'-1$ where $p^{(2)}$ is defined in (\ref{p2}). Based on these observations, Rx~$1$ estimates $t_3$ by 
  \begin{eqnarray}
  \label{est_t2}
\hat{t}_3=\min\left\{t_1+n'+n_1\leq t\leq t_2+n'+n_2-1: \textrm{$\big((s'_{1,l})_{l=0}^{n'-1},(y_{1,l})_{l=t}^{t+n'-1}\big)\in A^{(n')}_{\epsilon}[p^{(2)}]$} \right\},
\end{eqnarray}
where following the proof of Proposition~\ref{prop_2}, it can be shown that $\mathbb{P}(\,\hat{\boldsymbol{t}}_3\neq \boldsymbol{t}_3)\leq \Theta(n)e^{-\Theta(n')}$. 
\end{enumerate}
The detection/estimation procedure described here can be easily extended to scenarios other than the one depicted in Fig.~\ref{bufpic3}. Throughout the rest of the paper, we assume both receivers know the values of $\tau_i^{(j)}$ for any $i=1,2$ and $1\leq j\leq N_i$.   
\section{Decoding strategy and achievability results}
In the previous section, we described how each receiver is capable of estimating $\tau_i^{(j)}$ with vanishingly small probability of error. Our analysis heavily relied on the conditions in (\ref{res_11}) which are also used in the following proposition:
\begin{proposition}
\label{prop_3}
Let $i=1,2$ and $1\leq j_i\leq N_i$. Assuming (\ref{res_11}) holds, the $j_i^{th}$ codeword of Tx~$i$ and the $j_{i'}^{th}$ burst of Tx~$i'$ overlap with arbitrarily large probability in the asymptote of large $n$ if and only if
\begin{equation}
\label{cond1}
j_{i'}\mu_{i'}-j_i\mu_i+\nu_{i'}-\nu_i\in (0,\theta_i)\bigcup(-\theta_{i'},\theta_i-\theta_{i'}).
\end{equation} 
\end{proposition}
\begin{proof}
See Appendix~F.
\end{proof}
\begin{figure*}[t]
\centering
\subfigure[]{
\includegraphics[scale=0.5]{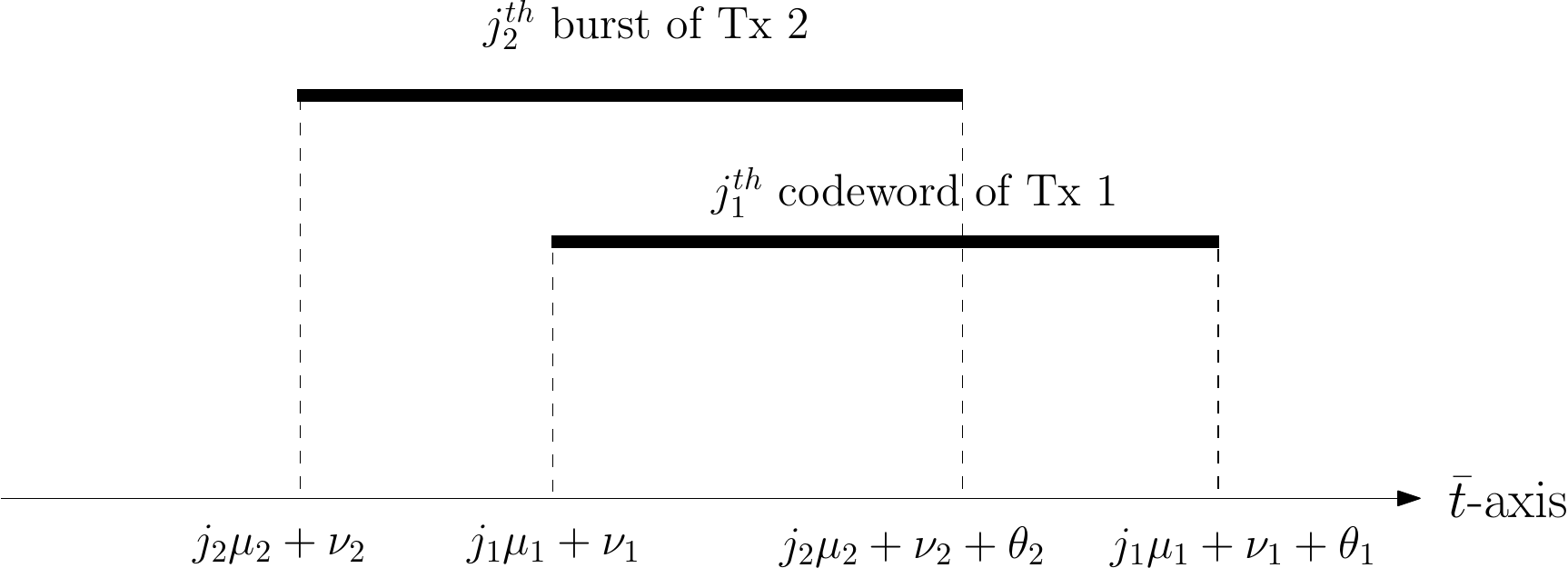}
\label{ref1}
}
\subfigure[]{
\includegraphics[scale=0.5]{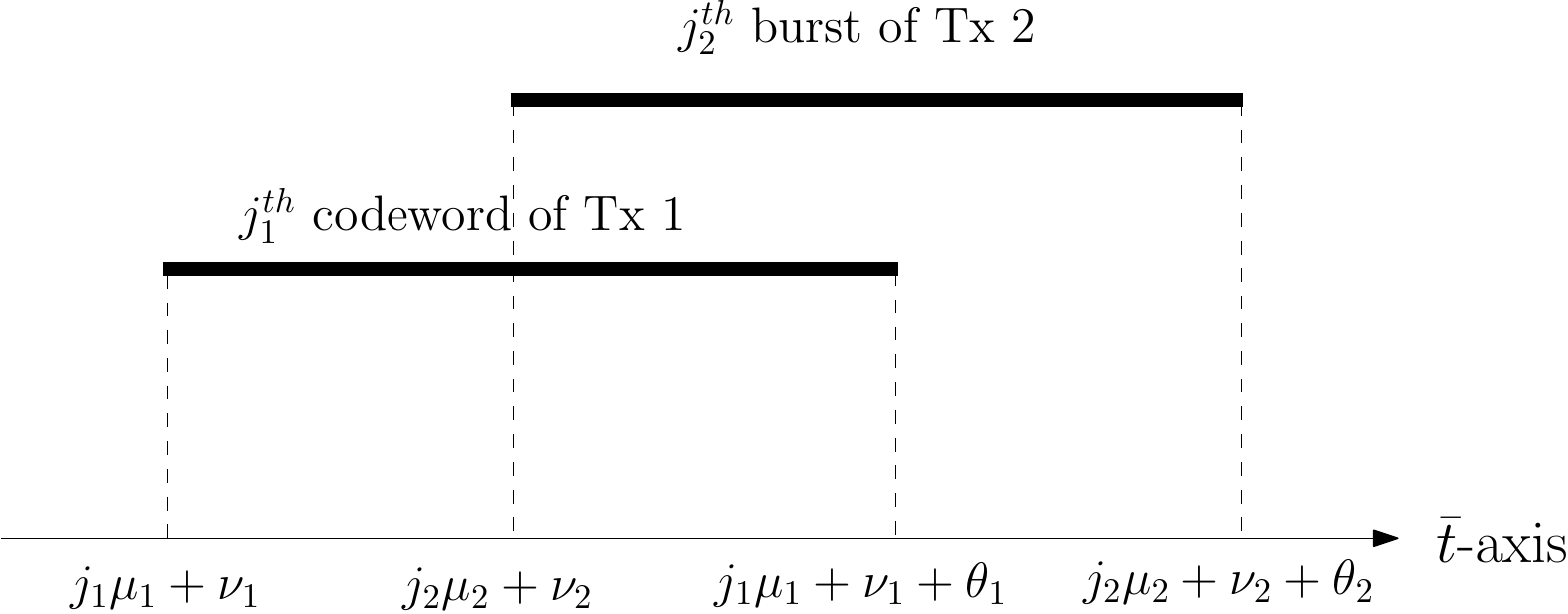}
\label{ref2}
}
\label{fig:subfigureExample}
\caption[Optional caption for list of figures]{If (\ref{pill1}) holds for $i=1$, the $j_2^{th}$ burst of Tx~$2$ overlaps (with high probability for large $n$) with the $j_1^{th}$ codeword of Tx~$1$ over its ``left~end'' as shown in panel~(a), while if (\ref{pill2}) holds for $i=1$, the $j_2^{th}$ burst of Tx~$2$ overlaps with the $j_1^{th}$ codeword of Tx~$1$ over its ``right~end'' as shown in panel~(b).}
\label{pic1}
\end{figure*}  
The result of Proposition~\ref{prop_3} can be intuitively described as follows. Define the scaled time variable 
\begin{equation}
\label{ }
\bar{t}:=\frac{t}{n}.
\end{equation}
 On the $\bar{t}$-axis, Tx~$i$ sends its $j_i^{th}$ codeword at times $\frac{1}{n}(\boldsymbol{\tau}_i^{(j_i)}+n'+1)$ to $\frac{1}{n}(\boldsymbol{\tau}_i^{(j_i)}+n'+n_i)$. In virtue of  SLLN, $\lim_{n\to\infty}\frac{1}{n}\boldsymbol{\tau}_i^{(j_i)}=j_i\mu_i+ \nu_i$. As such, in the limit as $n$ grows to infinity, the $j_i^{th}$ codeword of Tx~$i$ lies on the interval $\big(j_i\mu_i+\nu_i, j_i\mu_i+\nu_i+\theta_i\big)$ along the $\bar{t}$-axis. Similarly, one sees that the $j_{i'}^{th}$ burst of Tx~$i'$ lies on the interval $\big(j_{i'}\mu_{i'}+\nu_{i'}, j_{i'}\mu_{i'}+\nu_{i'}+\theta_{i'}\big)$ along the $\bar{t}$-axis. Provided~(\ref{res_11}) holds, the condition in (\ref{cond1}) is equivalent to saying that the two intervals $\big(j_i\mu_i+\nu_i, j_i\mu_i+\nu_i+\theta_i\big)$ and $\big(j_{i'}\mu_{i'}+\nu_{i'}, j_{i'}\mu_{i'}+\nu_{i'}+\theta_{i'}\big)$ overlap. More specifically, if 
\begin{equation}
\label{pill1}
-\theta_{i'}<j_{i'}\mu_{i'}-j_i\mu_i+\nu_{i'}-\nu_i<\min\{0,\theta_i-\theta_{i'}\},
\end{equation}
 the $j_{i'}^{th}$ burst of Tx~$i'$ overlaps (with high probability for large $n$) with the $j_i^{th}$ codeword of Tx~$i$ on its ``left~end'' as shown in Fig.~\ref{ref1} for $i=1$, while if 
 \begin{equation}
\label{pill2}
\max\{0,\theta_i-\theta_{i'}\}<j_{i'}\mu_{i'}-j_i\mu_i+\nu_{i'}-\nu_i<\theta_i,
\end{equation}
the $j_{i'}^{th}$ burst of Tx~$i'$ overlaps with the $j_i^{th}$ codeword of Tx~$i$ on its ``right~end'' as shown in Fig.~\ref{ref2} for $i=1$. Finally, 
\begin{equation}
\label{ }
\min\{0,\theta_i-\theta_{i'}\}<j_{i'}\mu_{i'}-j_i\mu_i+\nu_{i'}-\nu_i<\max\{0,\theta_i-\theta_{i'}\},
\end{equation}
implies that the $j_i^{th}$ codeword of Tx~$i$ is contained in the $j_{i'}^{th}$ burst of Tx~$i'$ or the other way around depending on whether $\theta_i<\theta_{i'}$ or $\theta_{i'}<\theta_i$, respectively. 

\textbf{Remark}- The geometric interpretation of the conditions in~(\ref{res_11}) is that the endpoints of the intervals $\big(j_i\mu_i+\nu_i, j_i\mu_i+\nu_i+\theta_i\big)$ and $\big(j_{i'}\mu_{i'}+\nu_{i'}, j_{i'}\mu_{i'}+\nu_{i'}+\theta_{i'}\big)$ do not coincide.

We make the following definitions: 
\begin{itemize}
  \item  Fixing $j_i=j$, there exists at most one positive integer $j_{i'}$ that satisfies (\ref{pill1}). We denote this value of $j_{i'}$ by $\omega_{i,j}^-$. In fact, $\omega_{i,j}^-$ is the index of the burst of Tx~$i'$ that overlaps with the left~end of the $j^{th}$ codeword of Tx~$i$. In case, $\omega_{i,j}^-$ does not exist, we write $\omega_{i,j}^-=0$.
  \item Fixing $j_i=j$, there exists at most one positive integer $j_{i'}$ that satisfies (\ref{pill2}). We denote this value of $j_{i'}$ by $\omega_{i,j}^{+}$. In fact, $\omega_{i,j}^+$ is the index of the burst of Tx~$i'$ that overlaps with the right~end of the $j^{th}$ codeword of Tx~$i$. In case, $\omega_{i,j}^+$ does not exist, we write $\omega_{i,j}^+=0$.
  \item We define $\omega_{i,j}$ as the number of bursts of Tx~$i'$ that are completely contained within the $j^{th}$ codeword of Tx~$i$.
\end{itemize}
  \begin{figure}[t]
  \centering
  \includegraphics[scale=.9] {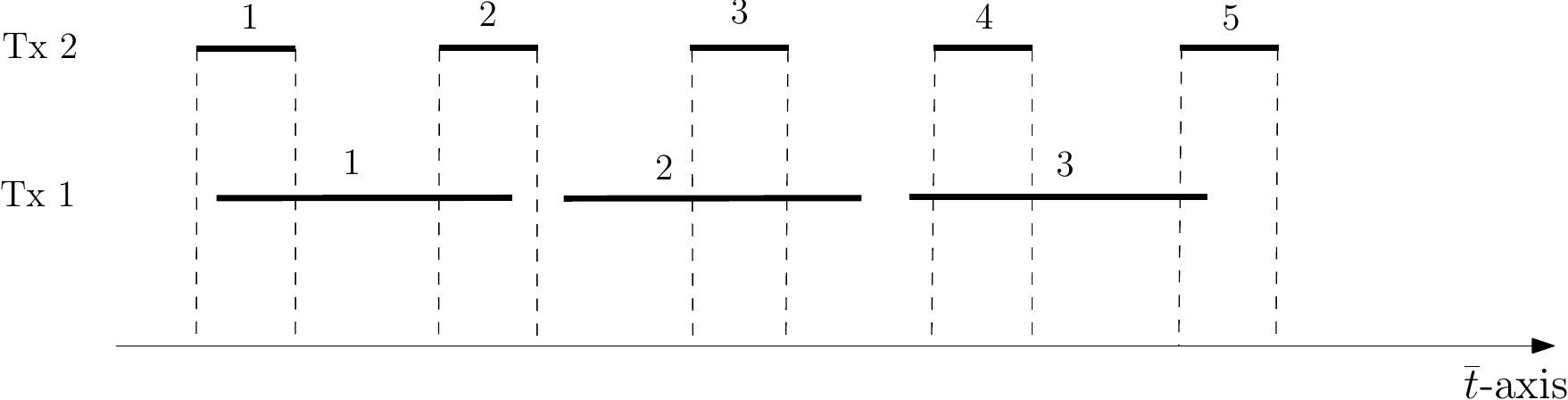}
  \caption{Positions of the bursts along the $\overline{t}$-axis in a scenario where $N_1=3$ and $N_2=5$.}
  \label{bufpic11}
 \end{figure} 
 For example, Fig.~\ref{bufpic11} presents a scenario where $N_1=3$ and $N_2=5$ and we have 
 \begin{eqnarray}
(\omega_{1,1}^-,\omega_{1,1}^+, \omega_{1,1})=(1,2,0),\,\,\,(\omega_{1,2}^-,\omega_{1,2}^+, \omega_{1,2})=(0,0,1),\,\,\,(\omega_{1,3}^-,\omega_{1,3}^+, \omega_{1,3})=(0,5,1).
\end{eqnarray}
Next, we study  achievability results for the $j^{th}$ transmitted codeword of Tx~$i$, i.e., we look for sufficient conditions that guarantee the $j^{th}$ transmitted codeword by Tx~$i$ is decoded successfully at Rx~$i$.  In order to describe the decoding strategy, we focus on Rx~1. For notational simplicity, in some equations we show $\omega_{1,j}^-$ and $\omega_{1,j}^+$ by $\omega^-$ and $\omega^+$,~respectively.  
 \begin{itemize}
  \item  Assume $\omega^-\neq 0$, $\omega^+\neq 0$ and $\omega_{1,j}=0$. Then  
 \begin{eqnarray}
\tau_{2}^{(\omega^{-})}+1<\tau_1^{(j)}+n'+1\leq \tau_{2}^{(\omega^{-})}+n'+n_{2}\leq \tau_{2}^{(\omega^{+})}+1\leq \tau_{1}^{(j)}+n'+n_1<\tau_{2}^{(\omega^+)}+n'+n_2.
\end{eqnarray}
This situation is shown in Fig.~\ref{bufpic7}. 
   \begin{figure}[t]
  \centering
  \includegraphics[scale=.7] {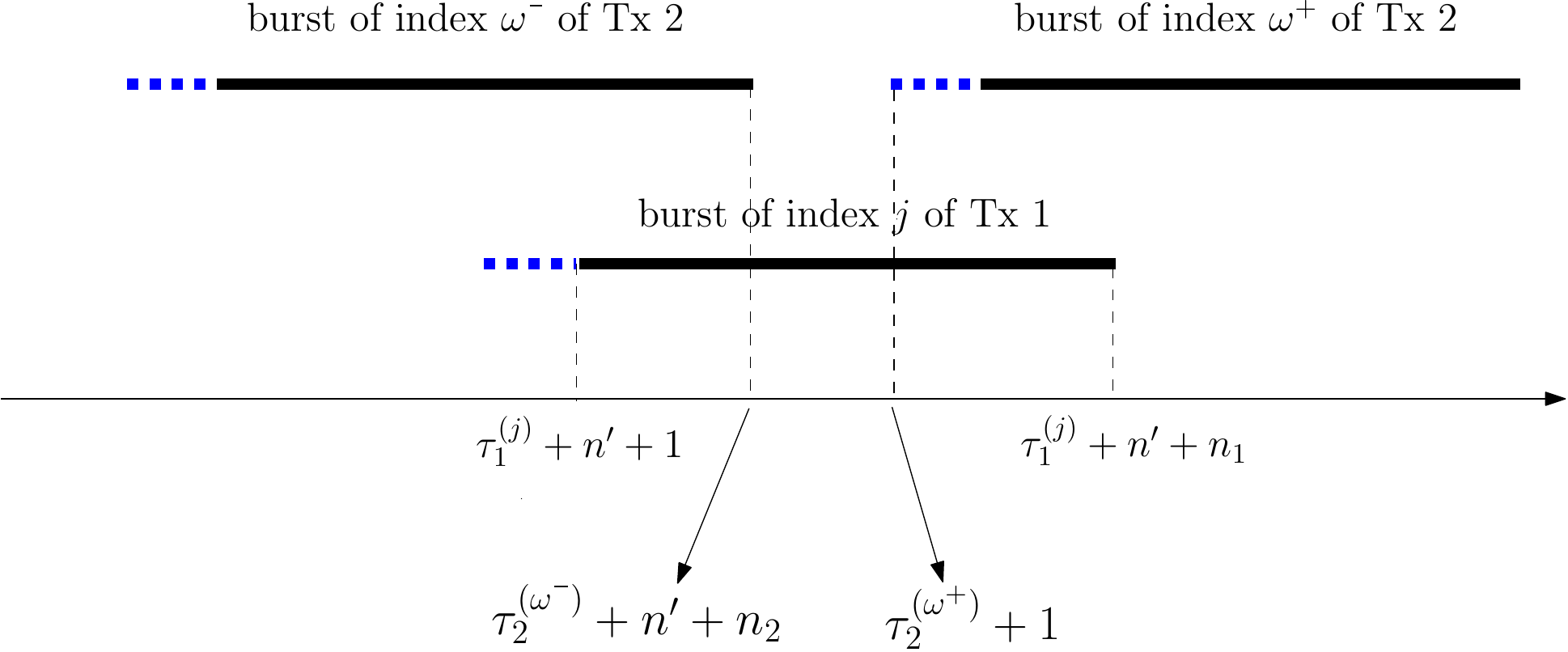}
  \caption{A scenario where $\omega_{1,j}^-\neq 0$, $\omega_{1,j}^+\neq 0$ and $\omega_{1,j}=0$. For notational simplicity, we have shown $\omega_{1,j}^-$ and $\omega_{1,j}^+$ by $\omega^-$ and $\omega^+$, respectively.}
  \label{bufpic7}
 \end{figure} 
The $j^{th}$ codeword of Tx~$1$ is transmitted during time slots $\tau_1^{(j)}+n'+1$ to $\tau_1^{(j)}+n'+n_1$. The interference pattern over this codeword is described as follows: 
\begin{itemize}
  \item Any symbol of the codeword transmitted during time slots $\tau_1^{(j)}+n'+1$ to $\tau_2^{(\omega^{-})}+n'+n_2$ is received in the presence of interference. For any $\tau_1^{(j)}+n'+1\leq l\leq \tau_2^{(\omega^{-})}+n'+n_2$, we have $p_{\boldsymbol{x}_{1,l},\boldsymbol{y}_{1,l}}=p^{(2)}$ where the PDF $p^{(2)}$ is defined in (\ref{p2}). 
  \item Any symbol of the codeword transmitted during time slots $\tau_2^{(\omega^{-})}+n'+n_2+1$ to $\tau_2^{(\omega^{+})}$ does not experience interference. For any $\tau_2^{(\omega^{-})}+n'+n_2+1\leq l\leq \tau_2^{(\omega^{+})}$, we have $p_{\boldsymbol{x}_{1,l},\boldsymbol{y}_{1,l}}=p^{(1)}$ where the PDF $p^{(1)}$ is defined in (\ref{p1}).
  \item Any symbol of the codeword transmitted during time slots $\tau_2^{(\omega^{+})}+1$ to $\tau_1^{(j)}+n'+n_1$ is received in the presence of interference. For any $\tau_2^{(\omega^{+})}+1\leq l\leq\tau_1^{(j)}+n'+n_1$, we have $p_{\boldsymbol{x}_{1,l},\boldsymbol{y}_{1,l}}=p^{(2)}$ where the PDF $p^{(2)}$ is defined in (\ref{p2}). 
\end{itemize}
According to the interference pattern just described, Rx~$1$ finds the unique codeword $(s_{1,l})_{l=0}^{n_1-1}$ such that all three statements
\begin{eqnarray}
\label{mid1}
\Big((s_{1,l})_{l=0}^{\tau_2^{(\omega^{-})}-\tau_1^{(j)}+n_{2}-1},(y_{1,l})_{l=\tau_1^{(j)}+n'+1}^{\tau_2^{(\omega^{-})}+n'+n_2}\Big)\in A_{\epsilon}^{(\tau_2^{(\omega^{-})}-\tau_1^{(j)}+n_2)}[p^{(2)}],
\end{eqnarray} 
\begin{eqnarray}
\label{mid2}
\Big((s_{1,l})_{l=\tau_2^{(\omega^{-})}-\tau_1^{(j)}+n_{2}}^{\tau_2^{(\omega^{+})}-\tau_1^{(j)}-n'-1},(y_{1,l})_{l=\tau_2^{(\omega^{-})}+n'+n_2+1}^{\tau_2^{(\omega^{+})}}\Big)\in A_{\epsilon}^{(\tau_2^{(\omega^{+})}-\tau_2^{(\omega^{-})}-n'-n_{2})}[p^{(1)}]
\end{eqnarray}
and 
\begin{eqnarray}
\label{mid3}
\Big((s_{1,l})_{l=\tau_2^{(\omega^{+})}-\tau_1^{(j)}-n'}^{n_{1}-1},(y_{1,l})_{l=\tau_2^{(\omega^{+})}+1}^{\tau_1^{(j)}+n'+n_1}\Big)\in A_{\epsilon}^{(\tau_1^{(j)}-\tau_2^{(\omega^{+})}+n'+n_1)}[p^{(2)}]
\end{eqnarray}
hold. We have the following proposition:
\begin{proposition}
\label{prop_4}
Given the index $j$ of a transmitted codeword of Tx~$i$, assume $\omega_{i,j}^{-}\neq0$ and $\omega_{i,j}^{+}\neq 0$. If $\omega_{i,j}^{-}\neq \omega_{i,j}^{+}$, then
 \begin{eqnarray}
 \label{help0}
\eta_i-\frac{1}{\lambda_{i'}}(1+\omega_{i,j})(\phi_i- \psi_i)\eta_{i'}<\theta_i \psi_i-\theta_{i'}(1+\omega_{i,j}) (\phi_i- \psi_i),
\end{eqnarray}
is a sufficient condition for reliable decoding of the $j^{th}$ message of Tx~$i$ where
\begin{equation}
\label{kapp_1}
\phi_i:=\mathsf{C}(\gamma_i),\,\,\,\psi_i:=\mathsf{C}\Big(\frac{\gamma_i}{1+a_{i'}\gamma_{i'}}\Big)
\end{equation}
and the function $\mathsf{C}(\cdot)$ is defined in (\ref{cap_fun11}). If $\omega_{i,j}^{-}=\omega_{i,j}^{+}$, then (\ref{help0}) is replaced by
\begin{equation}
\eta_i<\theta_i\psi_i.
\end{equation}
\end{proposition}
\begin{proof}
See Appendix~G.
\end{proof}
  \item Assume $\omega^{-}\neq 0$, $\omega^{+}=0$ and $\omega_{1,j}=0$. Then 
  \begin{equation}
\label{ }
\tau_{2}^{(\omega^{-})}+1<\tau_1^{(j)}+n'+1\leq \tau_{2}^{(\omega^{-})}+n'+n_2\leq \tau_{1}^{(j)}+n'+n_1<\tau_{2}^{(\omega^-+1)}+1.
\end{equation}
This situation is shown in Fig. \ref{bufpic7} after removing the bust with index $\omega^{+}$ of Tx~$2$ from the picture. The interference pattern over the $j^{th}$ codeword of Tx~$1$ is described as follows: 
\begin{itemize}
  \item Any symbol of the codeword transmitted during time slots $\tau_1^{(j)}+n'+1$ to $\tau_2^{(\omega^{-})}+n'+n_2$ is received in the presence of interference. For any $\tau_1^{(j)}+n'+1\leq l\leq \tau_2^{(\omega^{-})}+n'+n_2$, we have $p_{\boldsymbol{x}_{1,l},\boldsymbol{y}_{1,l}}=p^{(2)}$ where the PDF $p^{(2)}$ is defined in (\ref{p2}). 
  \item Any symbol of the codeword transmitted during time slots $\tau_2^{(\omega^{-})}+n'+n_2+1$ to $\tau_1^{(j)}+n'+n_1$ does not experience interference. For any $\tau_2^{(\omega^{-})}+n'+n_2+1\leq l\leq \tau_1^{(j)}+n'+n_1$, we have $p_{\boldsymbol{x}_{1,l},\boldsymbol{y}_{1,l}}=p^{(1)}$ where the PDF $p^{(1)}$ is defined in (\ref{p1}).
\end{itemize}
According to the interference pattern just described, Rx~$1$ finds the unique codeword $(s_{1,l})_{l=0}^{n_1-1}$ such that the two constraints
\begin{eqnarray}
\label{mid1_11}
\Big((s_{1,l})_{l=0}^{\tau_2^{(\omega^{-})}-\tau_1^{(j)}+n_2-1},(y_{1,l})_{l=\tau_1^{(j)}+n'+1}^{\tau_2^{(\omega^{-})}+n'+n_2}\Big)\in A_{\epsilon}^{(\tau_2^{(\omega^{-})}-\tau_1^{(j)}+n_2)}[p^{(2)}]
\end{eqnarray} 
and
\begin{eqnarray}
\label{mid2_22}
\Big((s_{1,l})_{l=\tau_2^{(\omega^{-})}-\tau_1^{(j)}+n_2}^{n_1-1},(y_{1,l})_{l=\tau_2^{(\omega^{-})}+n'+n_2+1}^{\tau_1^{(j)}+n'+n_1}\Big)\in A_{\epsilon}^{(\tau_1^{(j)}-\tau_2^{(\omega^{-})}+n_1-n_2)}[p^{(1)}]
\end{eqnarray}
hold. 
\begin{proposition}
\label{prop_5}
Given the index $j$ of a transmitted codeword of Tx~$i$, assume $\omega_{i,j}^{-}\neq 0$ and $\omega_{i,j}^{+}=0$. Then
 \begin{eqnarray}
 \label{help1}
      \big(1-\frac{j}{\lambda_i}( \phi_i- \psi_i)\big)\eta_i+\frac{\omega_{i,j}^-}{\lambda_{i'}}( \phi_i- \psi_i)\eta_{i'}<\theta_i\phi_i-\big(\nu_{i'}-\nu_i+\theta_{i'}(1+\omega_{i,j})\big)(\phi_i- \psi_i)
\end{eqnarray}
is a sufficient condition for reliable decoding of the $j^{th}$ message of Tx~$i$ where $\phi_i$ and $\psi_i$ are defined in (\ref{kapp_1}). 
\end{proposition}
\begin{proof}
See Appendix~H.
\end{proof}
  \item Assume $\omega^{+}\neq 0$, $\omega^{-}=0$ and $\omega_{1,j}=0$. We have 
  \begin{equation}
\label{}
\tau_2^{(\omega^+-1)}+n'+n_2<\tau_1^{(j)}+n'+1\leq \tau_{2}^{(\omega^{+})}+1\leq \tau_{1}^{(j)}+n'+n_1<\tau_{2}^{(\omega^{+})}+n'+n_2.
\end{equation}
This situation is shown in Fig. \ref{bufpic7} after removing the burst with index $\omega^{-}$ of Tx~$2$ from the picture. The interference pattern over the $j^{th}$ codeword of Tx~$1$ is described as follows: 
\begin{itemize}
  \item Any symbol of the codeword transmitted during time slots $\tau_1^{(j)}+n'+1$ to $\tau_2^{(j^{+})}$ does not experience interference. For any $\tau_1^{(j)}+n'+1\leq l\leq \tau_2^{(\omega^{+})}$, we have $p_{\boldsymbol{x}_{1,l},\boldsymbol{y}_{1,l}}=p^{(1)}$ where the PDF $p^{(1)}$ is defined in (\ref{p1}). 
  \item Any symbol of the codeword transmitted during time slots $\tau_2^{(\omega^{+})}+1$ to $\tau_1^{(j)}+n'+n_1$ is received in the presence of interference. For any $\tau_2^{(\omega^{+})}+1\leq l\leq \tau_1^{(j)}+n'+n_1$, we have $p_{\boldsymbol{x}_{1,l},\boldsymbol{y}_{1,l}}=p^{(2)}$ where the PDF $p^{(2)}$ is defined in (\ref{p2}).
\end{itemize}
According the to the interference pattern just described, Rx~$1$ finds the unique codeword $(s_{1,l})_{l=0}^{n_1-1}$ such that the two constraints
\begin{eqnarray}
\label{mid1_11}
\Big((s_{1,l})_{l=0}^{\tau_2^{(\omega^{+})}-\tau_1^{(j)}-n'-1},(y_{1,l})_{l=\tau_1^{(j)}+n'+1}^{\tau_2^{(\omega^{+})}}\Big)\in A_{\epsilon}^{(\tau_2^{(\omega^{+})}-\tau_1^{(j)}-n')}[p^{(1)}]
\end{eqnarray} 
and
\begin{eqnarray}
\label{mid2_22}
\Big((s_{1,l})_{l=\tau_2^{(\omega^{+})}-\tau_1^{(j)}-n'}^{n_1-1},(y_{1,l})_{l=\tau_2^{(\omega^{+})}+1}^{\tau_1^{(j)}+n'+n_1}\Big)\in A_{\epsilon}^{(\tau_1^{(j)}-\tau_2^{(\omega^{+})}+n'+n_1)}[p^{(2)}]
\end{eqnarray}
hold. 
\begin{proposition}
\label{prop_6}
Given the index $j$ of a transmitted codeword of Tx~$i$, assume $\omega_{i,j}^{-}=0$ and $\omega_{i,j}^{+}\neq 0$. Then
\begin{equation}
\label{help2}
      \big(1+\frac{j}{\lambda_i}(\phi_i- \psi_i)\big)\eta_i-\frac{\omega_{i,j}^+}{\lambda_{i'}}(\phi_i- \psi_i)\eta_{i'}< \theta_i \psi_i+(\nu_{i'}-\nu_i-\theta_{i'}\omega_{i,j})(\phi_i- \psi_i).
\end{equation} 
is a sufficient condition for reliable decoding of the $j^{th}$ message of Tx~$i$ where $\phi_i$ and $\psi_i$ are defined in (\ref{kapp_1}).
\end{proposition}
\begin{proof}
The proof is similar to the proof of Proposition~\ref{prop_5} and is omitted. 
\end{proof}
\item If $\omega_{i,j}^-=\omega_{i,j}^+=0$, 
 \begin{equation}
\label{ynbt11}
\eta_i< \theta_i\phi_i-\theta_{i'}\omega_{i,j}(\phi_i-\psi_i)
\end{equation}
is a sufficient condition for reliable decoding of the $j^{th}$ message of Tx~$i$.
    \end{itemize}
        \begin{corollary}
    \label{coro_1}
    If $\eta_i<\theta_i  \psi_i$, the probability of error in decoding the $j^{th}$ message of Tx~$i$ vanishes by increasing $n$ for any $j$ regardless of the values of $\omega_{i,j}^-$, $\omega_{i,j}^+$, $\omega_{i,j}$, $\nu_i$ and $\nu_{i'}$.  If $\eta_i\geq  \theta_i\phi_i$, the probability of error in decoding any message of Tx~$i$  approaches one by increasing $n$.
    \end{corollary}
    \begin{proof}
     Let $\eta_i<\theta_i\psi_i$. We consider the following cases:
     \begin{itemize}
  \item Assume $\omega^-_{i,j}\neq 0$, $\omega^+_{i,j}\neq 0$ and $\omega^-_{i,j}\neq \omega^+_{i,j}$. Since $\eta_i<\theta_i\psi_i$ and $\theta_{i'}<\mu_{i'}$ by (\ref{buffstab}), we must have $\eta_i<\theta_i \psi_i+(\mu_{i'}-\theta_{i'})(1+\omega_{i,j}) (\phi_i- \psi_i)$ which is exactly (\ref{help0}) after rearranging terms.
  \item Assume $\omega^-_{i,j}\neq 0$ and $\omega^+_{i,j}=0$. On the $\overline{t}$-axis, the interval $\mathcal{J}_0=(j\mu_i+\nu_i,\omega^-\mu_{i'}+\nu_{i'}+\theta_{i'})$ together with $\omega_{i,j}$ intervals $\mathcal{J}_1,\cdots, \mathcal{J}_{\omega_{i,j}}$ each of length $\theta_{i'}$ corresponding to the bursts of indices $\omega^-+1,\cdots, \omega^-+\omega_{i,j}$ of Tx~$i'$ are disjoint intervals and all are included in the interval $\mathcal{J}=(j\mu_i+\nu_i,j\mu_{i}+\nu_{i}+\theta_{i})$ corresponding to the $j^{th}$ codeword of Tx~$i$. Hence, the sum of the lengths of the intervals $\mathcal{J}_0,\mathcal{J}_1,\cdots, \mathcal{J}_{\omega_{i,j}}$ which is $\vartheta:=\omega^-\mu_{i'}+\nu_{i'}+\theta_{i'}-(j\mu_i+\nu_i)+\theta_{i'}\omega_{i,j}$ must be less than or equal to the length $\theta_i$ of $\mathcal{J}$. Then one can write $\theta_i\psi_i\leq (\theta_i-\vartheta)\phi_i+\vartheta\psi_i$ due to $\theta_i-\vartheta\geq 0$. Since $\eta_i<\theta_i\psi_i$ by assumption, we get $\eta_i<(\theta_i-\vartheta)\phi_i+\vartheta\psi_i$ which is exactly~(\ref{help1}) after rearranging terms.
  \item The cases $\omega^-_{i,j}=0, \omega^+_{i,j}\neq 0$ and $\omega^-_{i,j}=\omega^+_{i,j}= 0$ are analyzed similarly. We omit the details.  
\end{itemize}
     If $\eta_i>\theta_i\phi_i$, reliable communication is impossible due to the fact that the capacity of an AWGN~channel with SNR $\gamma_i$ is $ \phi_i$.  The codebook~rate for Tx~$i$ is $\lim_{n\to\infty}\frac{\lfloor n\eta_i\rfloor}{\lfloor n\theta_i\rfloor}=\frac{\eta_i}{\theta_i}$. The probability of error tends to one if $\frac{\eta_i}{\theta_i}\geq \phi_i$.  
     \end{proof}
     
     \textbf{Remark}- Let us describe a simple method to obtain the constraints (\ref{help0}), (\ref{help1}) and (\ref{help2}) in Propositions~\ref{prop_4}, \ref{prop_5} and~\ref{prop_6}, respectively.  For example, let us discuss how to obtain (\ref{help1}) by looking at the positions of the bursts on the $\overline{t}$-axis. Fig.~\ref{bic_fig2} depicts the $j^{th}$ codeword of Tx~$i$ in a situation where $\omega^-_{i,j}\neq 0$ and $\omega^+_{i,j}=0$. The table in~(\ref{table1}) shows the numbers on the $\overline{t}$-axis corresponding to different points in Fig~\ref{bic_fig2}. The interval during which the $j^{th}$ burst of Tx~$i$ is sent can be divided into two subintervals, i.e., subinterval~1 where the two users interfere and subinterval~2 where there is no interference. Using~(\ref{table1}) it is easy to see that 
        \begin{figure}[t]
  \centering
  \includegraphics[scale=.8] {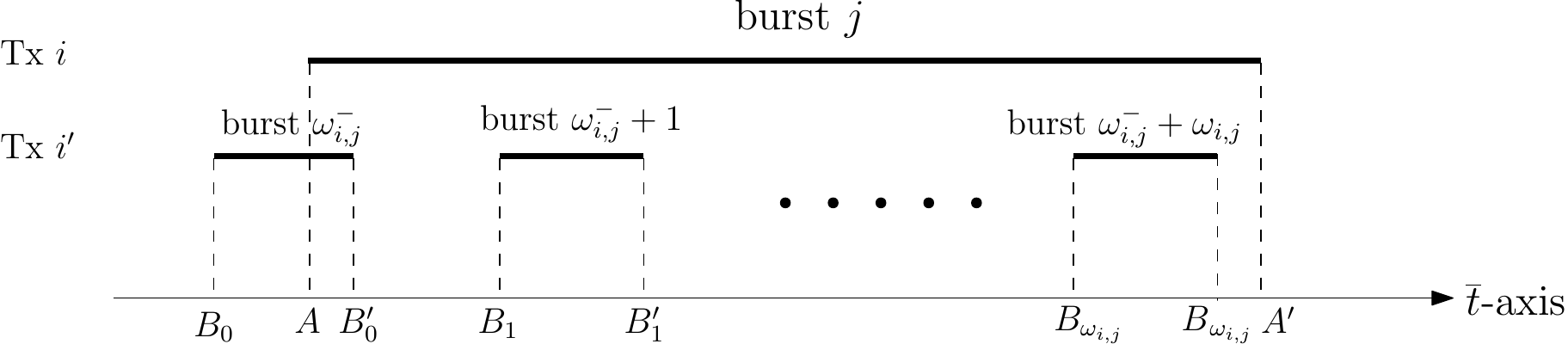}
  \caption{This picture shows the bursts of both users on the $\overline{t}$-axis in a situation where $\omega^-_{i,j}\neq 0$ and $\omega^{+}_{i,j}=0$}
  \label{bic_fig2}
 \end{figure} 
     \begin{figure*}
      \begin{eqnarray}
 \label{table1}
\begin{tabular}{|c|c|c|c|}
\hline
&\\
   $A=j\mu_i+\nu_i$ & $A'=j\mu_i+\nu_i+\theta_i$ \\
   &\\
   \hline
   &\\
   $\begin{array}{c}
      B_m=(\omega^-_{i,j}+m)\mu_{i'}+\nu_{i'}    \\\\
       m=0,1,\cdots,\omega_{i,j}   
\end{array}$ & $\begin{array}{c}
      B'_m=(\omega^-_{i,j}+m)\mu_{i'}+\nu_{i'}+\theta_{i'}    \\\\
       m=0,1,\cdots,\omega_{i,j}   
\end{array}$\\
   &\\
\hline
\end{tabular}
\end{eqnarray}
\end{figure*}
     \begin{eqnarray}
     \label{three_0}
\textrm{length of subinterval~1}&=&AB'_0+B_1B'_1+\cdots+B_{\omega_{i,j}}B'_{\omega_{i,j}}\notag\\
&=&\omega^-_{i,j}\mu_{i'}+\nu_{i'}+\theta_{i'}-(j\mu_i+\nu_i)+\omega_{i,j}\theta_{i'}
\end{eqnarray}
and 
  \begin{eqnarray}
  \label{three_1}
\textrm{length of subinterval~2}&=&\theta_i-\textrm{length of subinterval~1}\notag\\
&=&\theta_i-\big(\omega^-_{i,j}\mu_{i'}+\nu_{i'}+\theta_{i'}-(j\mu_i+\nu_i)+\omega_{i,j}\theta_{i'}\big).
\end{eqnarray}
Finally, the criterion for successful decoding of the $j^{th}$ transmitted codeword of Tx~$i$ is given by 
\begin{eqnarray}
\label{three_2}
\eta_{i}<\textrm{(length of subinterval~1)}\,\psi_{i}+\textrm{(length of subinterval~2)}\,\phi_{i}.
\end{eqnarray}
Substituting (\ref{three_0}) and (\ref{three_1}) in (\ref{three_2}) and rearranging terms, we get (\ref{help1}).
    \section{The admissible set $\mathcal{A}$ for $(\nu_1,\nu_2)$ and the probability of outage}
    \subsection{System Design}
In Section~II.D we obtained the average transmission power and the average transmission rate for Tx~$i$ as $Q_i=\frac{N_i\gamma_i}{1+\frac{1}{q_i\theta_i}}$ and $R_i=\frac{\lambda_i}{1+q_i\theta_i}$, respectively, in the asymptote of large $n$. Throughout this section we assume none of the transmitters performs power control and both transmit at full power, i.e., 
\begin{equation}
\label{ }
Q_i=P_i,\,\,\,i=1,2.
\end{equation}
We get  
\begin{eqnarray}
\label{buffstab1}
\theta_i=\hat{\theta}_i:=\frac{1}{q_i}\Big(\frac{\lambda_i}{R_i}-1\Big)
\end{eqnarray}
and
  \begin{equation}
\label{ }
\gamma_i=\hat{\gamma}_{i}:=\frac{P_i}{N_i}\Big(1+\frac{1}{q_i\hat{\theta}_i}\Big)=\frac{P_i}{N_i}\frac{\lambda_i}{\lambda_i-R_i}.
\end{equation}
If $N_i>1$, then (\ref{buffstab}) together with (\ref{buffstab1}) imply that  
\begin{equation}
\label{r_11}
\frac{N_i}{N_i+1}\lambda_i<R_i<\lambda_i.
\end{equation}
If $N_i=1$, we simply have 
\begin{equation}
\label{r_22}
0<R_i<\lambda_i.
\end{equation}
 In Section~IV, we derived sufficient conditions for successful decoding at the receivers. Letting $\theta_i=\hat{\theta}_i$ and $\gamma_i=\hat{\gamma}_{i}$, we aim to characterize an admissible region $\mathcal{A}$ for $(\nu_1,\nu_2)$ such that reliable communication is guaranteed for all transmitted codewords. Define
\begin{equation}
\label{ }
\hat{\phi}_i:=\phi_i\big|_{\gamma_i=\hat{\gamma}_i },\,\,\,\,\hat{\psi}_i:=\psi_i\big|_{\gamma_1=\hat{\gamma}_1,\gamma_2=\hat{\gamma}_2 }
\end{equation}
If $\eta_i\geq \hat{\theta}_i\hat{\phi}_i$, reliable communication is impossible for Tx~$i$ as stated in Corollary~\ref{coro_1}. For any value of $N_i$, define\footnote{\label{ft_3}Using the change of variable $x:=-\frac{1}{N_i}\frac{\lambda_i}{\lambda_i-R_i}-\frac{1}{P_i}$, the inequality $\eta_i <\hat{\theta}_i\hat{\phi}_i$ can be written as $x2^x>-\frac{1}{P_i}2^{-\frac{1}{N_i}-\frac{1}{P_i}}$ which is equivalent to $x>x_0$ for some $x_0$. This in turn results in the solution $0<R_i<\overline{R}_i(N_i)$ for $R_i$.}
\begin{eqnarray}
\label{r_33}
\overline{R}_i(N_i):=\sup\{R_i: \eta_i<\hat{\theta}_i\hat{\phi}_i\}.
\end{eqnarray}
By (\ref{r_11}), (\ref{r_22}) and (\ref{r_33}), we demand that $R_i$ be in the interval
\begin{equation}
\label{r_44}
\frac{N_i}{N_i+1}\mathds{1}_{N_i>1}\lambda_i<R_i<\min\{\lambda_i,\overline{R}_i(N_i)\}=\overline{R}_i(N_i),
\end{equation}
where we have used the fact that\footnote{Recall from Footnote~\ref{ft_3} that $\eta_i <\hat{\theta}_i\hat{\phi}_i$ is equivalent to $0<R_i<\overline{R}_i(N_i)$. It is easy to see that $\lim_{R_i\to\lambda^-_i}\hat{\theta}_i\hat{\phi}_i=0$. Hence, there exists a $0<\delta<\lambda_i$ such that $\hat{\theta}_i\hat{\phi}_i<\eta_i$ for $R_i=\lambda_i-\delta$. This gives $\overline{R}_i(N_i)\leq\lambda_i-\delta$ as desired.} $\overline{R}_i(N_i)<\lambda_i$. 
 For any $R_1,R_2$ define 
\begin{eqnarray}
\mathcal{N}_{R_1,R_2}:=\big\{(N_1,N_2): \textrm{$R_i$ satisfies (\ref{r_44}) for $i=1,2$}\big\}.
\end{eqnarray}
We call $\mathcal{N}_{R_1,R_2}$ the \textit{active set} for the pair $(R_1,R_2)$. Note that $\mathcal{N}_{R_1,R_2}$ is finite due to the constraint $\frac{N_i}{N_i+1}\mathds{1}_{N_i>1}<\frac{R_i}{\lambda_i}<1$ in (\ref{r_44}). 

By Corollary~\ref{coro_1}, if $\eta_i< \hat{\theta}_i \hat{\psi}_i$, the codewords of user~$i$ are received reliably regardless of the values of $\nu_1$ and $\nu_2$. A more interesting situation occurs when $R_1,R_2$ satisfy 
\begin{equation}
\label{please1}
 \min_{(N_1,N_2)\in\,\mathcal{N}_{R_1,R_2}}\,\max\bigg\{\frac{\eta_1}{\hat{\theta}_1 \hat{\psi}_1}, \frac{\eta_2}{\hat{\theta}_2 \hat{\psi}_2}\bigg\}\geq 1.
\end{equation}
 The inequality in (\ref{please1})  implies that for any $(N_1,N_2)\in\mathcal{N}_{R_1,R_2}$, there is $i\in\{1,2\}$ such that $\eta_i\geq\hat{\theta}_i\hat{\psi}_i$ and hence, the values of $\nu_1$ and $\nu_2$ \textit{may} potentially affect reliable communication for Tx~$i$. 
 
 \textbf{Remark}- Throughout the rest of this section, we are only interested in rate pairs $(R_1,R_2)$ and system parameters $q_1,q_2,\lambda_1,\lambda_2, a_1,a_2, P_1,P_2$ such that (\ref{please1}) holds.   

Towards characterizing an admissible region $\mathcal{A}$ for $(\nu_1,\nu_2)$,  we need to define the concept of the \textit{state} in a GIC-SDA under immediate transmissions. Let $A_l=l\mu_1+\nu_1$ and $A'_{l}=l\mu_1+\nu_1+\theta_1$ be the starting point and the ending point of the $l^{th}$ codeword of Tx~1 on the $\overline{t}$-axis. The points $A_l, A'_l$ for $1\leq l\leq N_1$ partition the $\overline{t}$-axis into $2N_1+1$ disjoint intervals 
\begin{eqnarray}
\begin{array}{c}
        \mathcal{I}_1:=(-\infty,A_1)  \\
     \mathcal{I}_2:=(A_1,A'_1) \\   
     \mathcal{I}_3:=(A'_1,A_2)\\
     \vdots\\
     \mathcal{I}_{2N_1-1}:=(A'_{N_1-1},A_{N_1})\\
     \mathcal{I}_{2N_1}:=(A_{N_1},A'_{N_1})\\
     \mathcal{I}_{2N_1+1}=(A'_{N_1},\infty)
\end{array}.
\end{eqnarray}
 For any $1\leq j\leq N_2$, we assign a tuple $(u_j,v_j)$ to the $j^{th}$ transmitted codeword of Tx~2 where 
 \begin{itemize}
  \item $u_j$ is the unique index~$m$ such that the \textit{starting} point of the $j^{th}$ codeword of Tx~2 lies in interval $\mathcal{I}_m$.
  \item $v_j$ is the unique index~$m$ such that the \textit{ending} point of the $j^{th}$ codeword of Tx~2 lies in interval $\mathcal{I}_m$.
\end{itemize}
 Define the state of the asynchronous GIC-SDA by 
\begin{eqnarray}
S=\{(j; u_j,v_j): 1\leq j\leq N_2\}.
\end{eqnarray}
  \begin{figure}[t]
  \centering
  \includegraphics[scale=.8] {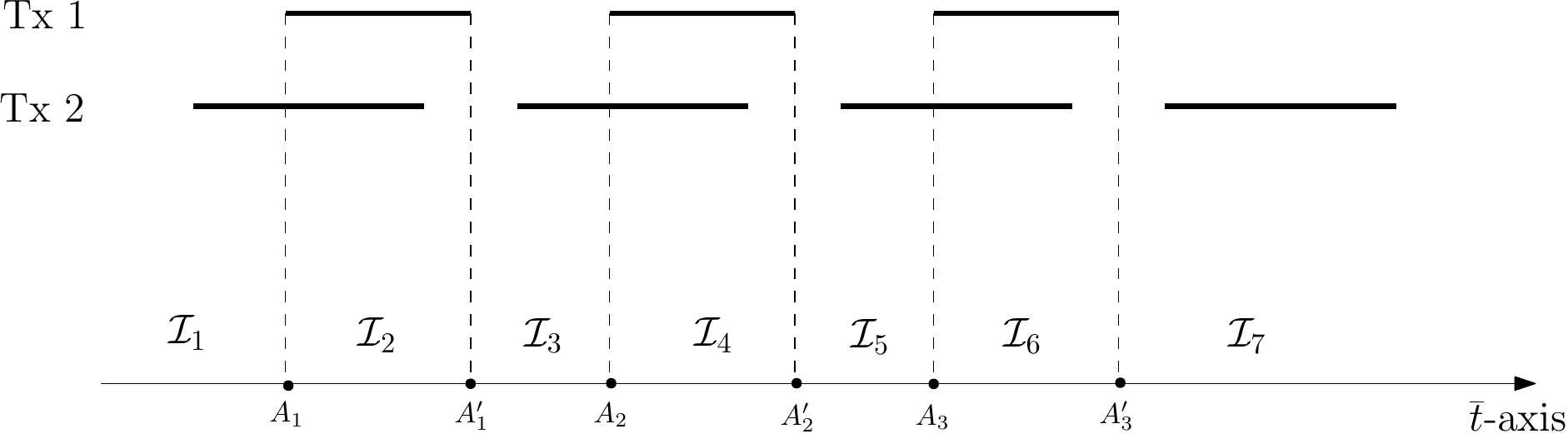}
  \caption{This picture shows the positions of the bursts of both users on the $\overline{t}$-axis in a situation where $N_1=3$ and $N_2=4$. For $1\leq l\leq 3$, the points $A_l=l\mu_1+\nu_1$ and $A'_{l}=l\mu_1+\nu_1+\theta_1$ are the starting point and the ending point of the $l^{th}$ codeword of Tx~1. These points partition the $\overline{t}$-axis into seven disjoint intervals $\mathcal{I}_1,\cdots,\mathcal{I}_7$.}
  \label{bic_fig10}
 \end{figure} 
 The set of all states is shown by $\mathscr{S}$. For example, Fig.~\ref{bic_fig10} depicts a situation where $N_1=3$ and $N_2=4$. The $\overline{t}$-axis is partitioned into seven intervals $\mathcal{I}_1,\cdots,\mathcal{I}_7$. We have $(u_1,v_1)=(1,2), (u_2,v_2)=(3,4), (u_3,v_3)=(5,6), (u_4,v_4)=(7,7)$ and the state of the channel is given by 
 \begin{eqnarray}
S=\{(1;1,2),(2;3,4),(3;5,6),(4;7,7)\}.
\end{eqnarray}   
In general, the number of states in a two-user asynchronous GIC-SDA is $|\mathscr{S}|=\binom{2N_1+2N_2}{2N_2}$. A proof of this fact is given in Appendix~I. Note that any state uniquely determines the parameters $\omega^-_{i,j}, \omega^+_{i,j}$ and $\omega_{i,j}$ for any $1\leq i\leq 2$ and $1\leq j\leq N_i$.

 For any state $S$, we impose two sets of constraints on $(\nu_1,\nu_2)$ referred to as the geometric constraints $\mathcal{A}^{\mathrm{(geom)}}_S$ and the reliability constraints $\mathcal{A}^{\mathrm{(rel)}}_S$. The admissible region $\mathcal{A}$ is defined by 
\begin{equation}
\label{ }
\mathcal{A}:=\bigcup_{S\in\mathscr{S}}\big(\mathcal{A}^{(\mathrm{geom})}_S\bigcap\mathcal{A}^{(\mathrm{rel})}_{S}\big).
\end{equation}
 The geometric constraints are dictated by the positions of the bursts on the $\overline{t}$-axis.  
   \begin{figure}[t]
  \centering
  \includegraphics[scale=.8] {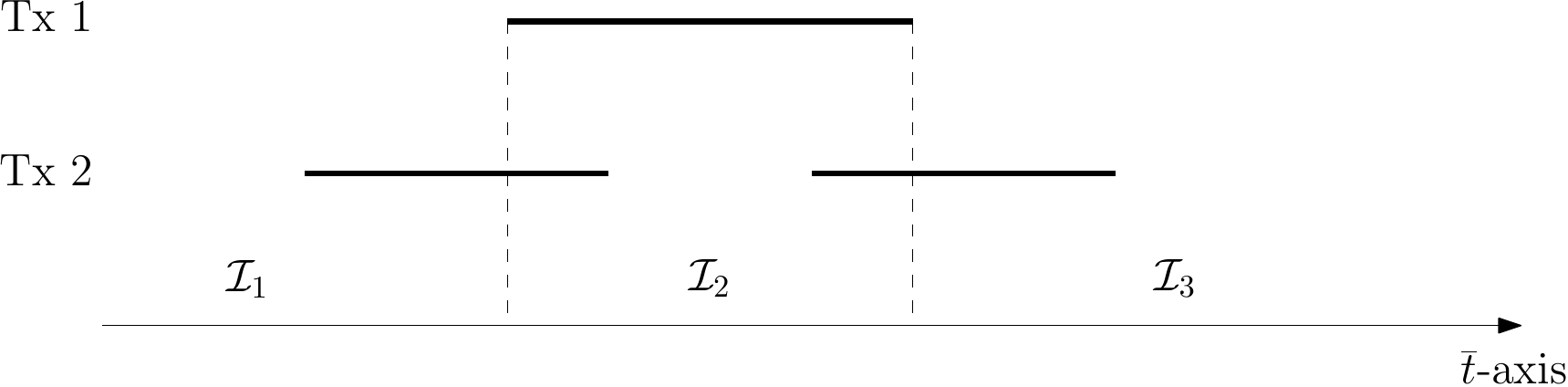}
  \caption{This picture shows the positions of the bursts of both users on the $\overline{t}$-axis in a situation where $N_1=1$ and $N_2=2$. The state of the channel is $\{(1;1,2),(2;2,3)\}$.}
  \label{bic_fig11}
 \end{figure} 
 For example, let $N_1=1$, $N_2=2$ and the state of the channel be $S=\{(1;1,2),(2;2,3)\}$ as shown in Fig.~\ref{bic_fig11}. Then  
 \begin{equation}
\label{oak11}
\mathcal{A}^{(\mathrm{geom})}_S=\big\{(\nu_1,\nu_2):\mu_2+\nu_2<\mu_1+\nu_1<\mu_2+\nu_2+\theta_2<2\mu_2+\nu_2<\mu_1+\nu_1+\theta_1<2\mu_2+\nu_2+\theta_2\big\}.
\end{equation}
The reliability constraints guarantee reliable communications for all transmitted codewords. For example, for the situation in Fig.~\ref{bic_fig11} there are three reliability constraints: 
\begin{itemize}
  \item For the first codeword of Tx~1, $\omega^-_{1,1}=1, \omega^+_{1,1}=2$ and $\omega_{1,1}=0$. By Proposition~\ref{prop_4}, 
  \begin{eqnarray}
  \label{ineq_1}
\eta_1-\frac{1}{\lambda_{2}}(\hat{\phi}_1- \hat{\psi}_1)\eta_{2}<\hat{\theta}_1 \hat{\psi}_1-\hat{\theta}_{2}(\hat{\phi}_1- \hat{\psi}_1).
\end{eqnarray}
  \item For the first codeword of Tx~2, $\omega^-_{2,1}=0, \omega^+_{2,1}=1$ and $\omega_{2,1}=0$. By Proposition~\ref{prop_6}, 
  \begin{eqnarray}
  \label{ineq_2}
\big(1+\frac{1}{\lambda_2}(\hat{\phi}_2- \hat{\psi}_2)\big)\eta_2-\frac{1}{\lambda_{1}}(\hat{\phi}_2- \hat{\psi}_2)\eta_{1}< \hat{\theta}_2 \hat{\psi}_2+(\nu_{1}-\nu_2)(\hat{\phi}_2- \hat{\psi}_2).
\end{eqnarray}
   \item For the second codeword of Tx~2, $\omega^-_{2,2}=1, \omega^+_{2,2}=0$ and $\omega_{2,2}=0$. By Proposition~\ref{prop_5}, 
  \begin{eqnarray}
  \label{ineq_3}
\big(1-\frac{2}{\lambda_2}( \hat{\phi}_2- \hat{\psi}_2)\big)\eta_2+\frac{1}{\lambda_{1}}( \hat{\phi}_2- \hat{\psi}_2)\eta_{1}<\hat{\theta}_2\hat{\phi}_2-(\nu_{1}-\nu_2+\hat{\theta}_{1})(\hat{\phi}_2- \hat{\psi}_2).
\end{eqnarray}
\end{itemize} 
Then $\mathcal{A}^{(\mathrm{rel})}_{S}$ for $S=\{(1;1,2),(2;2,3)\}$ is the set of all $(\nu_1,\nu_2)$ such that the three inequalities in (\ref{ineq_1}), (\ref{ineq_2}) and (\ref{ineq_3}) hold. 

We are ready to state our design problem. Let $(R_1,R_2)$ be such that (\ref{please1}) holds and  $\nu_1$ and $\nu_2$ be realizations of independent \textit{uniform} random variables\footnote{One can consider any arbitrary continuous distribution for $\boldsymbol{\nu}_1$ and $\boldsymbol{\nu}_2$. We consider the uniform distribution due to its realistic~nature.} $\boldsymbol{\nu}_1$ and $\boldsymbol{\nu}_2$, respectively, with support~$[0,d]$ for some $d>0$. We aim to find $(N_1,N_2)\in\mathcal{N}_{R_1,R_2}$ such that the probability of $(\boldsymbol{\nu}_1,\boldsymbol{\nu}_2)$ not being in the admissible region $\mathcal{A}$ is minimized, i.e., 
\begin{eqnarray}
\label{des_11}
(\widehat{N}_1,\widehat{N}_2)=\arg\,\min_{(N_1,N_2)\in\,\mathcal{N}_{R_1,R_2}}\,\mathbb{P}\big((\boldsymbol{\nu}_1,\boldsymbol{\nu}_2)\notin\mathcal{A}\big).
\end{eqnarray}
In words, (\ref{des_11}) answers the following question:

\textit{Given the value of $d$ and assuming the rate pair $(R_1,R_2)$ is such that (\ref{please1}) holds, What is the optimum number of transmission bursts $N_i$ for Tx~$i$ in order to minimize the probability of the outage event, i.e., the event that $(\boldsymbol{\nu}_1,\boldsymbol{\nu}_2)$ does not lie in the admissible region $\mathcal{A}$? }

 Since the sets $\mathcal{A}^{(\mathrm{geom})}_S\bigcap\mathcal{A}^{(\mathrm{rel})}_{S}$ are disjoint\footnote{This is due to the fact that the geometric constraints $\mathcal{A}^{(\mathrm{geom})}_S$ are disjoint for different states $S$.} for different states $S$, we can write 
 \begin{eqnarray}
 \label{out_11}
\mathbb{P}\big((\boldsymbol{\nu}_1,\boldsymbol{\nu}_2)\notin\mathcal{A}\big)=1-\sum_{S\in\mathscr{S}}\mathbb{P}\big((\boldsymbol{\nu}_1,\boldsymbol{\nu}_2)\in\,\mathcal{A}^{(\mathrm{geom})}_S\bigcap\mathcal{A}^{(\mathrm{rel})}_{S}\big).
\end{eqnarray}
Define 
\begin{eqnarray}
\boldsymbol{\alpha}:=\boldsymbol{\nu}_2-\boldsymbol{\nu}_1.
\end{eqnarray}
Each constraint $(\boldsymbol{\nu}_1,\boldsymbol{\nu}_2)\in\,\mathcal{A}^{(\mathrm{geom})}_S\bigcap\mathcal{A}^{(\mathrm{rel})}_{S}$ is in fact a constraint on $\boldsymbol{\alpha}$, i.e., for any state $S$ there are real numbers $\alpha^{(l)}_S$ and $\alpha^{(u)}_S$ such that\footnote{For example, look at the geometric and reliability constraints given for the state depicted in Fig.~\ref{bic_fig11}.} 
\begin{eqnarray}
\label{out_22}
(\boldsymbol{\nu}_1,\boldsymbol{\nu}_2)\in\,\mathcal{A}^{(\mathrm{geom})}_S\bigcap\mathcal{A}^{(\mathrm{rel})}_{S} \iff \alpha^{(l)}_S<\boldsymbol{\alpha}<\alpha^{(u)}_S.
\end{eqnarray} 
By (\ref{out_11}) and (\ref{out_22}),  
\begin{eqnarray}
\mathbb{P}\big((\boldsymbol{\nu}_1,\boldsymbol{\nu}_2)\notin\mathcal{A}\big)=1-\sum_{S\in\mathscr{S}}\big(F_{\boldsymbol{\alpha}}(\alpha^{(u)}_{S})-F_{\boldsymbol{\alpha}}(\alpha^{(l)}_S)\big),
\end{eqnarray}
where $F_{\boldsymbol{\alpha}}(\alpha)=\frac{1}{d}(1-\frac{|\alpha|}{d})\mathds{1}_{|\alpha|\le d}$ is the cumulative distribution function of $\boldsymbol{\alpha}$. 

The next proposition provides conditions on the parameter $d>0$ such that reliable communication is guaranteed for both users regardless of the values of $\nu_1,\nu_2\in[0,d]$: 
\begin{proposition}
\label{prop_r}
Let $\mathcal{S}_d:=(0,d)\times(0,d)$. Then  $\mathbb{P}\big((\boldsymbol{\nu}_1,\boldsymbol{\nu}_2)\notin\mathcal{A}\big)=0$ if and only if 
\begin{eqnarray}
\label{empty_11}
\mathcal{S}_d\bigcap(\mathcal{A}^{(\mathrm{geom})}_S\setminus\mathcal{A}^{(\mathrm{rel})}_S)=\emptyset,
\end{eqnarray}
for any $S\in\mathscr{S}$.
\end{proposition}
\begin{proof}
Assume $\mathbb{P}\big((\boldsymbol{\nu}_1,\boldsymbol{\nu}_2)\in\mathcal{A}\big)=\sum_{S\in\mathscr{S}}\mathbb{P}\big((\boldsymbol{\nu}_1,\boldsymbol{\nu}_2)\in\,\mathcal{A}^{(\mathrm{geom})}_S\bigcap\mathcal{A}^{(\mathrm{rel})}_{S}\big)=1$. Since $\{\mathcal{A}_S^{(\mathrm{geom})}: S\in\mathscr{S}\}$ is a partition of the sample space, $\sum_{S\in\mathscr{S}}\mathbb{P}\big((\boldsymbol{\nu}_1,\boldsymbol{\nu}_2)\in \mathcal{A}^{(\mathrm{geom})}_S\big)=1$. But, $\mathbb{P}\big((\boldsymbol{\nu}_1,\boldsymbol{\nu}_2)\in\,\mathcal{A}^{(\mathrm{geom})}_S\bigcap\mathcal{A}^{(\mathrm{rel})}_{S}\big)\leq \mathbb{P}\big((\boldsymbol{\nu}_1,\boldsymbol{\nu}_2)\in \mathcal{A}^{(\mathrm{geom})}_S\big)$ for any $S\in\mathscr{S}$. Hence, $\mathbb{P}\big((\boldsymbol{\nu}_1,\boldsymbol{\nu}_2)\in\,\mathcal{A}^{(\mathrm{geom})}_S\bigcap\mathcal{A}^{(\mathrm{rel})}_{S}\big)=\mathbb{P}\big((\boldsymbol{\nu}_1,\boldsymbol{\nu}_2)\in \mathcal{A}^{(\mathrm{geom})}_S\big)$ or equivalently, $\mathbb{P}\big((\boldsymbol{\nu}_1,\boldsymbol{\nu}_2)\in\,\mathcal{A}^{(\mathrm{geom})}_S\setminus\mathcal{A}^{(\mathrm{rel})}_{S}\big)=0$ for any $S\in\mathscr{S}$. This in turn means $\mathcal{S}_d\bigcap(\mathcal{A}^{(\mathrm{geom})}_S\setminus\mathcal{A}^{(\mathrm{rel})}_{S})$ has Lebesgue measure zero. But, $\mathcal{A}^{(\mathrm{geom})}_S\setminus\mathcal{A}^{(\mathrm{rel})}_S$ is a union of a finite number of (disjoint) strips of the form $\{(\nu_1,\nu_2): \beta^{(l)}<\nu_2-\nu_1<\beta^{(u)}\}$ where $\beta^{(l)}$ and $\beta^{(u)}$ are real numbers. As~such, $\mathcal{S}_d\bigcap(\mathcal{A}^{(\mathrm{geom})}_S\setminus\mathcal{A}^{(\mathrm{rel})}_{S})$ has Lebesgue measure zero if and only if it is empty. 
\end{proof}

Motivated by Proposition~\ref{prop_r}, we define
\begin{eqnarray}
d_{\max}:=\sup\Big\{d>0: \textrm{$\mathcal{S}_d\bigcap(\mathcal{A}^{(\mathrm{geom})}_S\setminus\mathcal{A}^{(\mathrm{rel})}_{S})=\emptyset$ for any $S\in\mathscr{S}$}\Big\}.
\end{eqnarray}
If $d_{\max}>0$, then the probability of outage is zero for any $d<d_{\max}$.
 In the next subsection, we offer simulation results to study the effects of different system parameters on the optimal choices for $N_1,N_2$~in~(\ref{des_11}). In particular, we will see an example of a rate tuple $(R_1,R_2)$ that satisfies~(\ref{please1}) and still there exists $d>0$ such that (\ref{empty_11}) holds for any $S\in\mathscr{S}$. Therefore, if $(R_1,R_2)$ satisfies (\ref{please1}), it does not necessarily mean that $\mathbb{P}\big((\boldsymbol{\nu}_1,\boldsymbol{\nu}_2)\notin\mathcal{A}\big)>0$.
\subsection{Simulations}
In this subsection, we study the optimum choices for $N_1,N_2$ in~(\ref{des_11}) in a few examples.  
\begin{figure}[t]
\centering
\subfigure[]{
\includegraphics[scale=0.65]{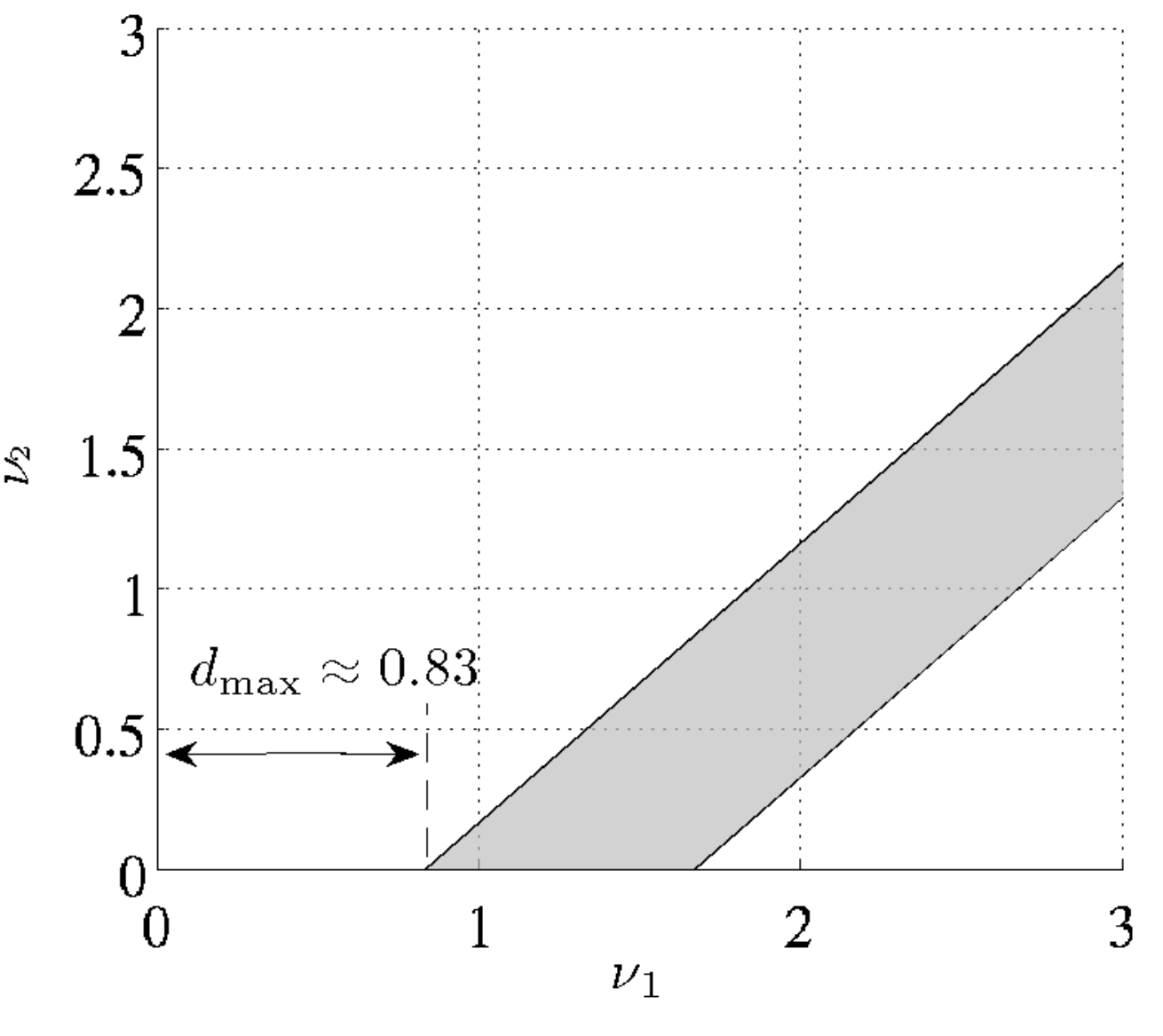}
\label{ref1}
}
\subfigure[]{
\includegraphics[scale=0.65]{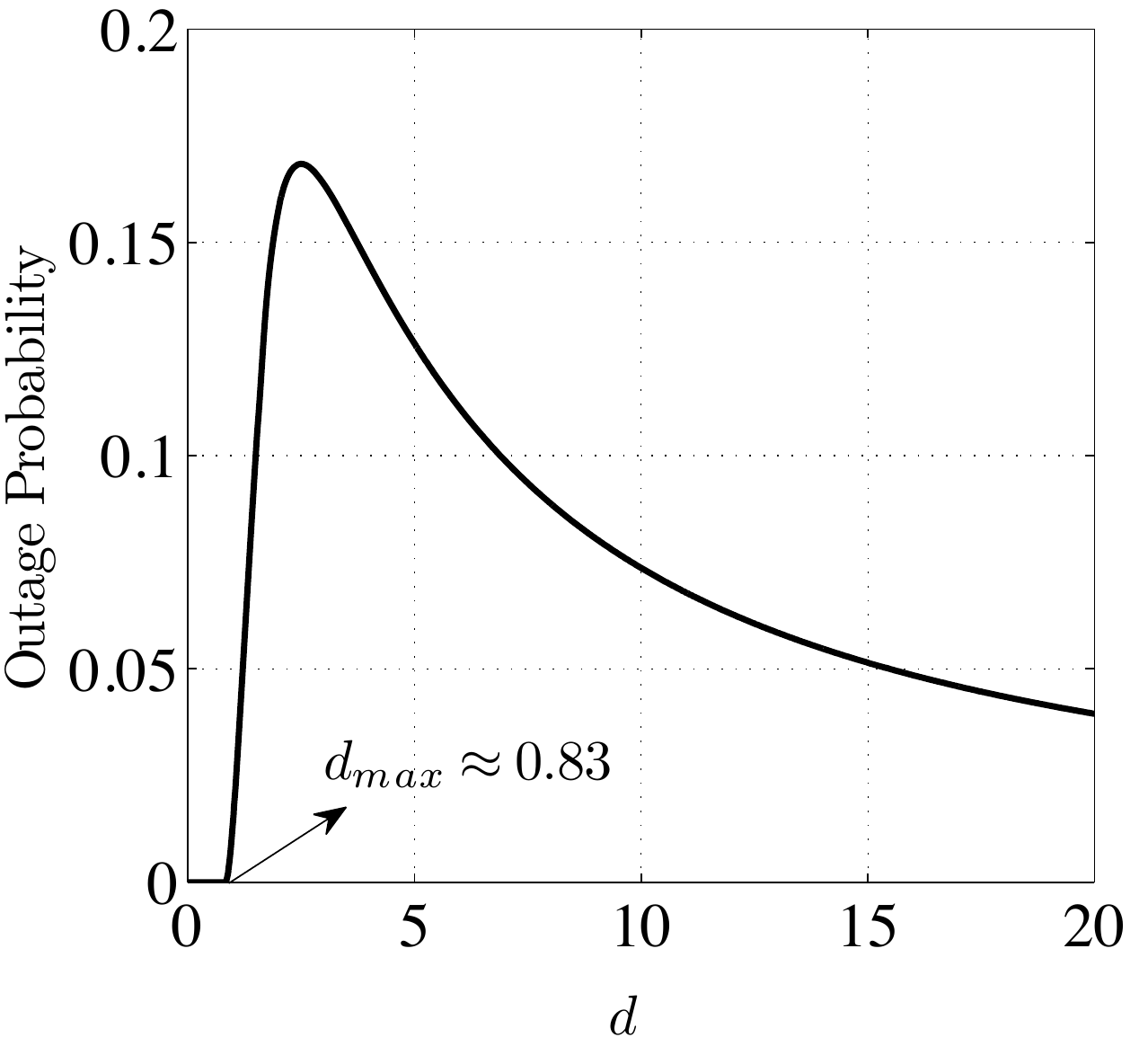}
\label{ref2}
}
\label{fig:subfigureExample}
\caption[Optional caption for list of figures]{Let $k_1=3, k_2=2, q_1=0.3, q_2=0.4, a_1=0.5, a_2=0.7$, $P_1=P_2=30\,\mathrm{dB}$, $R_1=0.5\lambda_1=0.45$ and $R_2=0.5 \lambda_2=0.4$. Then $\mathcal{N}_{R_1,R_2}=\{(1,1)\}$ and (\ref{please1}) is satisfied. Fixing $N_1=N_2=1$, panel~(a) shows the region $\bigcup_{S\in\mathscr{S}}(\mathcal{A}^{(\mathrm{geom})}_S\setminus\mathcal{A}^{(\mathrm{rel})}_{S})$ in grey shade. Panel~(b) shows the probability of outage in terms of $d$.}
\label{pict_390}
\end{figure}
\begin{figure}[t]
\centering
\subfigure[$R_1=0.7\lambda_1, R_2=0.7\lambda_2$]{
\includegraphics[scale=0.63]{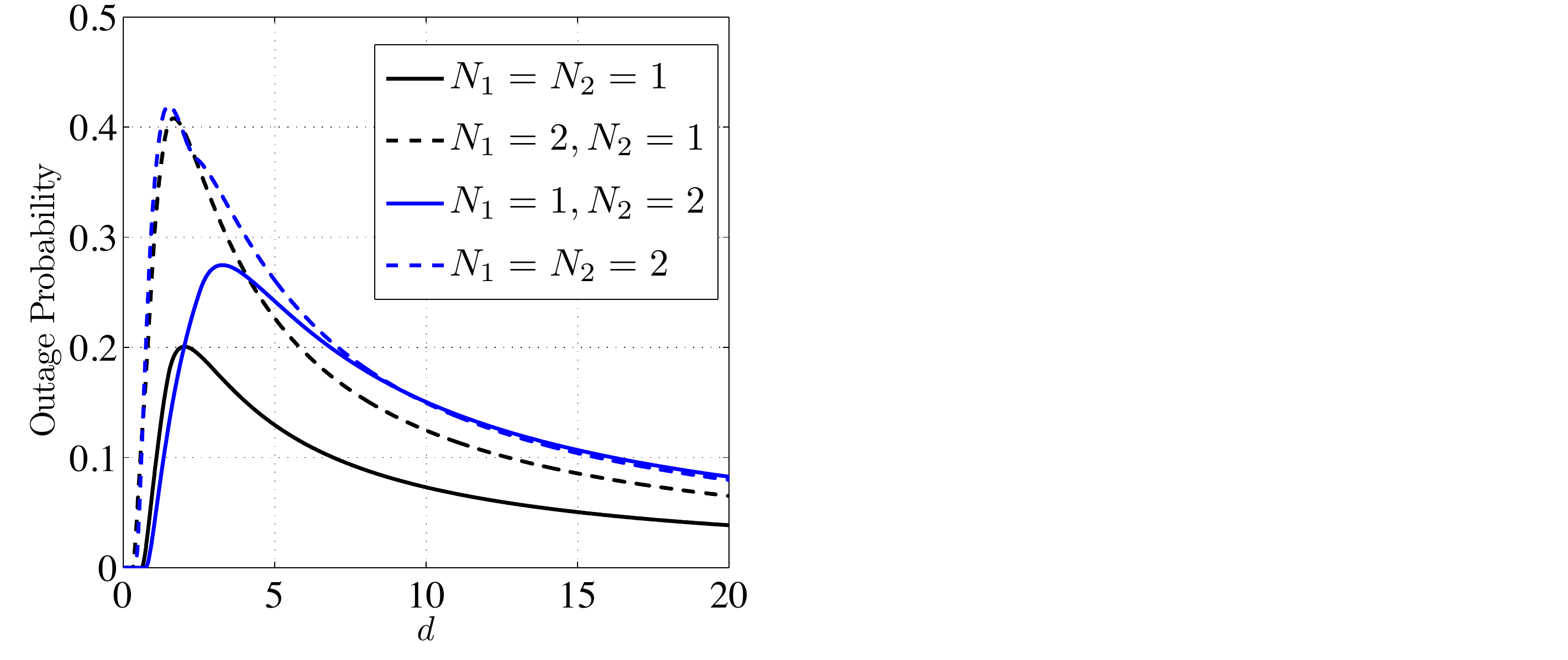}
\label{ref1}
}
\subfigure[$R_1=0.8\lambda_1, R_2=0.8\lambda_2$]{
\includegraphics[scale=0.62]{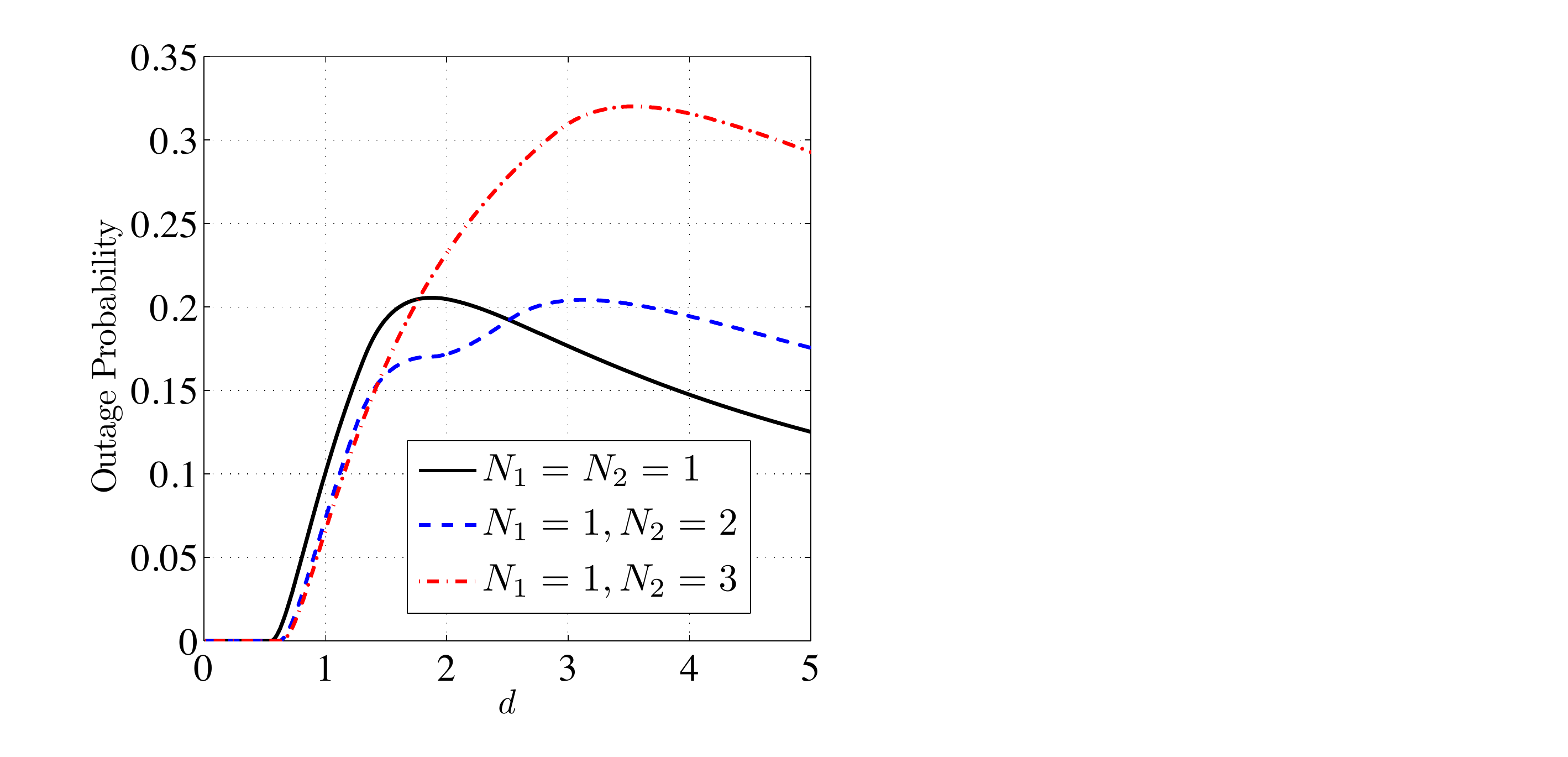}
\label{ref2}
}
\label{fig:subfigureExample}
\caption[Optional caption for list of figures]{Let $k_1=3, k_2=2, q_1=0.3, q_2=0.4, a_1=0.5, a_2=0.7$ and $P_1=P_2=30\,\mathrm{dB}$. If $R_1=0.7\lambda_1=0.63$ and $R_2=0.7\lambda_2=0.56$, then $\mathcal{N}_{R_1,R_2}=\{(m_1,m_2): 1\leq m_1,m_2\leq 2\}$ and (\ref{please1}) is satisfied. Panel~(a) shows the probability of outage in terms of $d$ for different values of $(N_1,N_2)$. If $R_1=0.8\lambda_1=0.72$ and $R_2=0.8\lambda_2=0.64$, then $\mathcal{N}_{R_1,R_2}=\{(m_1,m_2): 1\leq m_1,m_2\leq 3\}$ and (\ref{please1}) is satisfied. It turns out that depending on the value of $d$, the best choices are $(N_1,N_2)=(1,1),(1,2)$ or $(1,3)$ as shown in panel~(b).}
\label{pict_350}
\end{figure} 

\textbf{Example}- Let $k_1=3, k_2=2, q_1=0.3, q_2=0.4, a_1=0.5, a_2=0.7$ and $P_1=P_2=30\,\mathrm{dB}$. We consider several possibilities for $(R_1,R_2)$: 
\begin{itemize}
 \item Let $R_1=0.4\lambda_1=0.36$ and $R_2=0.4\lambda_2=0.32$. Then $\mathcal{N}_{R_1,R_2}=\{(1,1)\}$, however, condition~(\ref{please1}) is not satisfied. This means that by setting $N_1=N_2=1$, both receivers successfully decode the messages regardless of the values of $\nu_1$ and $\nu_2$.  
  \item Let $R_1=0.5\lambda_1=0.45$ and $R_2=0.5 \lambda_2=0.4$. Then $\mathcal{N}_{R_1,R_2}=\{(1,1)\}$ and the condition in (\ref{please1}) is satisfied. We fix $N_1=N_2=1$. Fig.~\ref{pict_390} in panel~(a) shows the region $\bigcup_{S\in\mathscr{S}}(\mathcal{A}^{(\mathrm{geom})}_S\setminus\mathcal{A}^{(\mathrm{rel})}_{S})$. By~Proposition~\ref{prop_r}, if $d<d_{\max}\approx 0.83$, the probability of outage is zero. Fig.~\ref{pict_390} in panel~(b) shows the probability of outage in terms of $d$.  
  \item Let $R_1=0.7\lambda_1=0.63$ and $R_2=0.7\lambda_2=0.56$. Then $\mathcal{N}_{R_1,R_2}=\{(1,1),(1,2),(2,1),(2,2)\}$ and the condition in (\ref{please1}) is satisfied. Fig.~\ref{pict_350} in panel~(a) shows the probability of outage in terms of $d$ for different values of $(N_1,N_2)\in\mathcal{N}_{R_1,R_2}$. If $d<2$, $(N_1,N_2)=(1,2)$ and if $d>2$, $(N_1,N_2)=(1,1)$ are the optimum choices. 
  \item Let $R_1=0.8\lambda_1=0.72$ and $R_2=0.8\lambda_2=0.64$. Then $\mathcal{N}_{R_1,R_2}=\{(m_1,m_2): 1\leq m_1,m_2\leq 3 \}$ and the condition in (\ref{please1}) is satisfied. It turns out that depending on the value of $d$, the best choices are $(N_1,N_2)=(1,1),(1,2)$ or $(1,3)$. Fig.~\ref{pict_350} in panel~(b) shows the probability of outage in terms of $d$ for these values of $(N_1,N_2)$. If $d<1.43$, $(N_1,N_2)=(1,3)$, if $1.43<d<2.51$, $(N_1,N_2)=(1,2)$ and if $d>2.51$, $(N_1,N_2)=(1,1)$ are the best choices.
\end{itemize}

\begin{figure}[t]
\centering
\subfigure[$P_1=P_2=10\,\mathrm{dB}$]{
\includegraphics[scale=0.63]{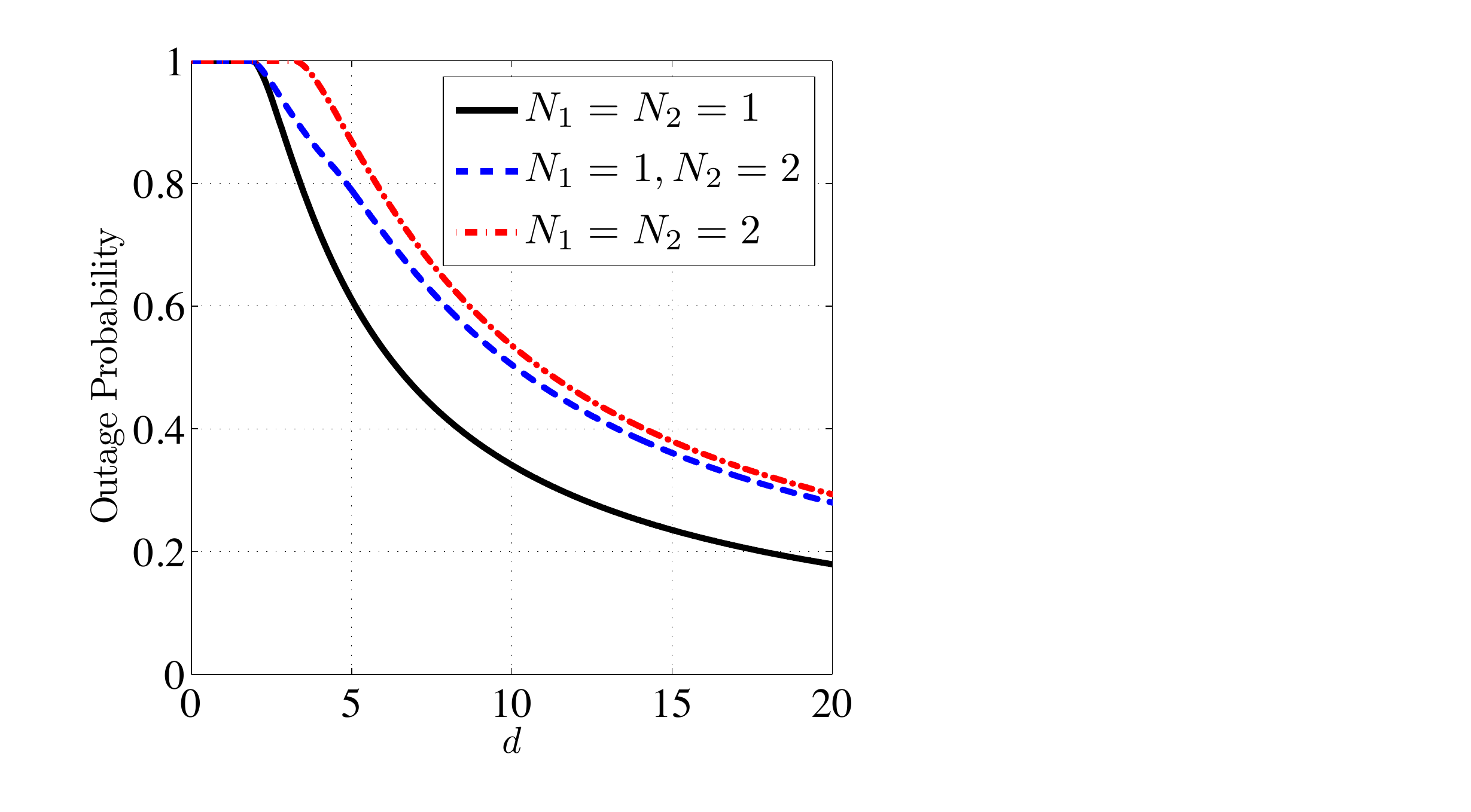}
\label{ref1}
}
\subfigure[$P_1=P_2=30\,\mathrm{dB}$]{
\includegraphics[scale=0.65]{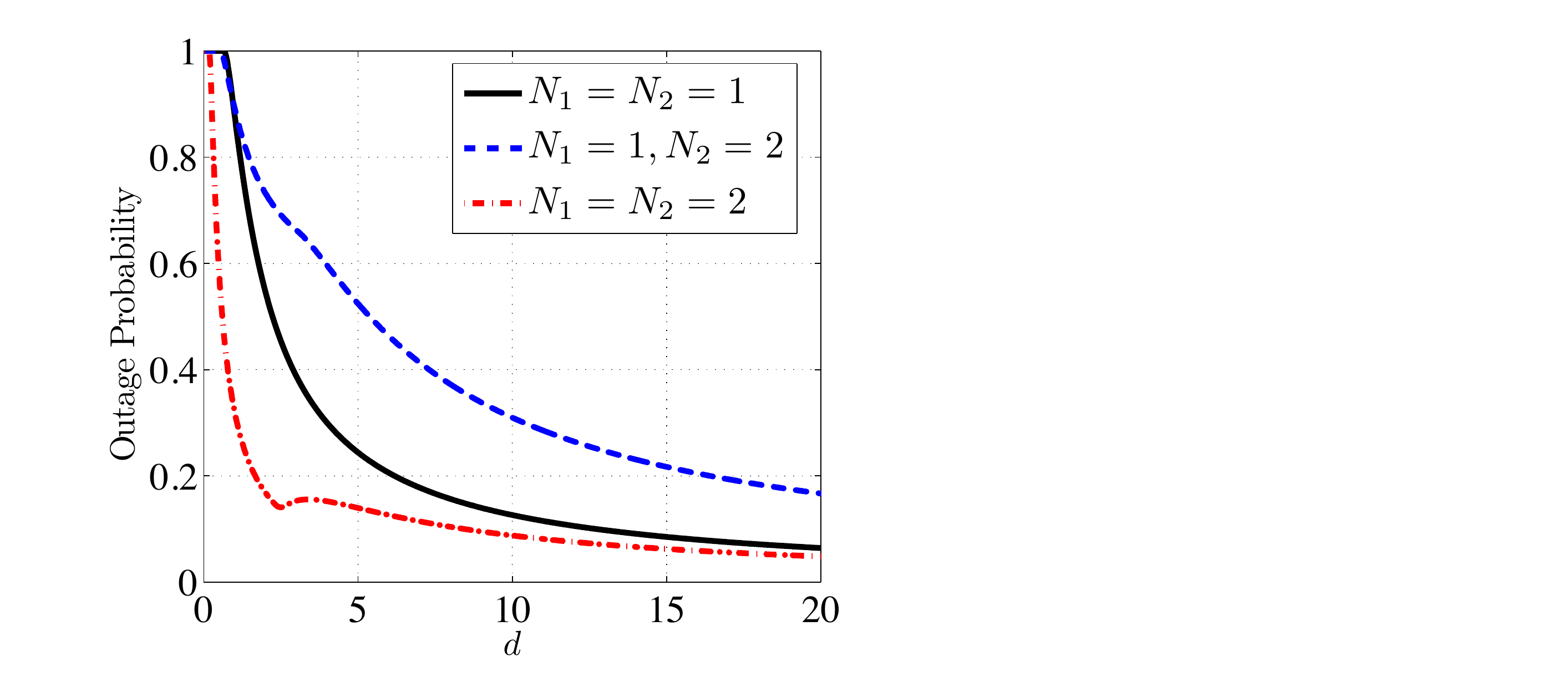}
\label{ref2}
}
\label{fig:subfigureExample}
\caption[Optional caption for list of figures]{A symmetric scenario where $k_1=k_2=5$, $q_1=q_2=0.2$, $a_1=a_2=0.5$, $R_1=0.7\lambda_1=0.7$ and $R_2=0.7\lambda_2=0.7$. Panel~(a) presents the probability of outage in terms of $d$ for different values of $(N_1,N_2)$ for $P_1=P_2=10\,\mathrm{dB}$. We see that  $N_1=N_2=1$ is the optimum choice for any value of $d$. Panel~(b) presents the probability of outage in terms of $d$ for different values of $(N_1,N_2)$ for $P_1=P_2=20\,\mathrm{dB}$. Here,  $N_1=N_2=2$ is the optimum choice for any value of $d$.}
\label{pict_3500}
\end{figure} 

\textbf{Example}- It is possible that $d_{\max}=0$. By Proposition~\ref{prop_r}, this happens when the line $\nu_1=\nu_2$ lies in the interior of the region $\bigcup_{S\in\mathscr{S}}(\mathcal{A}^{(\mathrm{geom})}_S\setminus\mathcal{A}^{(\mathrm{rel})}_S)$ in the $\nu_1$-$\nu_2$~plane. An example of this situation is a symmetric scenario where $k_1=k_2=5$, $q_1=q_2=0.2$, $a_1=a_2=0.5$, $R_1=0.7\lambda_1=0.7$ and $R_2=0.7\lambda_2=0.7$.  We consider two cases for the average transmission power, i.e., $P_1=P_2=10\,\mathrm{dB}$ and $P_1=P_2=30\,\mathrm{dB}$. It turns out that in both cases, $\mathcal{N}_{R_1,R_2}=\{(1,1),(1,2),(2,1),(2,2)\}$ and the condition in (\ref{please1}) is satisfied. 
\begin{itemize}
  \item Let $P_1=P_2=10\,\mathrm{dB}$. Fig.~\ref{pict_3500} in panel~(a) presents the probability of outage in terms of $d$ for different values of $(N_1,N_2)$. Due to symmetry, the cases $N_1=1, N_2=2$ and $N_1=2, N_2=1$ offer the the same performance. We see that  $N_1=N_2=1$ is the optimum choice for any value of $d$.
  \item Let $P_1=P_2=30\,\mathrm{dB}$. Fig.~\ref{pict_3500} in panel~(b) presents the probability of outage in terms of $d$ for different values of $(N_1,N_2)$. In contrast to the case $P_1=P_2=10\,\mathrm{dB}$ in panel~(a), the situation is reversed. Here,  $N_1=N_2=2$ is the optimum choice for any value of $d$.
\end{itemize}
\section{An achievable region for the asynchronous GIC-SDA}
\subsection{The General Model}
In this section we consider a different setting where 
\begin{itemize}
  \item The source of Tx~$i$ no longer turns off after generating a number of $k_in$ bits.
  \item The parameters $\nu_1$ and $\nu_2$ are known at both transmitters\footnote{This requires a certain level of coordination between the two transmitters in order to inform each other about their initial instants of activity.}. As before, we let $\alpha:=\nu_2-\nu_1$. 
\end{itemize}
 Accordingly, we adopt a slightly different notation where we assume Tx~$i$ has a codebook with rate $R_{c,i}$ consisting of $2^{\lfloor n_iR_{c,i}\rfloor}$ codewords of length $n_i=\lfloor n\theta_i\rfloor$ where $\theta_i>0$ is a constant. We pose the following~problem:

\textit{Given positive integers $N_1$ and $N_2$, determine the possible values for the codebook~rates $(R_{c,1},R_{c,2})$ such that
\begin{enumerate}
  \item The first $N_i$ codewords sent by Tx~$i$ are transmitted immediately in the sense defined in Section~II.B and decoded successfully at Rx~$i$.
  \item The average transmission power for Tx~$i$ satisfies (\ref{power_11}) where $\mathcal{T}_i$ is the period of activity for Tx~$i$ until the time slot it transmits the last symbol in its $N_i^{th}$ burst.
\end{enumerate}
}
Since the codewords are transmitted immediately, the content of the buffer of Tx~$i$ never exceeds $\lfloor n_iR_{c,i}\rfloor+k_i$ and therefore, the buffers are stable. The first $N_i$ bursts sent by Tx~$i$ represent a total number of $N_i\lfloor n_iR_{c,i}\rfloor$ bits. Therefore, the average transmission rate for Tx~$i$ is $R_i=\frac{N_i\lfloor n_iR_{c,i}\rfloor}{|\mathcal{T}_i|}$. Following similar steps in Section~II.D, 
\begin{eqnarray}
\label{rate_region_11}
R_i=\frac{N_iR_{c,i}}{1+\frac{N_iR_{c,i}}{\lambda_i}},
\end{eqnarray}
in the asymptote of large $n$. Similarly, the average transmission power for Tx~$i$ is 
\begin{eqnarray}
Q_i=\frac{N_i\gamma_i}{1+\frac{N_iR_{c,i}}{\lambda_i}}.
\end{eqnarray}
Then the average power constraint $Q_i\leq P_i$ in~(\ref{power_11}) becomes
\begin{equation}
\label{power_33}
0\le\gamma_i\le \Big(\frac{1}{N_i}+ \frac{R_{c,i}}{\lambda_i}\Big)P_i.
\end{equation}
Before proceeding further, let us reiterate the major differences between the current setup and the setup in the previous section: 
\begin{itemize}
  \item We show the codebook~rate of Tx~$i$ by $R_{c,i}$. No notation was selected for the codebook~rate $\frac{\eta_i}{\theta_i}$ in the previous section. All the achievability results in Propositions~\ref{prop_4},~\ref{prop_5}~and~\ref{prop_6} remain valid after replacing $\eta_i$ by $\theta_iR_{c,i}$.
  \item $N_i$ and $\theta_i$ are constants, while they served as design parameters in the previous section. 
 \item The parameters $\nu_1$ and $\nu_2$ are known at both transmitters, while $\nu_i$ was unknown to Tx~$i'$ in the previous section. 
  \item The information source at Tx~$i$ turns on at time slot $\lfloor n\nu_i\rfloor$ and remains active indeterminately.
\end{itemize}

We aim to characterize a region $\mathcal{R}$ of all codebook~rate tuples $(R_{c,1},R_{c,2})$ such that both transmitters send their codewords immediately and reliably and such that the power constraints in (\ref{power_11}) are not violated. We call $\mathcal{R}$ the achievable~(codebook)~rate~region. One can also define an achievable rate region $\mathcal{R}'$ of all transmission rate tuples $(R_1,R_2)$ such that the aforementioned properties hold. Since $\mathcal{R}'$ and $\mathcal{R}$ are related through the mappings in (\ref{rate_region_11}), we only focus on $\mathcal{R}$. 
  Immediate transmission of a scheduled codeword is impossible if a previously scheduled codeword is not fully transmitted. By~(\ref{buffstab}), immediate transmission of the codewords is guaranteed if $\frac{\theta_i R_{c,i}}{\lambda_i}>\theta_i$ for $N_i>1$, i.e., 
\begin{equation}
\label{rate_value1}
R_{c,i}>\lambda_i\mathds{1}_{N_i>1},\,\,\,i=1,2.
\end{equation}
An achievable $R_{c,i}$ must satisfy\footnote{By~Corollary~\ref{coro_1}, we must have $\theta_iR_{c,i}<\theta_i\phi_i$ which simplifies to (\ref{nav_11}).}
\begin{equation}
\label{nav_11}
R_{c,i}<\frac{1}{2}\log(1+\gamma_i).
\end{equation}
Combining~(\ref{power_33}) and (\ref{nav_11}), 
\begin{equation}
\label{ }
R_{c,i}<\frac{1}{2}\log\Big(1+\Big(\frac{1}{N_i}+\frac{R_{c,i}}{\lambda_i}\Big)P_i\Big).
\end{equation}
This is equivalent to
\begin{equation}
\label{bo_22}
R_{c,i}<\overline{R}_{c,i},\,\,\,\,\,i=1,2,
\end{equation}
where $\overline{R}_{c,i}$ is the unique positive solution\footnote{Define $f(x)=x-\frac{1}{2}\log(1+(\frac{1}{N_i}+\frac{x}{\lambda_i})P_i)$. Note that $f(0)=-\frac{1}{2}\log(1+\frac{P_i}{N_i})<0$, $\lim_{x\to\infty}f(x)=\infty$ and $f'(x)=1-\frac{1}{2\ln 2}\frac{\frac{P_i}{\lambda_i}}{1+(\frac{1}{N_i}+\frac{x}{\lambda_i})P_i}$. If $f'(x)> 0$ for all $x>0$, then $f(x)=0$ has only one positive solution. If there is $x_0>0$ such that $f'(x_0)=0$, then $f$ is decreasing over $(0,x_0)$ and increasing over $(x_0,\infty)$. This again implies that $f(x)=0$ has only one positive solution.} for $R_{c,i}$ in the equation $R_{c,i}=\frac{1}{2}\log(1+(\frac{1}{N_i}+\frac{R_{c,i}}{\lambda_i})P_i)$. By~(\ref{rate_value1}) and~(\ref{bo_22}), $\mathcal{R}$ lies inside the rectangle $[\lambda_1\mathds{1}_{N_1>1},\overline{R}_{c,1}]\times [\lambda_2\mathds{1}_{N_2>1},\overline{R}_{c,2}]$.

 \begin{figure}[t]
  \centering
  \includegraphics[scale=1] {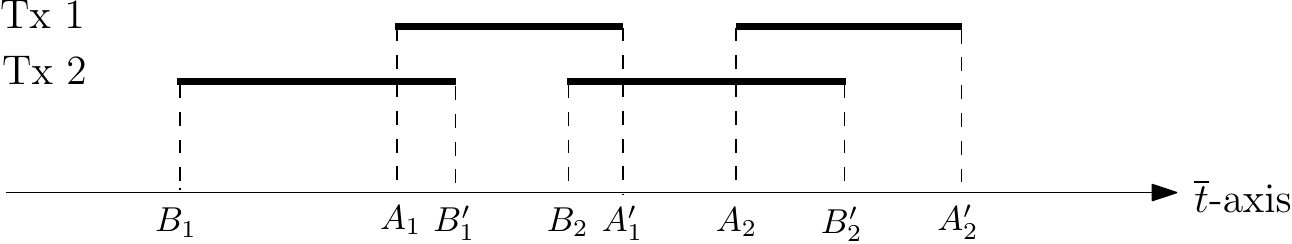}
  \caption{This picture shows the positions of different bursts on the $\overline{t}$-axis corresponding to the state $S=\{(1;1,2),(2;2,4)\}$ in a scenario where $N_1=N_2=2$. The table in~(\ref{table2}) shows the numbers on the $\overline{t}$-axis corresponding to different points on the $\overline{t}$-axis.}
  \label{blow_1}
 \end{figure}
 \begin{figure*}
 \begin{eqnarray}
 \centering
\begin{tabular}{|c|c|c|c|}
\hline
&\\
   $A_m=\frac{m\theta_1R_{c,1}}{\lambda_1}+\nu_1$ & $B_m=\frac{m\theta_2R_{c,2}}{\lambda_2}+\nu_2$ \\
   &\\
   \hline 
   &\\
   $A'_m=\frac{m\theta_1R_{c,1}}{\lambda_1}+\nu_1+\theta_1$ & $B'_m=\frac{m\theta_2R_{c,2}}{\lambda_2}+\nu_2+\theta_2$ \\
   &\\
\hline
\end{tabular}
\label{table2}
\end{eqnarray}
\end{figure*}

To describe $\mathcal{R}$, we need the concept of the state introduced in Section~V.A for a GIC-SDA with immediate transmissions. For each state $S$, we impose two sets of constraints on $(R_{c,1},R_{c,2})$, i.e., the geometric constraints and the reliability constraints shown by $\mathcal{R}^{(\mathrm{geom})}_S$ and $\mathcal{R}^{(\mathrm{rel})}_{S}$, respectively. To describe these constraints, let us consider the situation shown in Fig.~\ref{blow_1} where $N_1=N_2=2$ and the state of the channel is $S=\{(1;1,2),(2;2,4)\}$. This is only one of $\binom{8}{4}=70$ possible states. The table in~(\ref{table2}) shows the numbers on the $\overline{t}$-axis corresponding to different points in Fig~\ref{blow_1}. 
\begin{itemize}
  \item The geometric constraints are imposed by the positions of the bursts along the $\overline{t}$-axis. For example, point $B_1$ is on left of point $A_1$ which gives $\frac{\theta_2R_{c,2}}{\lambda_2}+\nu_2<\frac{\theta_1R_{c,1}}{\lambda_1}+\nu_1$. A complete list of the geometric constraints is given by the polyhedron 
  \begin{eqnarray}
\label{const_set1}
\begin{bmatrix}
     -\frac{\theta_1}{\lambda_1} & \frac{\theta_2}{\lambda_2}   \\
     \frac{\theta_1}{\lambda_1}  & -\frac{\theta_2}{\lambda_2}   \\
     -\frac{\theta_1}{\lambda_1}  & \frac{2\theta_2}{\lambda_2}   \\
     \frac{2\theta_1}{\lambda_1} & -\frac{2\theta_2}{\lambda_2}  \\
     -\frac{2\theta_1}{\lambda_1} &\frac{2\theta_2}{\lambda_2} \\
        -1 &0\\
     0 &-1
\end{bmatrix}\begin{bmatrix}
      R_{c,1}    \\
      R_{c,2} 
\end{bmatrix}<\begin{bmatrix}
      -\alpha \\
      \theta_2+\alpha\\
      \theta_1-\alpha\\
      \theta_2+\alpha\\
     \theta_1-\theta_2-\alpha\\
         -\lambda_1\\
      -\lambda_2
 \end{bmatrix},
\end{eqnarray}
 i.e.,  $\mathcal{R}^{(\mathrm{geom})}_S$ for $S=\{(1;1,2),(2;2,4)\}$ in the set of all $(R_{c,1},R_{c,2})$ such that (\ref{const_set1}) holds.  The last two constraints in (\ref{const_set1}) are the inequalities in (\ref{rate_value1}).
  \item The reliability constraints guarantee successful decoding for all $N_1+N_2=2+2=4$ transmitted codewords subject to the power conditions in~(\ref{power_33}).  For example, the second codeword of Tx~1 in Fig.~\ref{blow_1} only interferes with the second codeword of Tx~2 on its ``left~end'', i.e., $\omega_{2,2}^-=2$, $\omega_{2,2}^+=0$, $\omega_{2,2}=0$. We invoke Proposition~\ref{prop_5} to write 
  \begin{eqnarray}
\big(1-\frac{2}{\lambda_1}(\phi_1-\psi_1)\big)\theta_1R_{c,1}+\frac{2}{\lambda_2}(\phi_1-\psi_1)\theta_2R_{c,2}<\theta_1\phi_1-(\alpha+\theta_2)(\phi_1-\psi_1).
\end{eqnarray}
A complete list of reliability constraints is given by the polyhedra
\begin{eqnarray}
\label{const_set22}
\begin{bmatrix}
   \theta_1   & -\frac{1}{\lambda_2}(\phi_1-\psi_1)\theta_2   \\
      (1-\frac{2}{\lambda_1}(\phi_1-\psi_1))\theta_1 & \frac{2}{\lambda_2}(\phi_1-\psi_1)\theta_2  \\
      -\frac{1}{\lambda_1}(\phi_2-\psi_2)\theta_1 &(1+\frac{1}{\lambda_2}(\phi_2-\psi_2))\theta_2\\
      -\frac{1}{\lambda_1}(\phi_2-\psi_2)\theta_1 & \theta_2\\
      -1 & 0\\
      0 &-1
\end{bmatrix}\begin{bmatrix}
      R_{c,1}   \\
      R_{c,2} 
\end{bmatrix}<\begin{bmatrix}
      \theta_1\psi_1-\theta_2(\phi_1-\psi_1)    \\ 
       \theta_1\phi_1-(\alpha+\theta_2)(\phi_1-\psi_1)\\
       \theta_2\psi_2-\alpha(\phi_2-\psi_2)\\
       \theta_2\psi_2-\theta_1(\phi_2-\psi_2)\\
       \lambda_1(\frac{1}{2}-\frac{\gamma_1}{P_1})\\
        \lambda_2(\frac{1}{2}-\frac{\gamma_2}{P_2})
\end{bmatrix},
\end{eqnarray}
for some $\gamma_1,\gamma_2\geq 0$, i.e., $\mathcal{R}^{(\mathrm{rel})}_S$ for $S=\{(1;1,2),(2;2,4)\}$ is the set of all $(R_{c,1},R_{c,2})$ such that~(\ref{const_set22}) holds for some $\gamma_1,\gamma_2\geq0$. The last two constraints in (\ref{const_set22}) are the inequalities in~(\ref{power_33}). Note that $\mathcal{R}^{(\mathrm{rel})}_S$ is the union of infinitely many polyhedra. More precisely, if we denote the polyhedron in (\ref{const_set22}) for fixed $\gamma_1,\gamma_2\ge0$ by $\mathcal{P}_S(\gamma_1,\gamma_2)$, then 
\begin{equation}
\label{gen_11}
\mathcal{R}^{(\mathrm{rel})}_S=\bigcup_{\gamma_1,\gamma_2\geq 0}\mathcal{P}_S(\gamma_1,\gamma_2).
\end{equation}
 It is needless to mention that $\phi_i$ and $\psi_{i}$ are functions of $\gamma_1,\gamma_2$.
\end{itemize}
Having $\mathcal{R}_S^{\mathrm{(geom)}}$ and $\mathcal{R}_S^{\mathrm{(rel)}}$ defined for any state $S$, the achievable rate region~$\mathcal{R}$ is given by
\begin{eqnarray}
\label{region_11}
\mathcal{R}=\bigcup_{S\in\mathscr{S}}\big(\mathcal{R}^{(\mathrm{geom})}_S\bigcap\mathcal{R}^{(\mathrm{rel})}_{S}\big).
\end{eqnarray} 
 \begin{figure}[t]
\centering
\subfigure[]{
\includegraphics[scale=0.67]{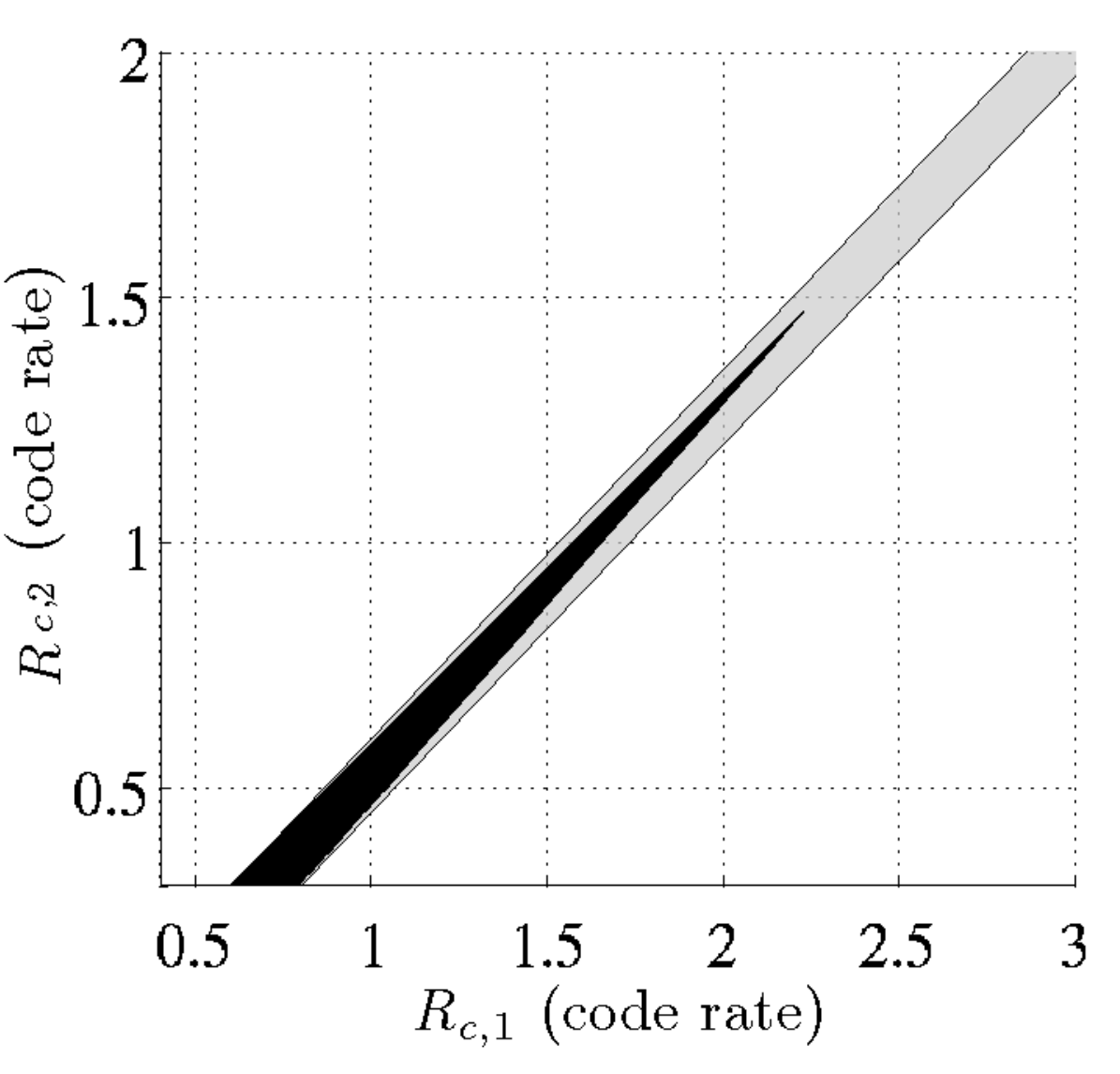}
\label{ref1}
}
\subfigure[]{
\includegraphics[scale=0.67]{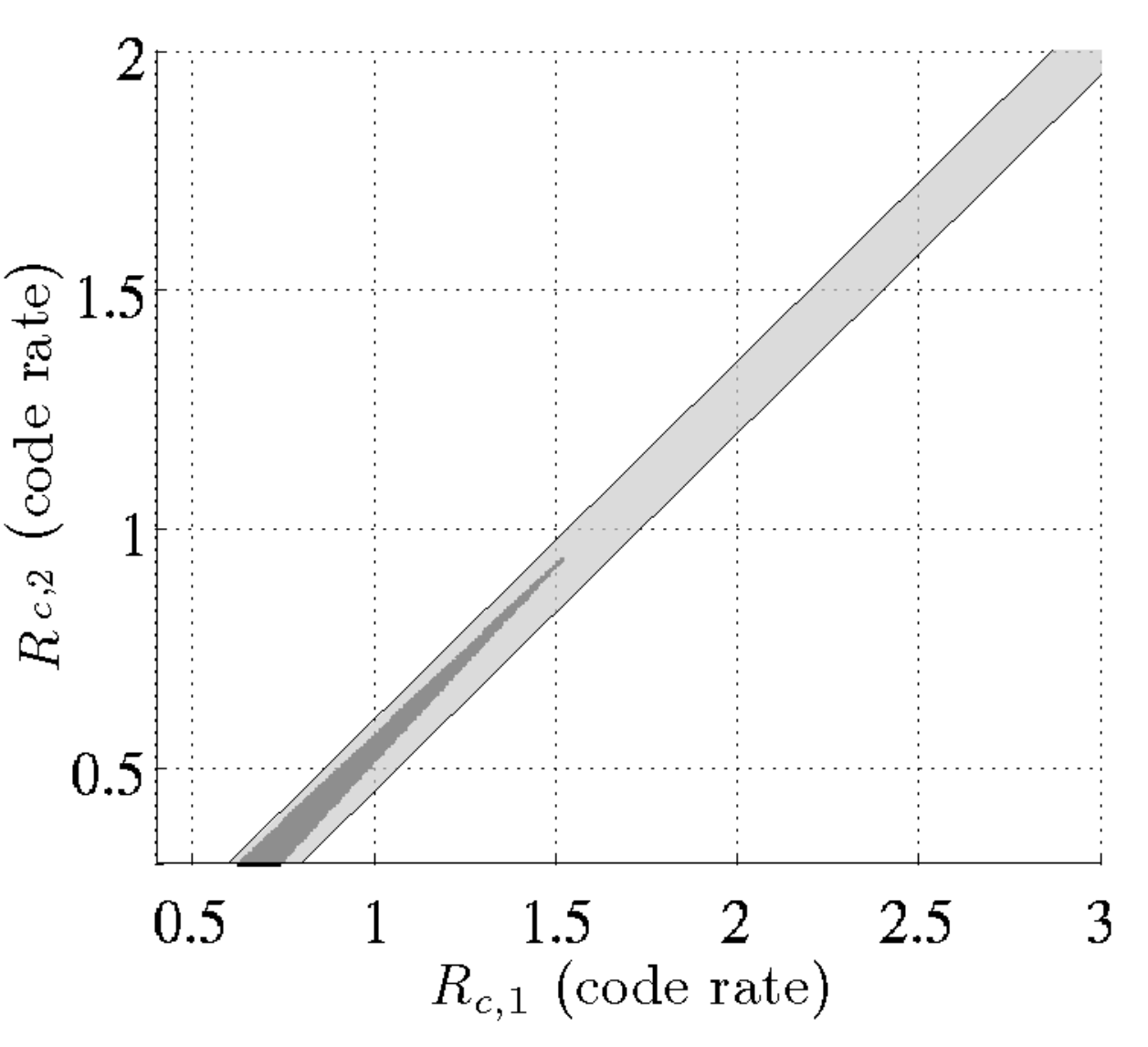}
\label{ref2}
}
\label{fig:subfigureExample}
\caption[Optional caption for list of figures]{Consider a setup where $N_1=3, N_2=2, \theta_1=\theta_2=1, k_1=2, k_2=3, q_1=0.2, q_2=0.1, a_1=1.5, a_2=0.5, P_1=20\,\mathrm{dB}, P_2=30\,\mathrm{dB}$ and $\alpha=1$. Panel~(a) shows the regions $\mathcal{R}^{(\mathrm{geom})}_S$ in grey and $\mathcal{R}^{(\mathrm{rel})}_S$ in black for $S=\{(1;2,3),(2;3,4)\}$. Panel~(b) shows the same regions under full~power transmission. It is seen that $\mathcal{R}^{(\mathrm{rel})}_S$ under full~power transmission is strictly smaller than $\mathcal{R}^{(\mathrm{rel})}_S$ in its general form given in~(\ref{gen_11}). }
\label{pict_1}
\end{figure}   
 \begin{figure}[t]
\centering
\subfigure[]{
\includegraphics[scale=0.75]{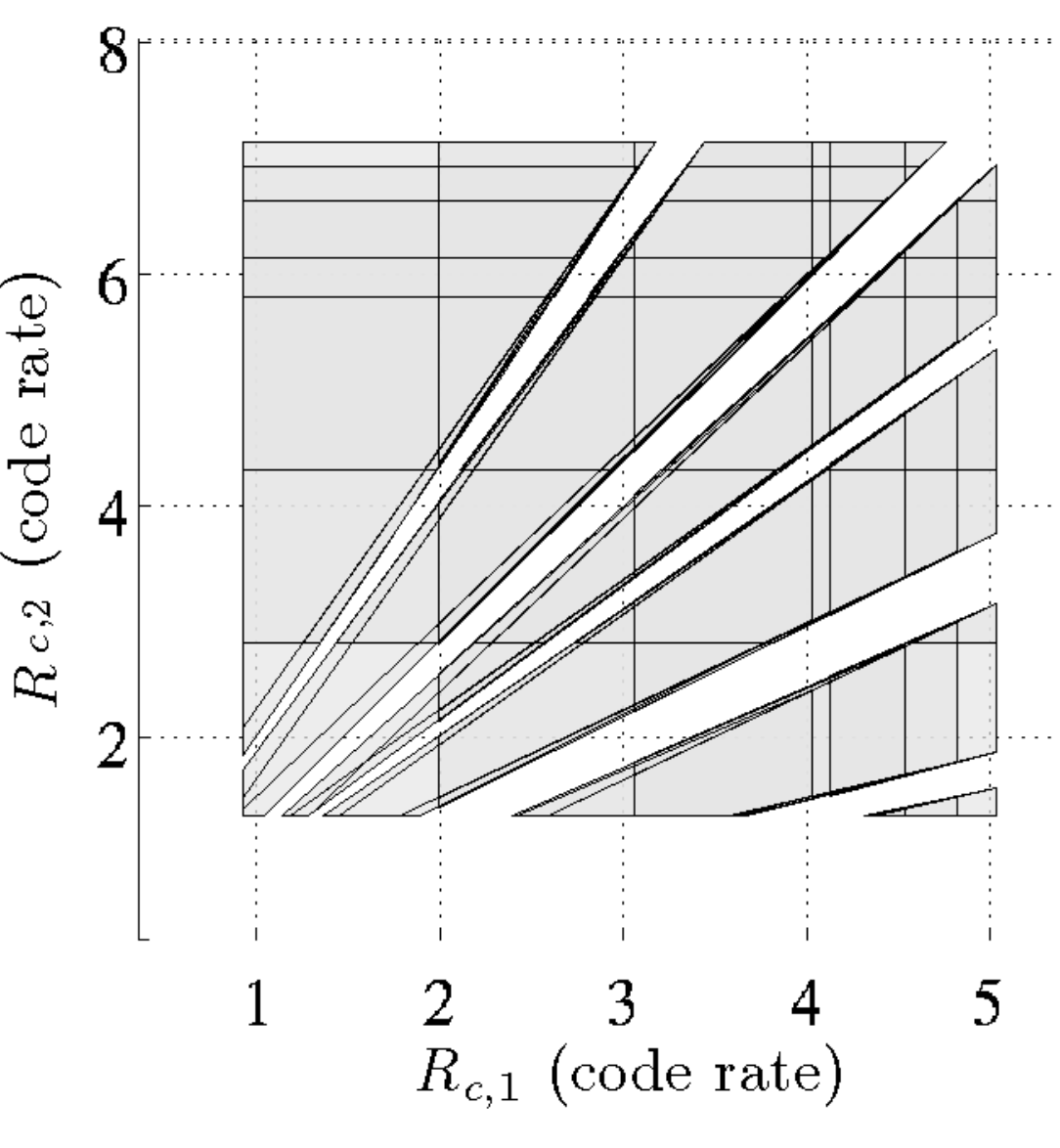}
\label{ref1}
}
\subfigure[]{
\includegraphics[scale=0.70]{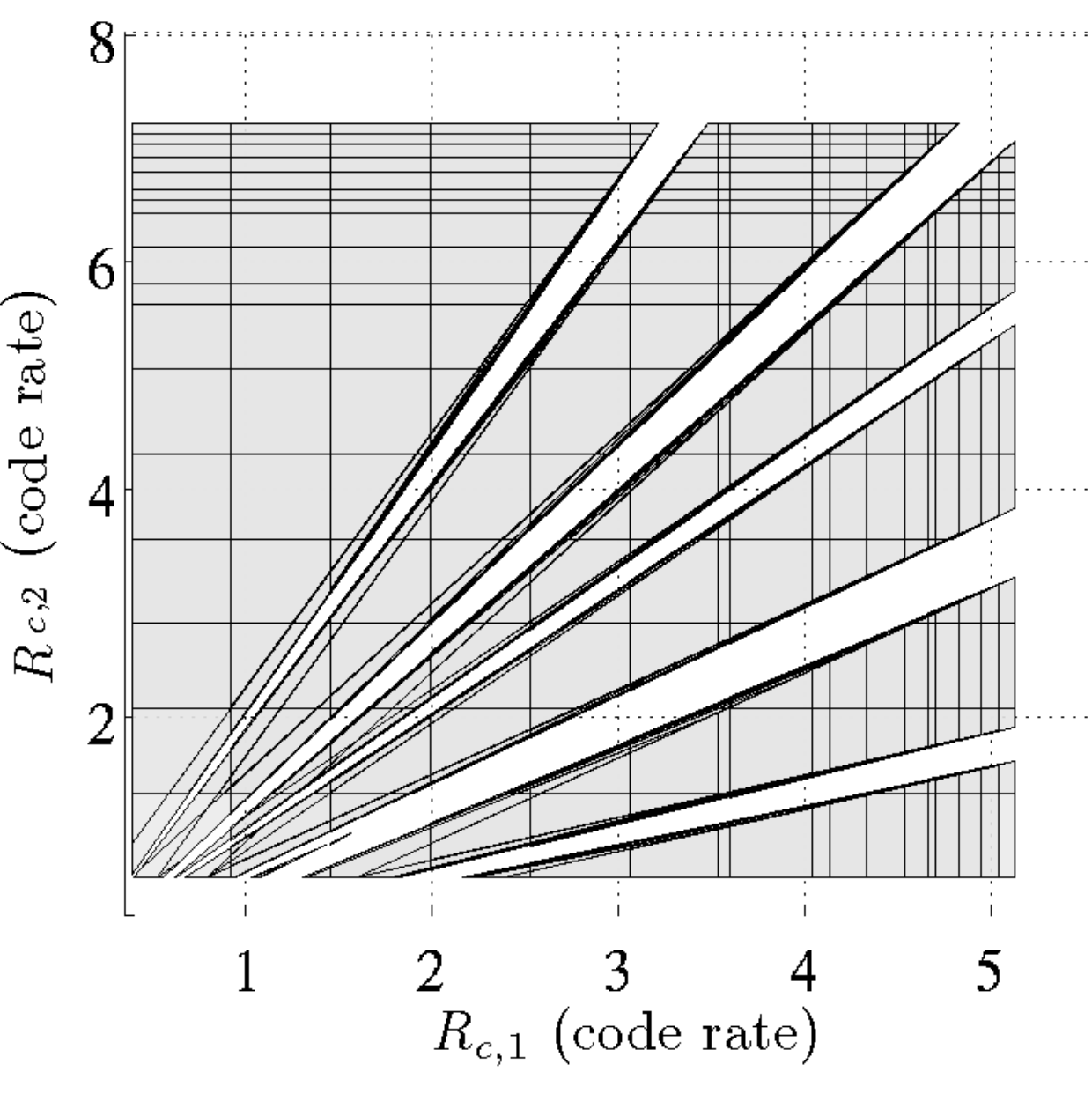}
\label{ref2}
}
\label{fig:subfigureExample}
\caption[Optional caption for list of figures]{Consider a setup where $N_1=3, N_2=2, \theta_1=\theta_2=1, k_1=2, k_2=3, q_1=0.2, q_2=0.1, a_1=1.5, a_2=0.5, P_1=20\,\mathrm{dB}, P_2=30\,\mathrm{dB}$ and $\alpha=1$. Panel~(a) shows the region $\widetilde{\mathcal{R}}^{(\mathrm{rel})}_S$ in~(\ref{pil_11}) where $S=\{(1;2,3),(2;3,4)\}$ and $\Gamma_i=\{\frac{l}{5}\overline{\gamma}_i: 1\leq l\leq 4\}$. Panel~(b) shows the same region for $\Gamma_i=\{\frac{l}{10}\overline{\gamma}_i: 1\leq l\leq 9\}$.}
\label{pict_123}
\end{figure}   
\begin{figure}[t]
\centering
\subfigure[$N_1=N_2=1$]{
\includegraphics[scale=0.69]{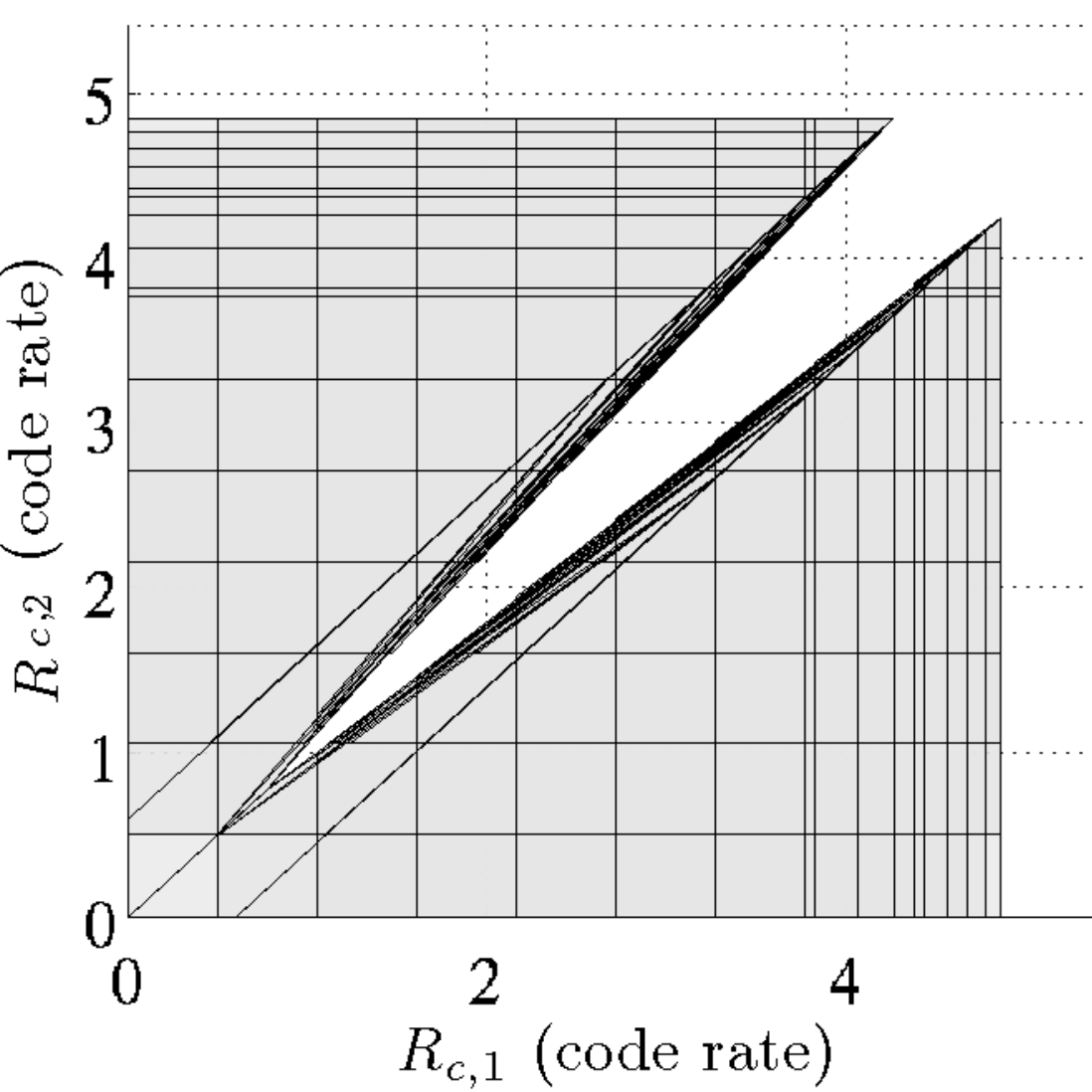}
\label{ref1}
}
\subfigure[$N_1=N_2=2$]{
\includegraphics[scale=0.73]{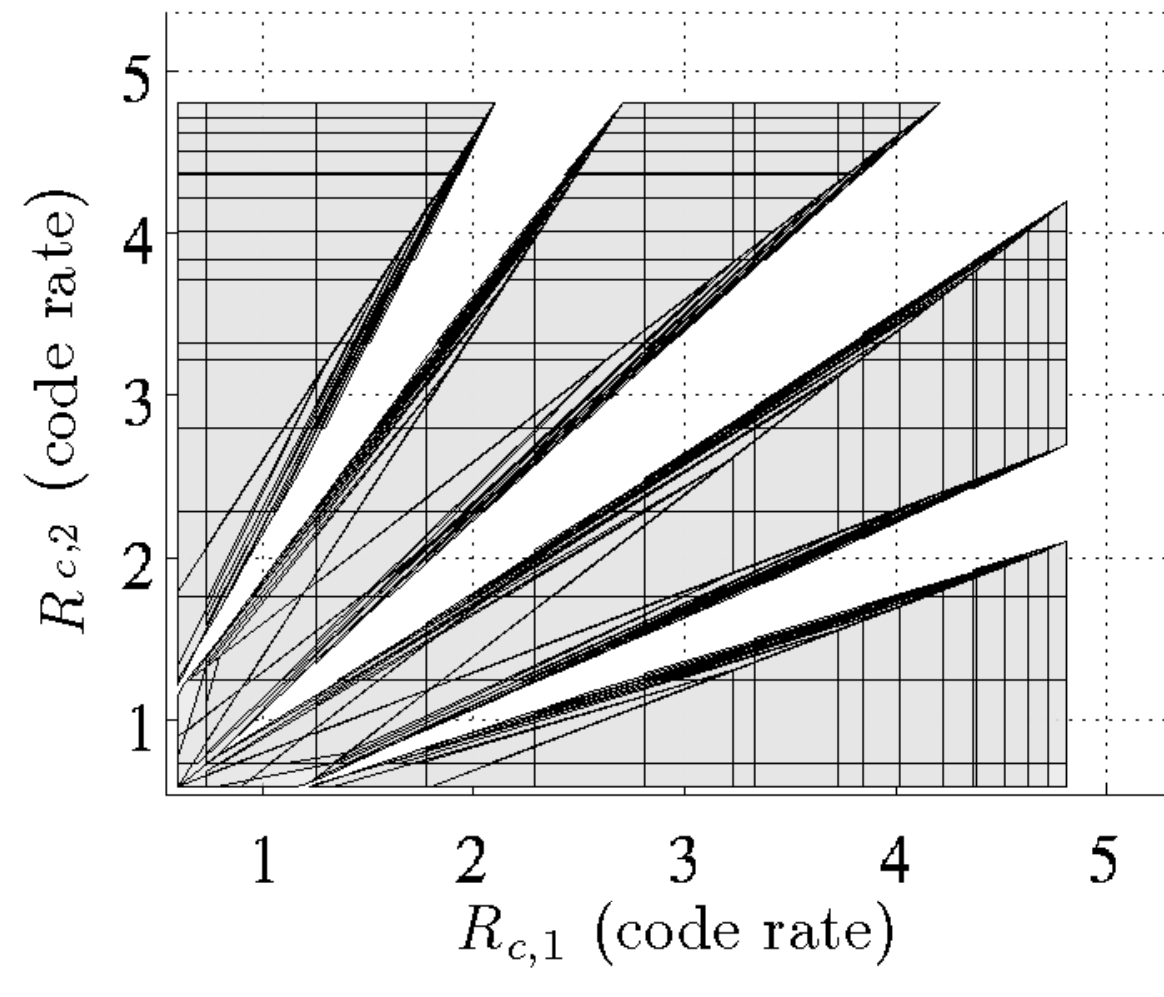}
\label{ref2}
}
\subfigure[$N_1=N_2=3$]{
\includegraphics[scale=0.77]{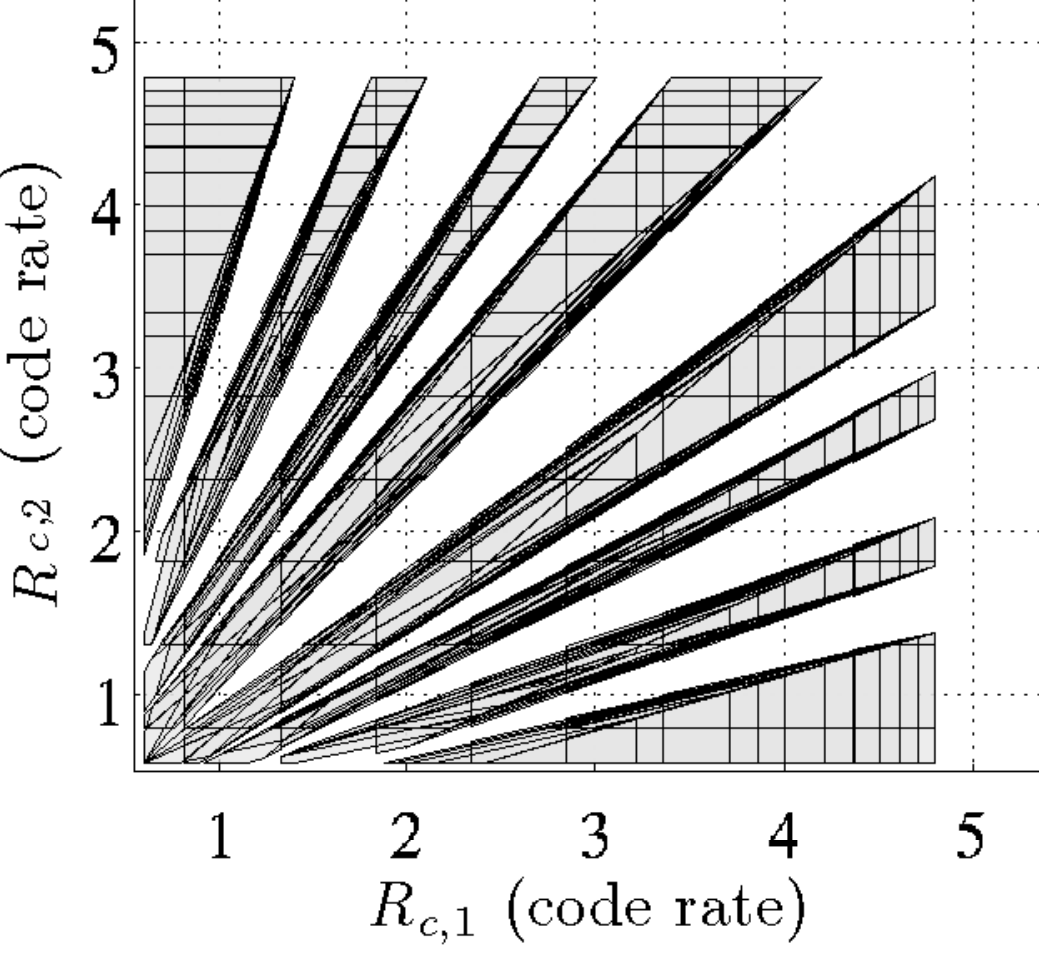}
\label{ref1}
}
\subfigure[$N_1=N_2=4$]{
\includegraphics[scale=0.73]{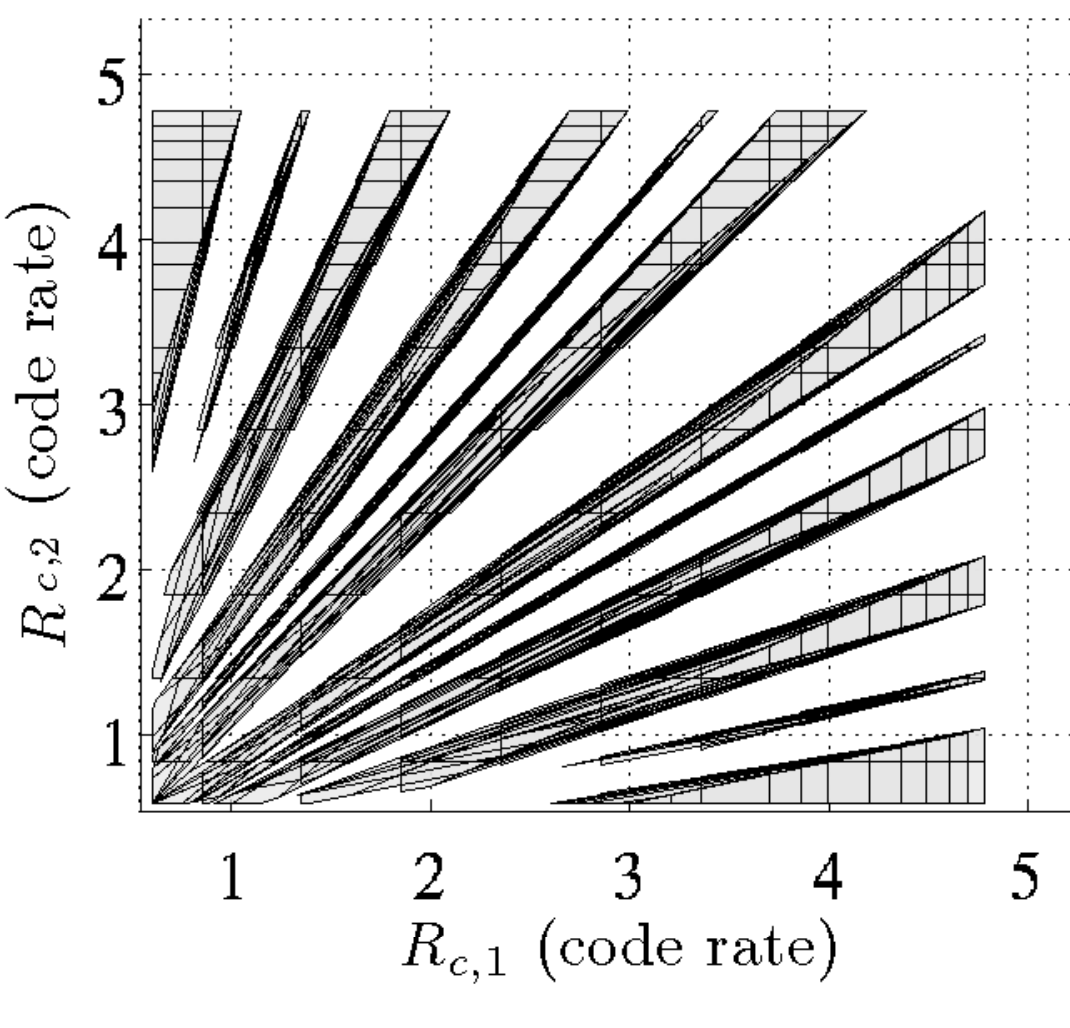}
\label{ref1}
}
\label{fig:subfigureExample}
\caption[Optional caption for list of figures]{A setting where $\theta_1=\theta_2=1, k_1=k_2=2, q_1=q_2=0.3, a_1=a_2=0.5, P_1=P_2=20\,\mathrm{dB}$ and $\alpha=0$. As the number of codewords $N_1=N_2$ increases, $\mathcal{R}$ becomes strictly smaller.   }
\label{pict_2}
\end{figure}  
A few remarks are in order:
\begin{itemize}
  \item In general, none of $\mathcal{R}_S^{(\mathrm{geom})}$ and $\mathcal{R}_S^{(\mathrm{rel})}$ is a subset of the other. 
  \item Depending on the system parameters, there might exist states $S$ such that $\mathcal{R}^{(\mathrm{geom})}_S\bigcap\mathcal{R}^{(\mathrm{rel})}_{S}=\emptyset$.
  \item Full~power transmission, i.e., $\gamma_i=(\frac{1}{N_i}+\frac{R_{c,i}}{\lambda_i})P_i$ is not in general optimum. For example, let $N_1=3$, $N_2=2$, $\theta_1=\theta_2=1, k_1=2, k_2=3, q_1=0.2, q_2=0.1, a_1=1.5, a_2=0.5, P_1=20\,\mathrm{dB}, P_2=30\,\mathrm{dB}$ and $\alpha=1$. Fig.~\ref{pict_1} in panel~(a) shows the regions $\mathcal{R}^{(\mathrm{geom})}_S$ in grey and $\mathcal{R}^{(\mathrm{rel})}_S$ in black for $S=\{(1;2,3),(2;3,4)\}$. Fig.~\ref{pict_1} in panel~(b) shows the same regions under full~power transmission. It is seen that $\mathcal{R}^{(\mathrm{rel})}_S$ under full~power transmission is strictly smaller than $\mathcal{R}^{(\mathrm{rel})}_S=\bigcup_{\gamma_1,\gamma_2\geq 0}\mathcal{P}_S(\gamma_1,\gamma_2)$. 
  \item In order to plot $\mathcal{R}_S^{(\mathrm{rel})}$ for a given state $S$, we choose a finite set of values for $\gamma_i$, namely $\Gamma_i$, and approximate $\mathcal{R}_S^{(\mathrm{rel})}$ by 
\begin{equation}
\label{pil_11}
\widetilde{\mathcal{R}}_S^{(\mathrm{rel})}:=\bigcup_{\gamma_1\in \Gamma_1,\gamma_2\in \Gamma_2}\mathcal{P}_S(\gamma_1,\gamma_2)\subseteq \mathcal{R}_S^{(\mathrm{rel})}.
\end{equation} 
To choose $\Gamma_i$, we observe that 
\begin{eqnarray}
 0\leq \gamma_i<\overline{\gamma}_i:=\Big(\frac{1}{N_i}+\frac{\overline{R}_{c,i}}{\lambda_i}\Big)P_i,
 \end{eqnarray}
 due to (\ref{power_33}) and (\ref{bo_22}).  Fix a natural number $m$ and let 
 \begin{equation}
\label{gamma_11}
\Gamma_i:=\Big\{\frac{l}{m}\overline{\gamma}_i: 1\leq l\leq m-1\Big\}.
\end{equation}
 The set difference $\mathcal{R}_S^{(\mathrm{rel})}\setminus\widetilde{\mathcal{R}}_S^{(\mathrm{rel})}$ becomes smaller as $m$ increases. For example, Fig.~\ref{pict_123} shows the region $\widetilde{R}_{S}^{(\mathrm{rel})}$ for $S=\{(1;2,3),(2;3,4)\}$ in a setup where $N_1=3, N_2=2, \theta_1=\theta_2=1, k_1=2, k_2=3, q_1=0.2, q_2=0.1, a_1=1.5, a_2=0.5, P_1=20\,\mathrm{dB}, P_2=30\,\mathrm{dB}$ and $\alpha=1$. In panel~(a), we have $\Gamma_i=\{\frac{l}{5}\overline{\gamma}_i: 1\leq l\leq 4\}$ and in panel~(b), $\Gamma_i=\{\frac{l}{10}\overline{\gamma}_i: 1\leq l\leq 9\}$. 
\end{itemize}

In the next two examples, we fix $\theta_1=\theta_2=1, k_1=k_2=2, q_1=q_2=0.3, a_1=a_2=0.5$, $P_1=P_2=20\,\mathrm{dB}$ and study the effects of $N_1,N_2$ and $\alpha$ on $\mathcal{R}$ in (\ref{region_11}). We also fix $\Gamma_i=\{\frac{l}{10}\overline{\gamma}_i: 1\leq l\leq 9\}$. 

\textbf{Example}- Let $\alpha=0$, i.e., both users become active at the same time. Fig.~\ref{pict_2} shows the region $\mathcal{R}$ for different values of $N_1=N_2$. As the number of transmitted codewords increases, $\mathcal{R}$ becomes strictly smaller.  

\begin{figure}[t]
\centering
\subfigure[$\alpha=5$]{
\includegraphics[scale=0.75]{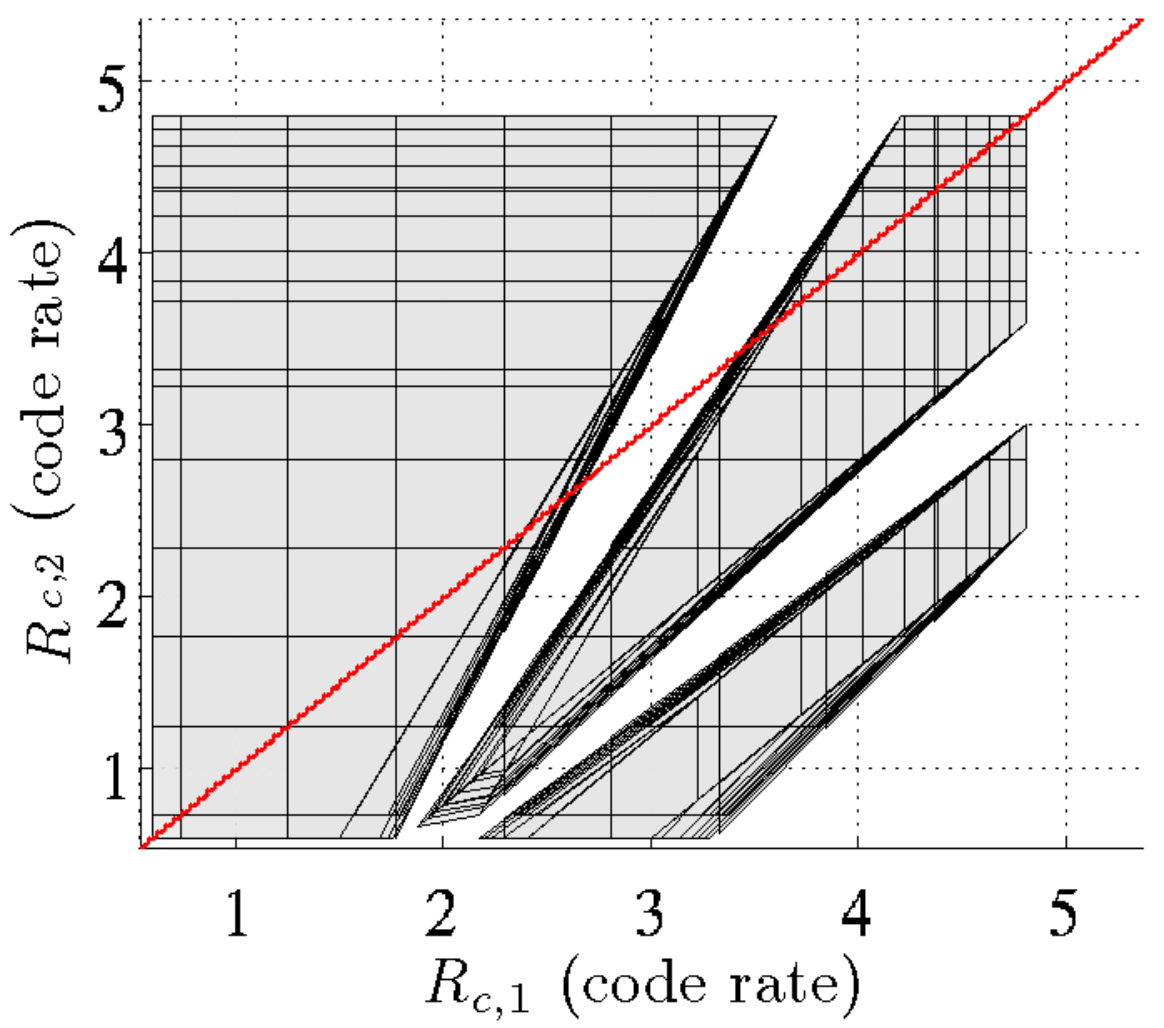}
\label{ref1}
}
\subfigure[$\alpha=10$]{
\includegraphics[scale=0.75]{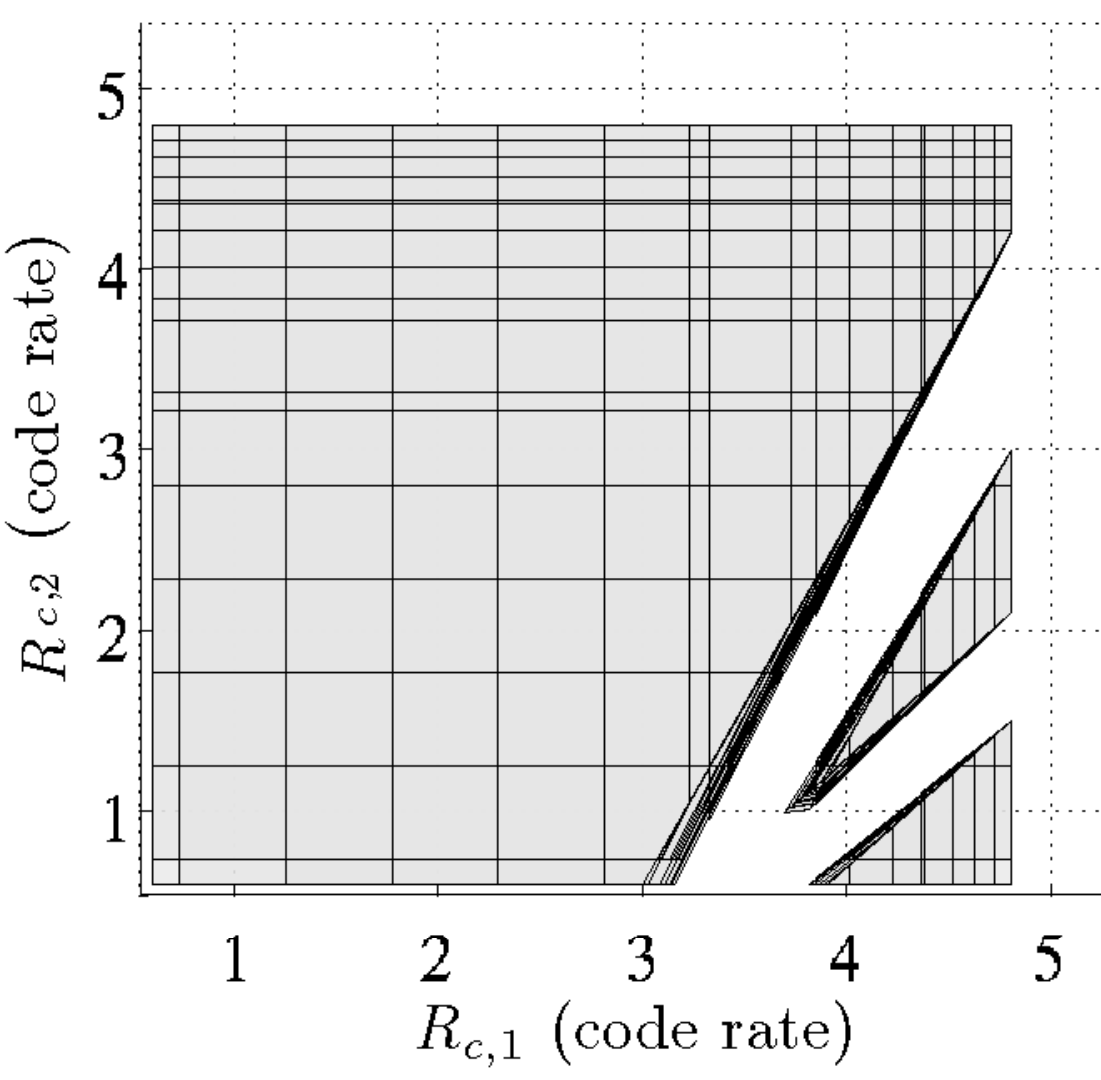}
\label{ref2}
}
\label{fig:subfigureExample}
\caption[Optional caption for list of figures]{A setting where $N_1=N_2=2, \theta_1=\theta_2=1, k_1=k_2=2, q_1=q_2=0.3, a_1=a_2=0.5$ and $P_1=P_2=20\,\mathrm{dB}$. As $\alpha$ increases, the region $\mathcal{R}$ converges to the square $\{(R_{c,1},R_{c,2}): \lambda_i<R_{c,i}<\overline{R}_{c,i},\,i=1,2\}$ where~$\overline{R}_{c,1}=\overline{R}_{c,2}\approx4.8774$. }
\label{pict_3}
\end{figure} 

\textbf{Example}- Let $N_1=N_2=2$. Fig.~\ref{pict_3} shows the region $\mathcal{R}$ for different values of $\alpha$. As $\alpha$ increases, the region $\mathcal{R}$ converges to the square $\{(R_{c,1},R_{c,2}): \lambda_i<R_{c,i}<\overline{R}_{c,i},,\,i=1,2\}$ where~$\overline{R}_{c,1}=\overline{R}_{c,2}\approx4.8774$.

\subsection{The Symmetric Model}
In this section we study a symmetric setting where except for $\nu_1$ and $\nu_2$, other system parameters for the two users are identical. In this case, we drop the index~$i=1,2$, i.e., $N_1=N_2=N$, $\lambda_1=\lambda_2=\lambda$, $R_{c,1}=R_{c,2}=R_c$, $\theta_1=\theta_2=\theta$,  $a_1=a_2=a$, $\gamma_1=\gamma_2=\gamma$ and $P_1=P_2=P$. Without loss of generality,  
 \begin{equation}
\label{ }
 \nu_2\geq \nu_1.
\end{equation}
Let\footnote{More precisely, $\mathcal{R}_{\mathrm{sym}}$ is the set of codebook~rates $R_c>\lambda \mathds{1}_{N>1}$ such that $\gamma<(\frac{1}{N}+\frac{R_c}{\lambda})P$ and all $2N$ transmitted codewords are decoded successfully according to the sufficient conditions put forth by Propositions~\ref{prop_4}, \ref{prop_5} and \ref{prop_6}.  } $\mathcal{R}_{\mathrm{sym}}$ be the set of all $R_c>\lambda \mathds{1}_{N>1}$ such that 
\begin{itemize}
  \item All $2N$ transmitted codewords are decoded successfully.
  \item The average transmission power for Tx~$i$ satisfies $\frac{1}{|\mathcal{T}_i|}\sum_{t\in \mathcal{T}_i}x^2_{i,t}\leq P$  where $\mathcal{T}_i$ is the period of activity for Tx~$i$ until the time slot it transmits the last symbol in its $N^{th}$ burst.
\end{itemize}
 Towards characterizing $\mathcal{R}_{\mathrm{sym}}$,  we define the set $\mathcal{P}(x,y;\gamma)$ for real numbers $x$ and $y$ by
   \begin{equation}
   \label{vart4}
   \mathcal{P}(x,y;\gamma)=\left\{R_c>\lambda\mathds{1}_{N>1}: \big(1-\frac{x}{\lambda}(\phi_{\gamma}-\phi_{\gamma}')\big)R_c<\phi_{\gamma}'-(\phi_{\gamma}-\phi_{\gamma}')\frac{y}{\theta}, \gamma<\big(\frac{1}{N}+\frac{R_c}{\lambda}\big)P \right\},
\end{equation}
where 
\begin{equation}
\label{ }
\phi_{\gamma}:=\mathsf{C}(\gamma),\,\,\,\,\,\,\psi_{\gamma}:=\mathsf{C}\big(\frac{\gamma}{1+a\gamma}\big).
\end{equation}
One can rephrase the statements in Propositions~\ref{prop_4}, \ref{prop_5} and \ref{prop_6} in Proposition~\ref{prop_7}:
\begin{proposition}
\label{prop_7}
For $1\leq i\leq 2$ and $1\leq j\leq N $ assume the following conditions hold:
\begin{itemize}
  \item If $\omega^-_{i,j}, \omega^+_{i,j}\neq 0$, then $R_c\in \mathcal{P}(1,\theta;\gamma)$.
  \item If $\omega^-_{i,j}\neq 0$ and $\omega^+_{i,j}=0$, then $R_c\in \mathcal{P}(j-\omega^-_{i,j},\nu_{i'}-\nu_{i};\gamma)$.
  \item If $\omega^-_{i,j}=0$ and $\omega^+_{i,j}\neq0$, then $R_c\in \mathcal{P}(\omega^+_{i,j}-j,\nu_{i}-\nu_{i'};\gamma)$.
   \item If $\omega^-_{i,j}=\omega^+_{i,j}=0$, then $R_c\in \mathcal{P}(0,-\theta;\gamma)$.
\end{itemize}
Then the probability of error in decoding the $j^{th}$ message of Tx~$i$ can be made arbitrarily small by choosing $n$ sufficiently large.
\end{proposition}
If $N=1$, one can easily find $\mathcal{R}_{\mathrm{sym}}$ by considering the cases $\alpha<\theta$ and $\alpha>\theta$, separately. If $\alpha<\theta$, the two transmitted codewords overlap and we have $(\omega^-_{1,1},\omega^+_{1,1},\omega_{1,1})=(0,1,0)$ and $(\omega^-_{2,1},\omega^+_{2,1},\omega_{2,1})=(1,0,0)$. Applying Proposition~\ref{prop_7}, $R_c\in \mathcal{P}(1-1,\nu_1-\nu_2;\gamma)=\mathcal{P}(0,-\alpha;\gamma)$. If $\alpha>\theta$, none of the transmitted codewords experiences interference and hence, $R_c\in \mathcal{P}(0,-\theta;\gamma)$. Therefore,
\begin{eqnarray}
\mathcal{R}_{\mathrm{sym}}=\left\{\begin{array}{cc}
    \bigcup_{\gamma\geq 0}\mathcal{P}(0,-\alpha;\gamma)  &  \alpha<\theta  \\
    \bigcup_{\gamma\geq 0} \mathcal{P}(0,-\theta;\gamma) &   \alpha>\theta
\end{array}\right..
\end{eqnarray}
Define
\begin{equation}
\label{ }
\overline{R}_c:=\overline{R}_{c,1}=\overline{R}_{c,2},
\end{equation}
where $\overline{R}_{c,i}$ is given in~(\ref{bo_22}) and let $\gamma^*$ be the unique positive solution for $\gamma$ in the equation $\phi_{\gamma}'+\frac{\alpha}{\theta}(\phi_{\gamma}-\phi_{\gamma}')=\big(\frac{\gamma}{P}-1\big)\lambda$. Then it is easy to see that $\mathcal{R}_{\mathrm{sym}}\big|_{N=1}=\big(0,R_{c,\max}\big)$ where $R_{c,\max}$ is given by
\begin{equation}
\label{Rc_11}
R_{c,\max}=\left\{\begin{array}{cc}
    \big(\frac{\gamma^*}{P}-\frac{1}{N}\big)\lambda & \alpha<\theta   \\
    \overline{R}_c  &   \alpha>\theta
\end{array}\right..
\end{equation}
For $N\geq 2$, it is not necessarily the case that $\mathcal{R}_{\mathrm{sym}}=\big(\lambda,R_{\max}\big)$. For example, consider the setup  in panel~(a) of Fig.~\ref{pict_3}  where the line $R_{c,1}=R_{c,2}$ is shown in red. We see that $\mathcal{R}_{\mathrm{sym}}$ is the union of two disjoint intervals. 

Throughout the rest of this section let $N\geq 2$. Our goal is to characterize $\mathcal{R}_{\mathrm{sym}}$. Define\footnote{In Section~II.B we defined $\mu_i:=\frac{\eta_i}{\lambda_i}$ in (\ref{buffstab}).  Since $\eta_i$ is replaced by $\theta_iR_i$ in our new system model in this section, the choice of the letter~$\mu$ for the quotient $\frac{\theta R}{\lambda}$ in~(\ref{bio_11}) is in accordance with the one in~(\ref{buffstab}).}
  \begin{eqnarray}
\label{bio_11}
\mu:=\frac{\theta R_c}{\lambda}.
\end{eqnarray} 
Recall from Section~IV that the burst with index $j$ of Tx~$i$ extends from $j\mu+\nu_i$ to $j\mu+\nu_i+\theta$ on the $\overline{t}$-axis. Let $ j^*\geq 1$ be such that 
 \begin{eqnarray}
j^*\mu+\nu_1< \mu+\nu_2<(j^*+1)\mu+\nu_1,
\end{eqnarray} 
or equivalently,\footnote{If $j^*=1$, we drop the upper bound in (\ref{JSTAR1}).} 
\begin{equation}
\label{JSTAR1}
\frac{\alpha}{j^*}< \mu< \frac{\alpha}{j^*-1}.
\end{equation}
i.e., the starting point of the first burst of Tx~2 lies between the starting points of the bursts with indices $j^{*}$ and $j^*+1$ of Tx~1 as shown in~Fig.~\ref{bufpic6}. The interference pattern on the transmitted codewords depends on how the numbers $\mu$, $\frac{\alpha-\theta}{j^*-1}$ and $\frac{\alpha+\theta}{j^*}$ compare with each other. For example, Fig.~\ref{bufpic6} shows the case where $\frac{\alpha-\theta}{j^*-1}<\mu<\frac{\alpha+\theta}{j^*}$. As a result, each codeword of Tx~1 with index $j\ge j^*+1$ experiences interference on its both ends, the codeword of Tx~1 with index $j^*$ experiences interference only on its right~end and any codeword of Tx~1 with index $j\le j^*-1$ does not experience any interference. In general, it is easy to see that
\begin{eqnarray}
\label{sdb11}
\omega^-_{1,j}=\left\{\begin{array}{cc}
0 & j\leq j^*\\
    j-j^*  &  \mu<\frac{\alpha+\theta}{j^*},\,\,j\geq j^*+1  \\
     0 &  \mu>\frac{\alpha+\theta}{j^*} ,\,\, j\geq j^*+1
\end{array}\right.,\,\,\,\,\,\,
\end{eqnarray}
\begin{eqnarray}
\label{sdb22}
\omega^+_{1,j}=\left\{\begin{array}{cc}
0 & j\leq j^*-1\\
    j-j^*+1  &  \mu>\frac{\alpha-\theta}{j^*-1},\,\,j\geq j^*  \\
     0 &   \mu<\frac{\alpha-\theta}{j^*-1},\,\, j\geq j^*
\end{array}\right.,\,\,\,\,\,\,
\end{eqnarray}
\begin{figure}[t]
  \centering
  \includegraphics[scale=.7] {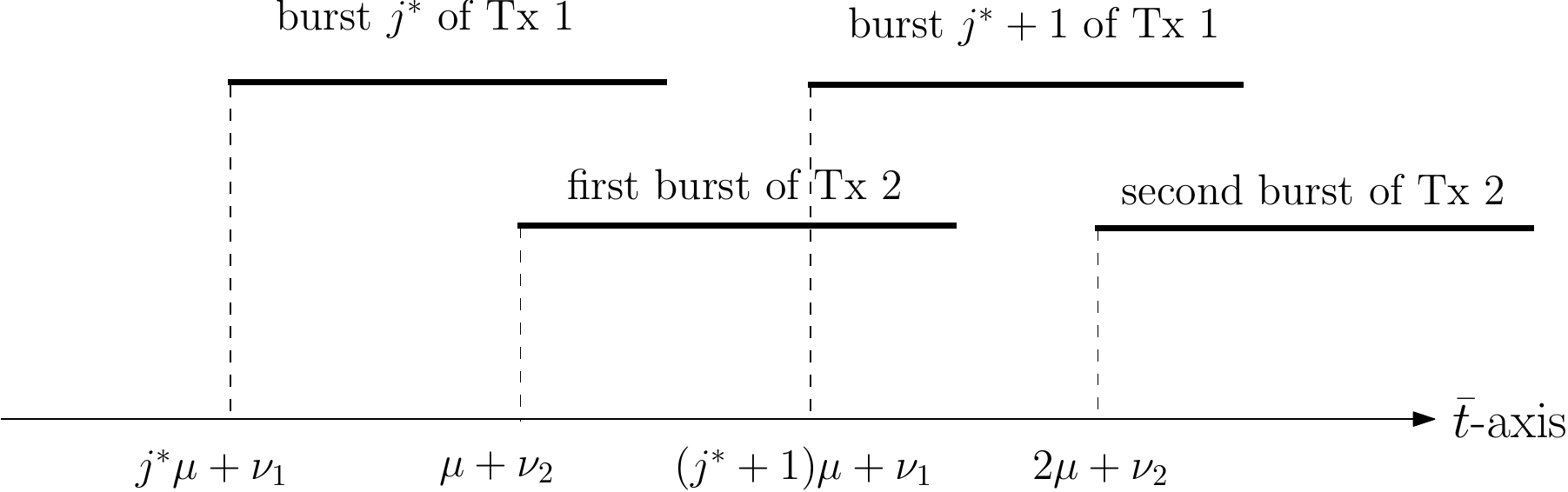}
  \caption{The integer $j^*\geq 1$ is such that the starting point of the first burst of Tx~2 lies between the bursts with indices $j^{*}$ and $j^*+1$ of Tx~1. The length of any burst is $\theta$ on the $\bar{t}$-axis. This picture shows the case where $(\mu+\nu_2)-(j^*\mu+\nu_1)<\theta$ and $((j^*+1)\mu+\nu_1)-(\mu+\nu_2)<\theta$, or equivalently, $\frac{\alpha-\theta}{j^*-1}<\mu<\frac{\alpha+\theta}{j^*}$. This implies that each codeword of Tx~1 with index $j\geq j^*+1$ experiences interference on its both ends, the codeword with index $j^{*}$ experiences interference only on its right~end and any codeword with index $j\leq j^*-1$ does not experience any interference.}
  \label{bufpic6}
 \end{figure} 
\begin{eqnarray}
\label{sdb33}
\omega^-_{2,j}=\left\{\begin{array}{cc}
    j+j^*-1  &   \mu>\frac{\alpha-\theta}{j^*-1},\,\,j\leq N-j^*+1  \\
     0 &    \mu<\frac{\alpha-\theta}{j^*-1},\,\,j\leq N-j^*+1\\
    0 & j\geq N-j^*+2
     \end{array}\right.\,\,\,\,\,\,
\end{eqnarray}
and
\begin{eqnarray}
\label{sdb44}
\omega^+_{2,j}=\left\{\begin{array}{cc}
    j+j^*  & \mu<\frac{\alpha+\theta}{j^*},\,\,j\leq N-j^*  \\
     0 & \mu>\frac{\alpha+\theta}{j^*},\,\,j\leq N-j^*\\
     0 & j\geq N-j^*+1
     \end{array}\right..\,\,\,\,\,\,
\end{eqnarray}
In view of the interference pattern described in~(\ref{sdb11}) to~(\ref{sdb44}) and considering the constraints in~(\ref{JSTAR1}), we define the four disjoint sets  
      \begin{eqnarray}
\label{con_11}
\mathscr{A}_{j^*}&:=&\left\{R_c>\lambda: \frac{\alpha-\theta}{j^*-1}<\mu<\frac{\alpha+\theta}{j^*},\,\,\frac{\alpha}{j^*}< \mu<\frac{\alpha}{j^*-1}\right\}\notag\\
&=&\left\{R_c>\lambda: \max\Big\{\frac{\alpha-\theta}{j^*-1}, \frac{\alpha}{j^*}\Big\}<\mu<\min\Big\{\frac{\alpha+\theta}{j^*}, \frac{\alpha}{j^*-1}\Big\}\right\},
\end{eqnarray}
 \begin{eqnarray}
\label{ }
\mathscr{B}_{j^*}&:=&\left\{R_c>\lambda: \mu<\min\Big\{\frac{\alpha-\theta}{j^*-1},\frac{\alpha+\theta}{j^*}\Big\},\,\,\,\frac{\alpha}{j^*}< \mu<\frac{\alpha}{j^*-1}\right\}\notag\\
&=&\left\{R_c>\lambda:\frac{\alpha}{j^*}<\mu<\min\Big\{\frac{\alpha-\theta}{j^*-1},\frac{\alpha+\theta}{j^*}\Big\}\right\},
\end{eqnarray}
 \begin{eqnarray}
\label{ }
\mathscr{C}_{j^*}&:=&\left\{R_c>\lambda: \mu>\max\Big\{\frac{\alpha-\theta}{j^*-1},\frac{\alpha+\theta}{j^*}\Big\},\,\,\,\frac{\alpha}{j^*}< \mu<\frac{\alpha}{j^*-1}\right\}\notag\\
&=&\left\{R_c>\lambda: \max\Big\{\frac{\alpha-\theta}{j^*-1},\frac{\alpha+\theta}{j^*}\Big\}<\mu<\frac{\alpha}{j^*-1}\right\}
\end{eqnarray}
and
 \begin{eqnarray}
\label{con_44}
\mathscr{D}_{j^*}&:=&\left\{R_c>\lambda:\frac{\alpha+\theta}{j^*}<\mu<\frac{\alpha-\theta}{j^*-1},\,\,\,\frac{\alpha}{j^*}< \mu<\frac{\alpha}{j^*-1} \right\}\notag\\
&=&\left\{R_c>\lambda:\frac{\alpha+\theta}{j^*}<\mu<\frac{\alpha-\theta}{j^*-1}\right\}.
\end{eqnarray}
  Next, we explicitly compute the sets $\mathcal{R}_{\mathrm{sym}}\bigcap\mathscr{A}_{j^*}, \mathcal{R}_{\mathrm{sym}}\bigcap\mathscr{B}_{j^*}, \mathcal{R}_{\mathrm{sym}}\bigcap\mathscr{C}_{j^*}$ and $\mathcal{R}_{\mathrm{sym}}\bigcap\mathscr{D}_{j^*}$. We will frequently invoke Proposition~\ref{prop_7} without mentioning to do so. 
     \begin{itemize}
   \item   \textbf{Computing $\mathcal{R}_{\mathrm{sym}}\bigcap\mathscr{A}_{j^*}$}:
   \begin{enumerate}
    \item Conditions for successful decoding at Rx~1
    \begin{itemize}
  \item  Any codeword of Tx~1 with index $j\leq j^*-1$ does not experience any interference.
  \item  The codeword of Tx~1 with index $j^*$ experiences interference only on its right~end and $\omega^+_{1,j^*}=1$.  We require $R_c\in \mathcal{P}(1-j^*,\nu_1-\nu_2;\gamma)=\mathcal{P}(1-j^*,-\alpha;\gamma)$.
 \item Any codeword of Tx~1 with index $j^*+1\leq j\leq N$ experiences interference on its both ends. We require $R_c\in \mathcal{P}(1,\theta;\gamma)$.
\end{itemize}
  \item Conditions for successful decoding at Rx~2
    \begin{itemize}
  \item  Any codeword of Tx~2 with index $1\leq j\leq N-j^*$ experiences interference on its both ends. We require $R_c\in \mathcal{P}(1,\theta;\gamma)$.
  \item  The codeword of Tx~2 with index $j=N-j^*+1$ experiences interference only on its left~end and  $\omega^-_{2,N-j^*+1}=N$. We require $R_c\in \mathcal{P}\big((N-j^*+1)-N,\nu_1-\nu_2;\gamma\big)=\mathcal{P}(1-j^*,-\alpha;\gamma)$.
  \item Any codeword of Tx~2 with index $N-j^*+2\leq j\leq N$ does not experience any interference. 
  \end{itemize}
\end{enumerate}
 It follows that 
  \begin{equation}
\label{PERT1}
\mathcal{R}_{\mathrm{sym}}\bigcap\mathscr{A}_{j^*}=\left\{\begin{array}{cc}
 \bigcup_{\gamma\geq 0}\big( \mathcal{P}(1,\theta;\gamma)\bigcap\mathcal{P}(1-j^*,-\alpha;\gamma)\big) \bigcap\mathscr{A}_{j^*}   &  1\leq j^*\leq N-1\\
  \bigcup_{\gamma\geq 0}\mathcal{P}(1-N,-\alpha;\gamma) \bigcap\mathscr{A}_{N}&  j^*=N\\
  \bigcup_{\gamma\geq 0}\mathcal{P}(0,-\theta;\gamma)\bigcap \mathscr{A}_{j^*}& j^*\geq N+1
\end{array}\right..
\end{equation}
 \item \textbf{Computing $\mathcal{R}_{\mathrm{sym}}\bigcap\mathscr{B}_{j^*}$}:
 \begin{enumerate}
 \item Conditions for successful decoding at Rx~1
    \begin{itemize}
  \item  Any codeword of Tx~1 with index $j\leq j^*$ does not experience any interference. 
  \item Any codeword of Tx~1 with index $j^*+1\leq j\leq N$ experiences interference only on its left~end and $\omega^-_{1,j}=j-j^*$. We require $R_c\in\mathcal{P}\big(j-(j-j^*),\nu_2-\nu_1;\gamma\big)=\mathcal{P}(j^*,\alpha;\gamma)$. 
\end{itemize}
  \item Conditions for successful decoding at Rx~2
    \begin{itemize}
  \item  Any codeword of Tx~2 with index $1\leq j\leq N-j^*$ experiences interference only on its right~end and $\omega_{2,j}^+=j+j^*$. We require $R_c\in \mathcal{P}\big((j+j^*)-j,\nu_2-\nu_1;\gamma\big)=\mathcal{P}(j^*,\alpha;\gamma)$.
    \item  Any codeword of Tx~2 with index $N-j^*+1\leq j\leq N$ does not experience any interference.  
  \end{itemize}   
  \end{enumerate}
  It follows that
    \begin{equation}
\label{PERT2}
\mathcal{R}_{\mathrm{sym}}\bigcap\mathscr{B}_{j^*}=\left\{\begin{array}{cc}
  \bigcup_{\gamma\ge0} \mathcal{P}(j^*,\alpha;\gamma) \bigcap\mathscr{B}_{j^*} &  1\leq j^*\leq N-1  \\
  \bigcup_{\gamma\ge0}\mathcal{P}(0,-\theta;\gamma)\bigcap\mathscr{B}_{j^*} &  j^*\geq N
\end{array}\right..
\end{equation}
\item \textbf{Computing $\mathcal{R}_{\mathrm{sym}}\bigcap\mathscr{C}_{j^*}$}:
 \begin{enumerate}
  \item Conditions for successful decoding at Rx~1
    \begin{itemize}
  \item  Any codeword of Tx~1 with index $j\leq j^*-1$ does not experience any interference. 
  \item Any codeword of Tx~1 with index $j^*\leq j\leq N$ experiences interference only on its right~end and $\omega_{1,j}^+=j-j^*+1$. We require $R_c\in \mathcal{P}\big((j-j^*+1)-j,\nu_1-\nu_2;\gamma\big)=\mathcal{P}(1-j^*,-\alpha;\gamma)$. 
\end{itemize}
  \item Conditions for successful decoding at Rx~2
    \begin{itemize}
  \item  Any codeword of Tx~2 with index $1\leq j\leq N-j^*+1$ experiences interference only on its left~end and $\omega_{2,j}^-=j+j^*-1$. We require $R_c\in \mathcal{P}\big(j-(j+j^*-1),\nu_1-\nu_2;\gamma\big)=\mathcal{P}(1-j^*,-\alpha;\gamma)$. 
    \item  Any codeword of Tx~2 with index $N-j^*+2\leq j\leq N$ does not experience any interference. 
  \end{itemize}  
  \end{enumerate}
  It follows that 
\begin{eqnarray}
\label{PERT3}
\mathcal{R}_{\mathrm{sym}}\bigcap\mathscr{C}_{j^*}=\left\{\begin{array}{cc}
   \bigcup_{\gamma\geq 0}\mathcal{P}(1-j^*,-\alpha;\gamma) \bigcap\mathscr{C}_{j^*}& 1\leq j^*\leq N\\
  \bigcup_{\gamma\ge0}\mathcal{P}(0,-\theta;\gamma)\bigcap\mathscr{C}_{j^*}& j^*\geq N+1
\end{array}\right..
\end{eqnarray}
\item \textbf{Computing $\mathcal{R}_{\mathrm{sym}}\bigcap\mathscr{D}_{j^*}$}: In this case, any codeword sent by Tx~1 or Tx~2 is received in the absence of interference. Hence, 
 \begin{equation}
\label{PERT4}
\mathcal{R}_{\mathrm{sym}}\bigcap\mathscr{D}_{j^*}=\bigcup_{\gamma\geq 0}\mathcal{P}(0,-\theta;\gamma)\bigcap\mathscr{D}_{j^*},\,\,\,\,j^*\ge1.
\end{equation}
 \end{itemize}
In general, one can characterize $\mathcal{R}_{\mathrm{sym}}$ for $N\geq 2$ by taking the following steps:
\begin{enumerate}
  \item Write 
  \begin{eqnarray}
\mathcal{R}_{\mathrm{sym}}=\bigcup_{j^*=1}^\infty\mathcal{R}_{j^*},
\end{eqnarray}
where
\begin{eqnarray}
\mathcal{R}_{j^*}=\mathcal{R}_{\mathrm{sym}}\bigcap\big(\mathscr{A}_{j^*}\bigcup\mathscr{B}_{j^*}\bigcup\mathscr{C}_{j^*}\bigcup\mathscr{D}_{j^*}\big).
\end{eqnarray}
  \item Use~(\ref{PERT1}),~(\ref{PERT2}),~(\ref{PERT3})~and~(\ref{PERT4}) to describe $\mathcal{R}_{j^*}$ for any $j^*\geq 1$.
\end{enumerate}
If $\alpha<\theta$, 
   \begin{equation}
\label{dim_11}
\mathscr{A}_{j^*}=\mathscr{B}_{j^*}=\mathscr{C}_{j^*}=\mathscr{D}_{j^*}=\emptyset,
\end{equation} 
for any $j^*\geq 2$, i.e., $\mathcal{R}_{\mathrm{sym}}=\mathcal{R}_1$. The following proposition characterizes $\mathcal{R}_{\mathrm{sym}}$ provided that $\alpha<\theta$.
\begin{proposition}
\label{prop_8}
Assume $\alpha<\theta$. Let $\gamma_0$, $\gamma_1$ and $\gamma_2$ be the solutions for $\gamma$ in $2\psi_{\gamma}=\phi_{\gamma}$, $\psi_{\gamma}=\lambda$ and $\psi_\gamma+\frac{\alpha}{\theta}(\phi_{\gamma}-\psi_{\gamma})=\lambda$, respectively. If $\gamma_1$ does not exist, let $\gamma_1=\infty$. 
\begin{itemize}
  \item If $\lambda<\psi_{\gamma_0}$, define
  \begin{eqnarray}
  \label{f_1111}
f(\gamma):=\left\{\begin{array}{cc}
    0  &  \gamma\leq \gamma_1  \\
    \lambda  &   \gamma>\gamma_1    
\end{array}\right.,\,\,\,\,\,\,\,\,\, g(\gamma)=\left\{\begin{array}{cc}
    0  &  \gamma\leq \gamma_1  \\
   \frac{2\psi_{\gamma}-\phi_{\gamma}}{1-\frac{1}{\lambda}(\phi_{\gamma}-\psi_{\gamma})} &   \gamma_1<\gamma\leq\gamma_2    \\
  \psi_\gamma+\frac{\alpha}{\theta}(\phi_{\gamma}-\psi_{\gamma}) &\gamma>\gamma_2
\end{array}\right..
\end{eqnarray}
  \item If $\lambda\geq \psi_{\gamma_0}$, define
    \begin{eqnarray}
    \label{g_1111}
f(\gamma):=\left\{\begin{array}{cc}
    0  &  \gamma\leq \gamma_2  \\
    \frac{2\psi_{\gamma}-\phi_{\gamma}}{1-\frac{1}{\lambda}(\phi_{\gamma}-\psi_{\gamma})}  &   \gamma_2<\gamma\leq \gamma_1  \\
  \lambda  & \gamma>\gamma_1
\end{array}\right.,\,\,\,\,\,\,\,\,\, g(\gamma)=\left\{\begin{array}{cc}
 0   &\gamma\leq\gamma_2\\
\psi_\gamma+\frac{\alpha}{\theta}(\phi_{\gamma}-\psi_{\gamma}) &  \gamma>\gamma_2
\end{array}\right..
\end{eqnarray}
\end{itemize} 
Then 
\begin{eqnarray}
\label{R_1234}
\mathcal{R}_{\mathrm{sym}}=\bigcup_{\gamma\geq 0}\Big(\max\Big\{f(\gamma), \big(\frac{\gamma}{P}-\frac{1}{N}\big)\lambda\Big\}, g(\gamma)\Big).
\end{eqnarray}
\end{proposition}
\begin{proof}
See Appendix~J. 
\end{proof}
A few remarks are in order:
\begin{itemize}
  \item Explicit expressions for $\gamma_0$ and $\gamma_1$ are 
\begin{equation}
\label{ }
\gamma_0=\frac{1+\sqrt{1+4a^2}}{2a^2},\,\,\,\,\,\gamma_1=\frac{2^{2\lambda}-1}{1-a(2^{2\lambda}-1)}.
\end{equation}
There is no closed-form expression for $\gamma_2$ and it must be computed numerically. 
  \item If $P$ is sufficiently large, e.g., $P>N\max\{\gamma_1,\gamma_2\}$, it is easy to see that $\mathcal{R}_{\mathrm{sym}}=(\lambda,R_{c,\max})$ where $R_{c,\max}$ is given in~(\ref{Rc_11}). For ``smaller'' values of $P$, $\inf_{R_c\in\mathcal{R}_{\mathrm{sym}}}R_c$ can be larger than $\lambda$ as we will see in the example in below.  
  \item In (\ref{R_1234}), $\mathcal{R}_{\mathrm{sym}}$ is given as the union of uncountably many intervals. It is more convenient to represent $\mathcal{R}_{\mathrm{sym}}$ as follows. Define the regions $\Omega'$, $\Omega''$ and $\Omega$ by
\begin{eqnarray}
\Omega':=\big\{(\gamma,R_c): f(\gamma)<R_c<g(\gamma)\big\},
\end{eqnarray}
\begin{eqnarray}
\Omega'':=\Big\{(\gamma,R_c): R_c\geq \big(\frac{\gamma}{P}-\frac{1}{N}\big)\lambda\Big\}
\end{eqnarray}
and 
\begin{equation}
\label{ }
\Omega:=\Omega'\bigcap\Omega'',
\end{equation}
i.e., $\Omega'$ is the set of all $(\gamma,R_c)$ such that the $2N$ transmitted codewords are sent immediately and decoded successfully at the receivers and $\Omega''$ is the set of all $(\gamma,R_c)$ such that the average power constraint in (\ref{power_33}) holds.  Then 
\begin{eqnarray}
\label{proj_11}
\mathcal{R}_{\mathrm{sym}}=\Pi(\Omega)
\end{eqnarray}
where the map $\Pi(\gamma,R_c)= R_c$ is the projection on the $R_c$-axis.
\end{itemize}
\begin{figure}[t]
\centering
\subfigure[]{
\includegraphics[scale=0.52]{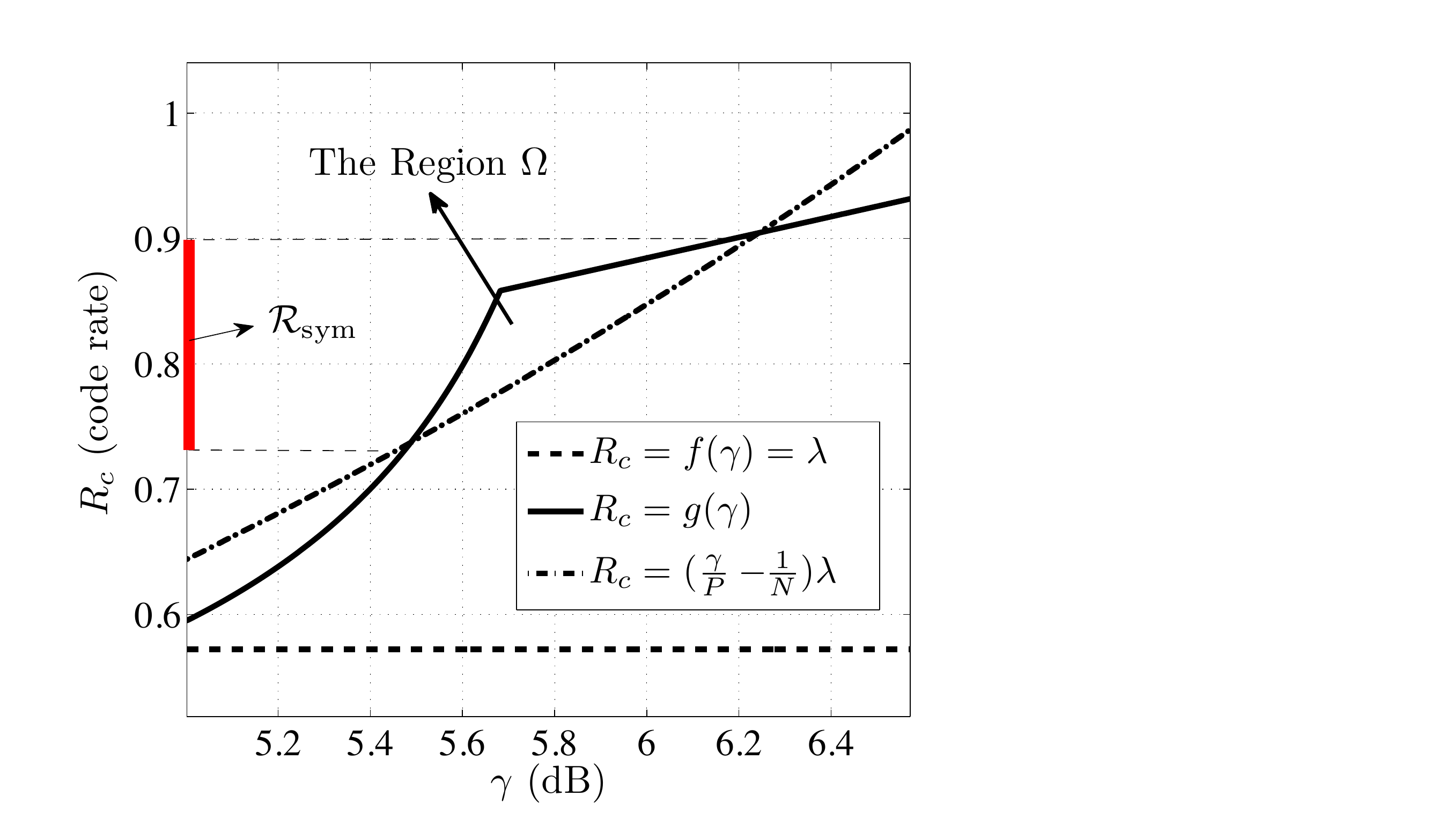}
\label{ref1}
}
\subfigure[]{
\includegraphics[scale=0.5]{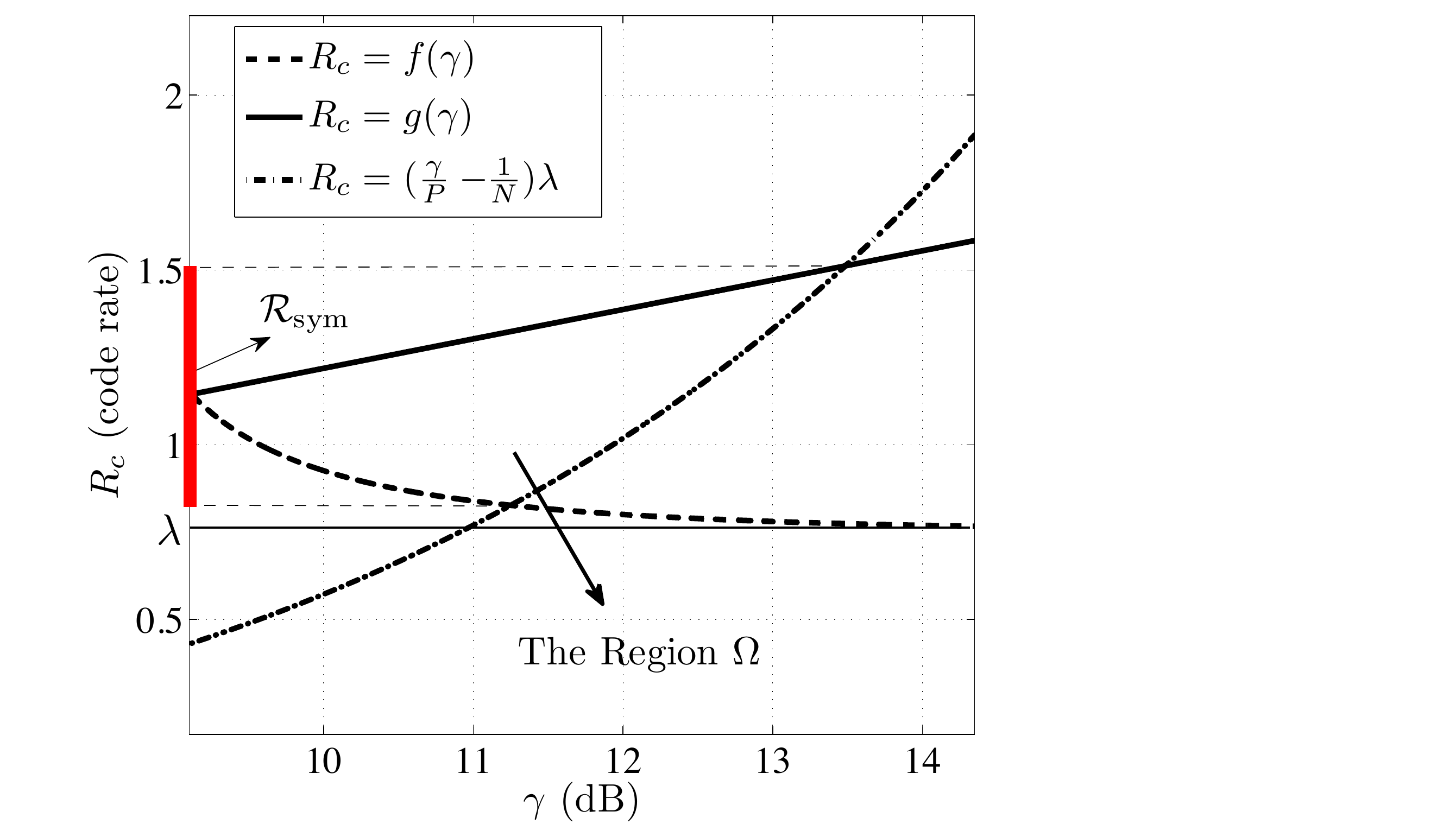}
\label{ref2}
}
\label{fig:subfigureExample}
\caption[Optional caption for list of figures]{Panel~(a) shows a setting where $N=4$, $\theta=1$, $\lambda=0.5722$, $a=0.5$, $P\approx 3.617\,\mathrm{dB}$ and $\alpha=0.5$. In this case, $\gamma_0\approx6.84\,\mathrm{dB}$ and $\lambda<\psi_{\gamma_0}\approx 0.6358$, i.e., $f(\gamma)$ and $g(\gamma)$ are given by~(\ref{f_1111}).  Panel~(b) shows a scenario where $N=4$, $\theta=1$, $\lambda=0.7629$, $a=0.5$, $P=10\,\mathrm{dB}$ and $\alpha=0.5$. In this case, $\gamma_0\approx 6.84\,\mathrm{dB}$ and $\lambda>\psi_{\gamma_0}\approx 0.6358$, i.e., $f(\gamma)$ and $g(\gamma)$ are given by~(\ref{g_1111}). The red strip on the $R_c$-axis is $\mathcal{R}_{\mathrm{sym}}$ as the projection of the region $\Omega$ on the $R_c$-axis. In both panel~(a) and panel~(b), $\inf_{R_c\in\mathcal{R}_{\mathrm{sym}}}R_c>\lambda$.  }
\label{pict_34}
\end{figure} 

\textbf{Example}- Let $N=4$, $\theta=1$, $a=0.5$ and $\alpha=0.5$. Then $\gamma_0\approx 6.84\,\mathrm{dB}$ and $\psi_{\gamma_0}\approx 0.6358$. We consider two cases:
\begin{itemize}
  \item If $\lambda<0.9\psi_{\gamma_0}=0.5722$, then $f(\gamma)$ and $g(\gamma)$ are given by~(\ref{f_1111}). Assuming $P\approx 3.617\,\mathrm{dB}$, Fig.~\ref{pict_34} in panel~(a) shows $f(\gamma)$, $g(\gamma)$ and $(\frac{\gamma}{P}-\frac{1}{N})\lambda$ as functions of $\gamma$. The set $\mathcal{R}_{\mathrm{sym}}$ is shown by a red strip as the projection of the region $\Omega$ on the $R_c$-axis.  
   \item If $\lambda\ge0.9\psi_{\gamma_0}=0.5722$, then $f(\gamma)$ and $g(\gamma)$ are given by~(\ref{g_1111}). Assuming $P=10\,\mathrm{dB}$, Fig.~\ref{pict_34} in panel~(b) shows $f(\gamma)$, $g(\gamma)$ and $(\frac{\gamma}{P}-\frac{1}{N})\lambda$ as functions of $\gamma$. The set $\mathcal{R}_{\mathrm{sym}}$ is shown by a red strip as the projection of the region $\Omega$ on the $R_c$-axis. 
\end{itemize}
In both cases we observe that $\inf_{R_c\in\mathcal{R}_{\mathrm{sym}}}R_c>\lambda$.

\section{Conclusion}
We have studied a two-user GIC-SDA with immediate transmissions under two different settings. In one scenario, the information source at each transmitter turned off after generating a given total number of bits and the transmitters only knew the statistics of the mutual delay between their bit streams. The  codebook rate at each transmitter was optimized in order to achieve a target average transmission rate and transmission power and maximize the probability of successful decoding at the receivers. In another scenario, the information sources were active indeterminately and the transmitters were aware of the exact mutual delay between their bit streams. We characterized an achievable rate region for the codebook~rates assuming the receivers treat interference as noise. This region was given as a union of uncountably many polyhedrons which is in general disconnected and non-convex due to infeasibility of time~sharing. 

\section*{Appendix A; Proof of (\ref{buffer3})}
We need the following Lemma which is a slightly weaker version of Theorem~4.4 in \cite{upfal}:
\begin{lem}
\label{lem1}
Let $\boldsymbol{x}$ be a $\mathrm{Bin}(N,p)$ random variable. Then  
\begin{equation}
\label{ }
\mathbb{P}\big(\boldsymbol{x}\geq(1+\epsilon)Np\big)\leq e^{-((1+\epsilon)\ln(1+\epsilon)-\epsilon)Np},
\end{equation} 
for any $\epsilon>0$.
\end{lem}
At the ``beginning'' of time slot $t=0$ the buffer is empty. Recall that $t_0:=\tau_i^{(1)}$ is the smallest $t$ such that $b_{i,t}+b'_{i,t}\geq \lfloor n\eta_i\rfloor$.  This implies that $b_{i,t_0}+b'_{i,t_0}=\big\lceil \frac{\lfloor n\eta_i\rfloor}{k_i}\big\rceil k_i$ is the smallest  multiple of $k_i$ which is larger than or equal to $\lfloor n\eta_i\rfloor$. The first codeword together with the preamble sequence are transmitted during the time slots $t_0+1,\cdots, t_0+n'+n_i$ and the content of the buffer at the beginning of time slot $t_0+1$ becomes
\begin{equation}
\label{cont}
b_{i,t_0+1}=\left\lceil \frac{\lfloor n\eta_i\rfloor}{k_i}\right\rceil k_i-\lfloor n\eta_i\rfloor.
\end{equation}
We are interested in computing the probability of the event that a new codeword is scheduled for transmission before or at the time slot $t_0+n'+n_i-1$, i.e., $\boldsymbol{b}_{i,\boldsymbol{t}_0+1}+\sum_{t=\boldsymbol{t}_0+1}^{\boldsymbol{t}_0+n'+n_i-1}\boldsymbol{b}'_{i,t}\geq \lfloor n\eta_i\rfloor$. We~have
\begin{eqnarray}
\label{cont1}
\mathbb{P}\bigg(\boldsymbol{b}_{i,\boldsymbol{t}_0+1}+\sum_{t=\boldsymbol{t}_0+1}^{\boldsymbol{t}_0+n'+n_i-1}\boldsymbol{b}'_{i,t}\geq \lfloor n\eta_i\rfloor\bigg)&\stackrel{}{\leq}&\mathbb{P}\bigg(\sum_{t=\boldsymbol{t}_0+1}^{\boldsymbol{t}_0+n'+n_i-1}\boldsymbol{b}'_{i,t}\geq \lfloor n\eta_i\rfloor-k_i\bigg)\notag\\
&=&\mathbb{P}\bigg(\sum_{t=\boldsymbol{t}_0+1}^{\boldsymbol{t}_0+n'+n_i-1}\frac{\boldsymbol{b}'_{i,t}}{k_i}\geq \frac{\lfloor n\eta_i\rfloor}{k_i}-1\bigg),
\end{eqnarray}
where the first step is due to the fact that $0\leq \boldsymbol{b}_{i,\boldsymbol{t}_0+1}< k_i$ due to (\ref{cont}). By assumption,  
\begin{equation}
\label{limit_11}
\lim_{n\to\infty}\frac{1}{(n'+n_i-1)q_i}\left(\frac{\lfloor n\eta_i\rfloor}{k_i}-1\right)=\frac{\eta_i}{\lambda_i \theta_i}=\frac{\mu_i}{\theta_i}>1.
\end{equation}
 Let 
 \begin{equation}
\label{delta111}
\epsilon_i:=\frac{1}{2}\Big(\frac{\mu_i}{\theta_i}-1\Big).
\end{equation}
 In view of (\ref{limit_11}), assume $n$ is large enough such that 
 \begin{eqnarray}
 \label{limit_22}
 \frac{1}{(n'+n_i-1)q_i}\left(\frac{\lfloor n\eta_i\rfloor}{k_i}-1\right)>1+\epsilon_i. 
\end{eqnarray}
 Since $\sum_{t=\boldsymbol{t}_0+1}^{\boldsymbol{t}_0+n'+n_i-1}\frac{\boldsymbol{b}'_{i,l}}{k_i}$ is a $\mathrm{Bin}(n'+n_i-1,q_i)$ random variable, the right side of (\ref{cont1}) is bounded from above by $\mathbb{P}(\mathrm{Bin}(n'+n_i-1,q_i)\geq (1+\epsilon_i)(n'+n_i-1)q_i)$ due to (\ref{limit_22}). Then Lemma~\ref{lem1} applies, i.e., 
 \begin{equation}
\label{ }
\mathbb{P}(\mathrm{Bin}(n'+n_i-1,q_i)\geq (1+\epsilon_i)(n'+n_i-1)q)\leq e^{-c_in},
\end{equation}
where $c_i=q_i\theta_i((1+\epsilon_i)\ln(1+\epsilon_i)-\epsilon_i)$ and we have assumed $n$ is large enough such that $n'+n_i-1=\ln n+\lfloor n\theta_i\rfloor-1>n\theta_i$.

\section{Appendix~B; Proof of Proposition~\ref{prop_opt11}}
We will use the following simple fact:
\begin{lem}
\label{sim_lem}
Let $\boldsymbol{x}_n$ for $n\geq 1$ be a sequence of real-valued random variables and $\lim_{n\to\infty}\boldsymbol{x}_n=a$ where $a>0$ is a real number. Then $\lim_{n\to\infty}\mathbb{P}(\boldsymbol{x}_n> 0)=1$.
\end{lem} 
\begin{proof}
Since $\lim_{n\to\infty}\boldsymbol{x}_n=a$, then $\boldsymbol{x}_n$ also converges to $a$ in probability. Fix $0<\epsilon<a$. We have $\lim_{n\to\infty}\mathbb{P}(\boldsymbol{x}_n> 0)\geq \lim_{n\to\infty}\mathbb{P}(|\boldsymbol{x}_n-a|<\epsilon)=1$.
\end{proof}
We only study the cases $j=1$ and $j=2$. The proof can be extended to any $j\geq 3$. Let $\beta_m:=\sum_{t=1}^{mn_i}b'_{i,t}$ for $m\ge1$ be the total number of bits arriving at the buffer of Tx~$i$ until the time slot of index $mn_i$. Under the Tx-Rx~synchronous scheme, Tx~$i$ checks its buffer at time slots $mn_i$ for $m\ge1$ and if its buffer content is more than $\lfloor n\eta_i\rfloor$, a codeword is sent over the channel during time slots $mn_i+1,\cdots,(m+1)n_i$. By SLLN, 
\begin{eqnarray}
\label{ehim_5}
\lim_{n\to\infty}\frac{\boldsymbol{\beta}_m}{n}=m\theta_i\lambda_i,
\end{eqnarray}
for any $m\geq 1$.
\begin{itemize}
  \item Let $j=1$. We have 
\begin{eqnarray}
\label{ehim_0}
 \varsigma_i^{(1)}=mn_i \iff \beta_1,\cdots,\beta_{m-1}<\lfloor n\eta_i\rfloor,\,\,\, \beta_m\ge\lfloor n\eta_i\rfloor
\end{eqnarray}
for any $m\geq 1$. Recall $\mu_i:=\frac{\eta_i}{\lambda_i}>\theta_i$. By assumption, $\mu_i$ is not an integer multiple of $\theta_i$.  Let $m^*\geq 1$ be such that  
\begin{equation}
\label{ehim_1}
m^*\theta_i<\mu_i<(m^*+1)\theta_i.
\end{equation}
 Putting (\ref{ehim_5}), (\ref{ehim_0}) and (\ref{ehim_1}) together and invoking Lemma~\ref{sim_lem}, 
\begin{eqnarray}
\label{tanem_1}
\lim_{n\to\infty}\mathbb{P}\big(\boldsymbol{\varsigma}_{i}^{(1)}=(m^*+1)n_i\big)=1.
\end{eqnarray}
By~(\ref{ehim_1}), fix $\delta_1>0$ such that $(1+\delta_1)\mu_i<(m^*+1)\theta_i$. Then  
\begin{eqnarray}
\label{tanem_3}
\mathbb{P}\Big(\boldsymbol{\varsigma}_i^{(1)}> (1+\delta_1)\boldsymbol{\tau}_i^{(1)}\Big)&\geq&\mathbb{P}\Big(\boldsymbol{\varsigma}_i^{(1)}> (1+\delta_1)\boldsymbol{\tau}_i^{(1)}, \boldsymbol{\varsigma}_{i}^{(1)}=(m^*+1)n_i\Big)\notag\\&=&\mathbb{P}\Big((m^*+1)n_i> (1+\delta_1)\boldsymbol{\tau}_i^{(1)}, \boldsymbol{\varsigma}_{i}^{(1)}=(m^*+1)n_i\Big).
\end{eqnarray}
Since $\boldsymbol{\tau}_i^{(1)}\sim\mathrm{NB}(\frac{\lfloor n\eta_i\rfloor}{k_i},q_i)$, we can apply SLLN to get $\lim_{n\to\infty}\frac{1}{n}\big((m^*+1)n_i-(1+\delta_1)\boldsymbol{\tau}_i^{(1)}\big)=(m^*+1)\theta_i-(1+\delta_1)\mu_i>0$. Using this together with Lemma~\ref{sim_lem},
\begin{eqnarray}
\label{tanem_2}
\lim_{n\to\infty}\mathbb{P}\big((m^*+1)n_i> (1+\delta_1)\boldsymbol{\tau}_i^{(1)}\big)=1. 
\end{eqnarray}
By (\ref{tanem_1}) and (\ref{tanem_2}), the expression on the right side of (\ref{tanem_3}) goes to $1$ as $n$ grows\footnote{If $\mathcal{A}_n$ and $\mathcal{B}_n$ are events such that $\lim_{n\to\infty}\mathbb{P}(\mathcal{A}_n)=\lim_{n\to\infty}\mathbb{P}(\mathcal{B}_n)=1$, then $\lim_{n\to\infty}\mathbb{P}(\mathcal{A}_n\bigcap\mathcal{B}_n)=1$.} and we get the desired result. 
  \item Let $j=2$. We have $ \varsigma_i^{(2)}=mn_i$ for some $m\geq 2$ if and only if there is $1\leq m'\leq m-1$ such that 
  \begin{eqnarray}
  \label{ehim_7}
\beta_1,\cdots,\beta_{m'-1}<\lfloor n\eta_i\rfloor,\,\,\beta_{m'}\geq \lfloor n\eta_i\rfloor,\,\,\beta_{m'+1},\cdots,\beta_{m-1}<2\lfloor n\eta_i\rfloor,\,\,\beta_{m}\ge2\lfloor n\eta_i\rfloor
\end{eqnarray}
Recall that $\mu_i$ is not a multiple of $\frac{\theta_i}{2}$. Let $m^*\geq 2$ be such that  
\begin{equation}
\label{ehim_9}
\frac{m^*\theta_i}{2}<\mu_i<\frac{(m^*+1)\theta_i}{2}.
\end{equation}
Then we can invoke Lemma~\ref{sim_lem} together with (\ref{ehim_5}), (\ref{ehim_7}) and (\ref{ehim_9}) to write
\begin{eqnarray}
\label{tanem_11}
\lim_{n\to\infty}\mathbb{P}\big(\boldsymbol{\varsigma}_{i}^{(2)}=(m^*+1)n_i\big)=1,
\end{eqnarray}
where $m'$ in (\ref{ehim_7}) is selected as $m'=\frac{m^*}{2}+1$ for even $m^*$ and $m'=\frac{m^*+1}{2}$ for odd $m^*$. By~(\ref{ehim_9}), fix $\delta_2>0$ such that $(1+\delta_2)\mu_i<\frac{(m^*+1)\theta_i}{2}$. Then  
\begin{eqnarray}
\label{tanem_3}
\mathbb{P}\Big(\boldsymbol{\varsigma}_i^{(2)}> (1+\delta_2)\boldsymbol{\tau}_i^{(2)}\Big)&\geq&\mathbb{P}\Big(\boldsymbol{\varsigma}_i^{(2)}> (1+\delta_2)\boldsymbol{\tau}_i^{(1)}, \boldsymbol{\varsigma}_{i}^{(2)}=(m^*+1)n_i\Big)\notag\\&=&\mathbb{P}\Big((m^*+1)n_i> (1+\delta_2)\boldsymbol{\tau}_i^{(2)}, \boldsymbol{\varsigma}_{i}^{(2)}=(m^*+1)n_i\Big).
\end{eqnarray}
Since $\boldsymbol{\tau}_i^{(2)}\sim\mathrm{NB}(\frac{2\lfloor n\eta_i\rfloor}{k_i},q_i)$, we can apply SLLN to get $\lim_{n\to\infty}\frac{1}{n}\big((m^*+1)n_i-(1+\delta_2)\boldsymbol{\tau}_i^{(2)}\big)=(m^*+1)\theta_i-2(1+\delta_2)\mu_i>0$. Using this together with Lemma~\ref{sim_lem},
\begin{eqnarray}
\label{tanem_2}
\lim_{n\to\infty}\mathbb{P}\big((m^*+1)n_i> (1+\delta_2)\boldsymbol{\tau}_i^{(2)}\big)=1. 
\end{eqnarray}
By (\ref{tanem_1}) and (\ref{tanem_2}), the expression on the right side of (\ref{tanem_3}) goes to $1$ as $n$ grows and we get the desired result. 
\end{itemize} 
Finally, we select $\delta:=\min_{1\leq j\leq N_i}\delta_j$ to make sure (\ref{sugar_11}) holds for any $1\leq j\leq N_i$. 

\section*{Appendix C; Proof of Proposition~\ref{prop1}}
Let $1\leq j_1\leq N_1$ and $1\leq j_2\leq N_2$. If $\tau_1^{(j_1)}+1\leq \tau_2^{(j_2)}+1\leq \tau_1^{(j_1)}+n'$, then the $j_2^{th}$  burst of Tx~$2$ starts while Tx~$1$ is  sending the preamble sequence in its $j_1^{th}$  burst. If  $\tau_1^{(j_1)}+1\leq \tau_2^{(j_2)}+n'+n_2\leq \tau_1^{(j_1)}+n'$, then the $j_2^{th}$ burst of Tx~$2$ ends while Tx~$1$ is sending the preamble sequence in its $j_1^{th}$ burst. Let $\mathcal{E}_{j_1,j_2}$ be the union of these two events, i.e., 
\begin{equation}
\label{ }
\mathcal{E}_{j_1,j_2}:=\big\{0\leq \boldsymbol{\tau}_2^{(j_2)}-\boldsymbol{\tau}_1^{(j_1)}\leq n'-1\big\}\bigcup\big\{n_2\leq \boldsymbol{\tau}_1^{(j_1)}-\boldsymbol{\tau}_2^{(j_2)}\leq n'+n_2-1\big\}.
\end{equation}
We have
\begin{equation}
\label{ref11}
\mathbb{P}\big(\bigcup_{j_1,j_2}\mathcal{E}_{j_1,j_2}\big)\leq \sum_{\substack{1\leq j_1\leq N_1\\ 1\leq j_2\leq N_2}}\mathbb{P}(\mathcal{E}_{j_1,j_2}).
\end{equation}
Moreover, 
\begin{eqnarray}
\label{ref22}
\mathbb{P}(\mathcal{E}_{j_1,j_2})\leq \mathbb{P}\big(0\leq \boldsymbol{\tau}_2^{(j_2)}-\boldsymbol{\tau}_1^{(j_1)}\leq n'-1\big)+\mathbb{P}\big(n_2\leq \boldsymbol{\tau}_1^{(j_1)}-\boldsymbol{\tau}_2^{(j_2)}\leq n'+n_2-1\big).
\end{eqnarray}
In view of (\ref{ref11}) and (\ref{ref22}), it is enough to show that 
\begin{equation}
\label{ber11}
\lim_{n\to\infty}\mathbb{P}\big(0\leq \boldsymbol{\tau}_2^{(j_2)}-\boldsymbol{\tau}_1^{(j_1)}\leq n'-1\big)=0
\end{equation}
and
\begin{equation}
\label{ber22}
\lim_{n\to\infty}\mathbb{P}\big(n_2\leq \boldsymbol{\tau}_1^{(j_1)}-\boldsymbol{\tau}_2^{(j_2)}\leq n'+n_2-1\big)=0
\end{equation}
for arbitrary choices of $j_1$ and $j_2$. To verify~(\ref{ber11}), define 
\begin{eqnarray}
\label{this00}
\boldsymbol{\rho}_n:=\frac{1}{n}\big(\boldsymbol{\tau}_2^{(j_2)}-\boldsymbol{\tau}_1^{(j_1)}\big),\,\,\,\boldsymbol{\rho}'_n:=\frac{1}{n}\big(\boldsymbol{\tau}_1^{(j_1)}-\boldsymbol{\tau}_2^{(j_2)}+n'-1\big).
\end{eqnarray}
Then 
\begin{eqnarray}
\label{conc11}
\mathbb{P}\big(0\leq \boldsymbol{\tau}_2^{(j_2)}-\boldsymbol{\tau}_1^{(j_1)}\leq n'-1\big)=\mathbb{P}\big( \boldsymbol{\rho}_n\geq 0, \boldsymbol{\rho}'_n\geq0\big)=\mathbb{P}\big(\min\{\boldsymbol{\rho}_n,\boldsymbol{\rho}'_n\}\geq 0\big).
\end{eqnarray}
 By (\ref{laws11}), 
\begin{eqnarray}
\label{label33}
\lim_{n\to\infty}\frac{\boldsymbol{\tau}_1^{(j_1)}}{n}&=&\lim_{n\to\infty}\frac{\boldsymbol{\xi}_1^{(j_1)}}{n}+\lim_{n\to\infty}\frac{\lfloor n\nu_1\rfloor}{n}=\lim_{n\to\infty}\frac{\boldsymbol{\xi}_1^{(j_1)}}{n}+\nu_1.
\end{eqnarray}
Using SLLN,  $\lim_{n\to\infty}\frac{\boldsymbol{\xi}_1^{(j_1)}}{n}=j_1\mu_1$ and we get 
\begin{equation}
\label{this11}
\lim_{n\to\infty}\frac{\boldsymbol{\tau}_1^{(j_1)}}{n}=j_1\mu_1+\nu_1.
\end{equation} 
Similarly,  
\begin{equation}
\label{this22}
\lim_{n\to\infty}\frac{\boldsymbol{\tau}_2^{(j_2)}}{n}=j_2\mu_2+\nu_2.
\end{equation} 
Define 
\begin{equation}
\label{ }
\rho:=j_2\mu_2-j_1\mu_1+\nu_2-\nu_1.
\end{equation}
By (\ref{this00}), (\ref{this11}) and (\ref{this22}) and noting that $\lim_{n\to\infty}\frac{n'-1}{n}=0$, we get $\lim_{n\to\infty}\boldsymbol{\rho}_{n}=\rho$, $\lim_{n\to\infty}\boldsymbol{\rho}'_{n}=-\rho$ and hence,
\begin{equation}
\label{saken_1}
\lim_{n\to\infty}\min\{\boldsymbol{\rho}_n,\boldsymbol{\rho}'_n\}=\min\{\rho,-\rho\}<0,
\end{equation}
where the last step is due to~(\ref{res_11}).  By~(\ref{saken_1}) and Lemma~\ref{sim_lem}, the proof of (\ref{ber11}) is complete. The proof of (\ref{ber22}) is quite similar and is omitted for brevity. 
\section*{Appendix D; Proof of Proposition~\ref{prop_2}}
We need the following Lemma: 
\begin{lem}[\textbf{Bernstein's inequality} \cite{lugosi}]
\label{lem_2}
Let $\boldsymbol{x}_1,\cdots,\boldsymbol{x}_m$ be independent zero mean real-valued random variables, $\boldsymbol{x}_i\leq 1$ for any $1\leq i\leq m$ and $\epsilon>0$. Then 
\begin{equation}
\label{ }
\mathbb{P}\Big(\frac{1}{m}\sum_{i=1}^m\boldsymbol{x}_i\geq \epsilon\Big)\leq e^{-\frac{m\epsilon^2}{2\left(\sigma^2+\frac{\epsilon}{3}\right)}},
\end{equation}
where $\sigma^2$ is the arithmetic average of the variances of $\boldsymbol{x}_1,\cdots,\boldsymbol{x}_m$. 
\end{lem}
We have
\begin{equation}
\label{uyh1}
\mathbb{P}(\,\hat{\boldsymbol{t}}_1< \boldsymbol{t}_1)=\sum_{t_1\geq 0}\mathbb{P}(\,\hat{\boldsymbol{t}}_1<\boldsymbol{t}_1|\boldsymbol{t}_1=t_1)\mathbb{P}(\boldsymbol{t}_1=t_1).
\end{equation}
Moreover,
\begin{eqnarray}
\label{union1}
&&\mathbb{P}(\,\hat{\boldsymbol{t}}_1<\boldsymbol{t}_1|\boldsymbol{t}_1=t_1)\notag\\
&&=\mathbb{P}\left(\exists\,t<t_1:  \textrm{$\big((\boldsymbol{s}'_{1,l})_{l=0}^{n'-1},(\boldsymbol{y}_{1,l})_{l=t}^{t+n'-1}\big)\in A^{(n')}_{\epsilon}[p^{(1)}]$ or $\big((\boldsymbol{s}'_{2,l})_{l=0}^{n'-1},(\boldsymbol{y}_{1,l})_{l=t}^{t+n'-1}\big)\in A^{(n')}_{\epsilon}[p^{(3)}]$}\right)\notag\\
&&\leq \sum_{t=0}^{t_1-1}\mathbb{P}\left( \textrm{$\big((\boldsymbol{s}'_{1,l})_{l=0}^{n'-1},(\boldsymbol{y}_{1,l})_{l=t}^{t+n'-1}\big)\in A^{(n')}_{\epsilon}[p^{(1)}]$}\right)\notag\\&&\hskip0.5cm+\sum_{t=0}^{t_1-1}\mathbb{P}\left( \textrm{$\big((\boldsymbol{s}'_{2,l})_{l=0}^{n'-1},(\boldsymbol{y}_{1,l})_{l=t}^{t+n'-1}\big)\in A^{(n')}_{\epsilon}[p^{(3)}]$}\right).
\end{eqnarray}
In the following, we find an upper bound on each term on the right side of (\ref{union1}).
\begin{itemize}
  \item The term $\sum_{t=0}^{t_1-1}\mathbb{P}\big( \textrm{$\big((\boldsymbol{s}'_{1,l})_{l=0}^{n'-1},(\boldsymbol{y}_{1,l})_{l=t}^{t+n'-1}\big)\in A^{(n')}_{\epsilon}[p^{(1)}]$}\big)$: We study the cases $t\leq t_{1}-n'$ and $t_{1}-n'+1\leq  t\leq t_{1}-1$, separately: 
\begin{itemize}
  \item Let $t\leq t_{1}-n'$.  We explicitly write the constraint defining $A_{\epsilon}^{(m)}[p^{(1)}]$ in (\ref{smart1}) as the set of all $(\vec{x},\vec{y}\,)$ such that $\big|\frac{1}{m}\sum_{l=1}^{m}\big(\frac{1}{2\gamma_1}x_l^2+\frac{1}{2}\big(y_l-x_l\big)^2-1\big)\big|<\frac{\epsilon }{\log e}$. Then we have the thread of inequalities 
 \begin{eqnarray}
 \label{platini}
&&\mathbb{P}\big(\, \big((\boldsymbol{s}'_{1,l})_{l=0}^{n'-1},(\boldsymbol{y}_{1,l})_{l=t}^{t+n'-1}\big)\in A_{\epsilon}^{(n')}[p^{(1)}]\big )\notag\\&&\stackrel{}{\leq} \mathbb{P}\Big(\Big|\frac{1}{n'}\sum_{l=0}^{n'-1}\Big(\frac{1}{2\gamma_1}|\boldsymbol{s}'_{1,l}|^2+\frac{1}{2}(\boldsymbol{y}_{1,l+t}-\boldsymbol{s}'_{1,l})^2-1\Big)\Big|<\frac{\epsilon}{\log e}\Big)\notag\\
&&\leq \mathbb{P}\Big(\frac{1}{n'}\sum_{l=0}^{n'-1}\Big(\frac{1}{2\gamma_1}|\boldsymbol{s}'_{1,l}|^2+\frac{1}{2}\big(\boldsymbol{y}_{1,l+t}-\boldsymbol{s}'_{1,l}\,\big)^2-1\Big)<\frac{\epsilon}{\log e}\Big)\notag\\
&&\stackrel{(a)}{=} \mathbb{P}\Big(\frac{1}{n'}\sum_{l=0}^{n'-1}\Big(\frac{1}{2\gamma_1}|\boldsymbol{s}'_{1,l}|^2+\frac{1}{2}\big(\boldsymbol{z}_{1,l+t}-\boldsymbol{s}'_{1,l}\,\big)^2-1\Big)<\frac{\epsilon}{\log e}\Big)\notag\\
&&\stackrel{(b)}{\leq} \mathbb{P}\Big(\frac{1}{n'}\sum_{l=0}^{n'-1}\Big(1+\frac{\gamma_1}{2}-\frac{1}{2\gamma_1}|\boldsymbol{s}'_{1,l}|^2-\frac{1}{2}\big(\boldsymbol{z}_{1,l+t}-\boldsymbol{s}'_{1,l}\,\big)^2\Big)>\frac{\gamma_1}{2}-\frac{\epsilon}{\log e}\Big)\notag\\
&&\stackrel{(c)}{=}\mathbb{P}\Big(\frac{1}{ n'}\sum_{l=0}^{n'-1}\Big(1-\frac{\frac{1}{2\gamma_1}|\boldsymbol{s}'_{1,l}|^2+\frac{1}{2}\big(\boldsymbol{z}_{1,l+t}-\boldsymbol{s}'_{1,l}\,\big)^2}{1+\frac{\gamma_1}{2}}\Big)>\frac{\frac{\gamma_1}{2}-\frac{\epsilon}{\log e}}{1+\frac{\gamma_1}{2}}\Big)\notag\\
&&\stackrel{(d)}{\leq}e^{-\Theta(n')},
\end{eqnarray}
\normalsize
where $(a)$ is due to the fact that  the signal at Rx~1 during time slots $t$ to $t+n'-1$ only consists of the ambient noise, in $(b)$ we have added $\frac{\gamma_1}{2}$ to both sides of the inequality in $(a)$ after multiplying both sides of the inequality by $-1$, in $(c)$ both sides are divided by $1+\frac{\gamma_1}{2}$ and $(d)$ is due to Lemma~\ref{lem_2} (Bernstein's inequality)  where it is assumed that  $\epsilon<\frac{\gamma_1}{2}\log e$. In fact, the random variables $\boldsymbol{w}_l=1-\frac{\frac{1}{2\gamma_1}|\boldsymbol{s}'_{1,l}|^2+\frac{1}{2}(\boldsymbol{z}_{1,l+t}-\boldsymbol{s}'_{1,l}\,)^2}{1+\frac{\gamma_1}{2}}$ are independent with zero mean and finite variance and $\boldsymbol{w}_l\leq 1$. Therefore, one can apply Bernstein's inequality.
  \item Let $t_1-n'+1\leq  t\leq t_{1}-1$. Then we get
   \begin{eqnarray}
   \label{platini2}
&&\mathbb{P}\big(\, \big((\boldsymbol{s}'_{1,l})_{l=0}^{ n'-1},(\boldsymbol{y}_{1,l})_{l=t}^{t+n'-1}\big)\in A_{\epsilon}^{(n')}[p^{(1)}]\big )\notag\\
&&\leq\mathbb{P}\Big(\frac{1}{n'}\sum_{l=0}^{ n'-1}\Big(\frac{1}{2\gamma_1}|\boldsymbol{s}'_{1,l}|^2+\frac{1}{2}\big(\boldsymbol{y}_{1,l+t}-\boldsymbol{s}'_{1,l}\,\big)^2-1\Big)<\frac{\epsilon}{\log e}\Big)
\notag\\&&\stackrel{(a)}{=} \mathbb{P}\Big(\frac{1}{n'}\sum_{l=0}^{t_{1}-t-1}\Big(\frac{1}{2\gamma_1}|\boldsymbol{s}'_{1,l}|^2+\frac{1}{2}\big(\boldsymbol{z}_{1,l+t}-\boldsymbol{s}'_{1,l}\,\big)^2-1\Big)\Big.\Big.\notag\\&&\hskip2.5cm\Big.\Big.+\frac{1}{n'}\sum_{l=t_{1}-t}^{n'-1}\Big(\frac{1}{2\gamma_1}|\boldsymbol{s}'_{1,l}|^2+\frac{1}{2}\big(\boldsymbol{s}'_{1,l+t-t_{1}}+\boldsymbol{z}_{1,l+t}-\boldsymbol{s}'_{1,l}\,\big)^2-1\Big)<\frac{\epsilon}{\log e}\Big)\notag\\
&&\stackrel{(b)}{=} \mathbb{P}\Big(\frac{1}{n'}\sum_{l=0}^{t_{1}-t-1}\Big(1+\frac{\gamma_1}{2}-\frac{1}{2\gamma_1}|\boldsymbol{s}'_{1,l}|^2-\frac{1}{2}\big(\boldsymbol{z}_{1,l+t}-\boldsymbol{s}'_{1,l}\,\big)^2\Big)\Big.\Big.\notag\\&&\hskip2.5cm\Big.\Big.+\frac{1}{n'}\sum_{l=t_{1}-t}^{n'-1}\Big(1+\gamma_1-\frac{1}{2\gamma_1}|\boldsymbol{s}'_{1,l}|^2-\frac{1}{2}\big(\boldsymbol{s}'_{1,l+t-t_{1}}+\boldsymbol{z}_{1,l+t}-\boldsymbol{s}'_{1,l}\,\big)^2\Big)>\epsilon_n\Big),\notag\\
\end{eqnarray}
\normalsize
where in $(a)$ we have used the fact that $\boldsymbol{y}_{1,l+t}$ is the ambient noise for $l<t_{1}-t$ and $\boldsymbol{y}_{1,l+t}=\boldsymbol{s}'_{1,l+t-t_{1}}+\boldsymbol{z}_{1,l+t}$ for $t_1-t\leq l\leq n'-1$ and in $(b)$ we have added $\frac{\gamma_1}{2}$ to each term in the first sum and $\gamma_1$ to each term in the second sum after multiplying both sides of the inequality in $(a)$ by $-1$. Moreover, $\epsilon_n$ is given by 
\begin{eqnarray}
\epsilon_n:=\frac{t_1-t}{2n'}\gamma_1+\frac{n'+t-t_1}{n'}\gamma_1-\frac{\epsilon}{\log e}.
\end{eqnarray}
We can write 
\begin{eqnarray}
\label{plat_22}
\epsilon_n=\big(1-\frac{t_1-t}{2n'}\big)\gamma_1-\frac{\epsilon}{\log e}\geq \big(1-\frac{n'-1}{2n'}\big)\gamma_1-\frac{\epsilon}{\log e}\geq \frac{\gamma_1}{3}-\frac{\epsilon}{\log e},
\end{eqnarray}
where the penultimate step is due to $t_1-t\leq n'-1$ and in the last step we are assuming $n$ is sufficiently large\footnote{We have $\lim_{n\to\infty}(1-\frac{n'-1}{2n'})=\frac{1}{2}$ and hence, $1-\frac{n'-1}{2n'}>\frac{1}{3}$ for sufficiently large $n$.} such that $1-\frac{n'-1}{2n'}\geq \frac{1}{3}$. By (\ref{platini2}) and (\ref{plat_22}), 
\begin{eqnarray}
\label{plat_33}
&&\mathbb{P}\big(\, \big((\boldsymbol{s}'_{1,l})_{l=0}^{ n'-1},(\boldsymbol{y}_{1,l})_{l=t}^{t+n'-1}\big)\in A_{\epsilon}^{(n')}[p^{(1)}]\big )\notag\\
&&\stackrel{}{\le} \mathbb{P}\Big(\frac{1}{n'}\sum_{l=0}^{t_{1}-t-1}\Big(1+\frac{\gamma_1}{2}-\frac{1}{2\gamma_1}|\boldsymbol{s}'_{1,l}|^2-\frac{1}{2}\big(\boldsymbol{z}_{1,l+t}-\boldsymbol{s}'_{1,l}\,\big)^2\Big)\Big.\Big.\notag\\&&\hskip1.5cm\Big.\Big.+\frac{1}{n'}\sum_{l=t_{1}-t}^{n'-1}\Big(1+\gamma_1-\frac{1}{2\gamma_1}|\boldsymbol{s}'_{1,l}|^2-\frac{1}{2}\big(\boldsymbol{s}'_{1,l+t-t_{1}}+\boldsymbol{z}_{1,l+t}-\boldsymbol{s}'_{1,l}\,\big)^2\Big)>\frac{\gamma_1}{3}-\frac{\epsilon}{\log e}\Big),\notag\\
\end{eqnarray}
where we are assuming that $\epsilon<\frac{\gamma_1}{3}\log e$. Define 
\normalsize
\begin{equation}
\label{ }
\boldsymbol{w}_l=\left\{\begin{array}{cc}
    1-\frac{1}{1+\frac{\gamma_1}{2}}\big(\frac{1}{2\gamma_1}|\boldsymbol{s}'_{1,l}|^2-\frac{1}{2}\big(\boldsymbol{z}_{1,l+t}-\boldsymbol{s}'_{1,l}\,\big)^2\big)  & 0\leq l\leq t_1-t-1   \\
    1-\frac{1}{1+\gamma_1}\big(\frac{1}{2\gamma_1}|\boldsymbol{s}'_{1,l}|^2-\frac{1}{2}\big(\boldsymbol{s}'_{1,l+t-t_{1}}+\boldsymbol{z}_{1,l+t}-\boldsymbol{s}'_{1,l}\,\big)^2\big)  &   t_1-t\leq l\leq n'-1
\end{array}\right..
\end{equation}
Then we can write (\ref{plat_33}) as 
\begin{eqnarray}
&&\mathbb{P}\big(\, \big((\boldsymbol{s}'_{1,l})_{l=0}^{ n'-1},(\boldsymbol{y}_{1,l})_{l=t}^{t+n'-1}\big)\in A_{\epsilon}^{(n')}[p^{(1)}]\big )\notag\\
&&\leq\mathbb{P}\Big(\frac{1}{n'}\sum_{l=0}^{t_1-t-1}\big(1+\frac{\gamma_1}{2}\big)\boldsymbol{w}_l+\frac{1}{n'}\sum_{l=t_1-t}^{n'-1}(1+\gamma_1)\boldsymbol{w}_l>\frac{\gamma_1}{3}-\frac{\epsilon}{\log e}\Big)\notag\\
&&\leq\mathbb{P}\Big(\frac{1}{n'}\sum_{l=0}^{n'-1}\boldsymbol{w}_l>\frac{1}{1+\gamma_1}\big(\frac{\gamma_1}{3}-\frac{\epsilon}{\log e}\big)\Big),
\end{eqnarray}
where in the last step we have used the fact that $\frac{1}{n'}\sum_{l=0}^{t_1-t-1}\big(1+\frac{\gamma_1}{2}\big)\boldsymbol{w}_l+\frac{1}{n'}\sum_{l=t_1-t}^{n'-1}(1+\gamma_1)\boldsymbol{w}_l<(1+\gamma_1)\times \frac{1}{n'}\sum_{l=0}^{n'-1}\boldsymbol{w}_l$. Note that $\boldsymbol{w}_l$ have zero mean and finite variance and $\boldsymbol{w}_l\leq 1$ for all $0\leq l\leq n'-1$. However, in contrast with the previous case where we had $t\leq t_1-n'$, one can not apply Bernstein's Lemma because $\boldsymbol{w}_l$ are no longer independent random variables. In fact, any two $\boldsymbol{w}_{l}$ and $\boldsymbol{w}_{l'}$ are dependent if and only if   $|l-l'|=t_{1}-t$. To circumvent this difficulty, we use a trick in Appendix~24B in \cite{gamkim}. We consider two cases: 
\begin{itemize}
  \item If $t_{1}-t$ is odd, the terms with odd indices are independent. Similarly,  the terms with even indices are independent. Then 
  \begin{eqnarray}
  \label{platini3}
\mathbb{P}\Big(\frac{1}{n'}\sum_{l=0}^{n'-1}\boldsymbol{w}_l>\frac{1}{1+\gamma_1}\big(\frac{\gamma_1}{3}-\frac{\epsilon}{\log e}\big)\Big)&\leq&\mathbb{P}\bigg(\frac{1}{n'}\sum_{\substack{l=0\\\textrm{$l$ even}}}^{n'-1}\boldsymbol{w}_{l}>\frac{1}{2(1+\gamma_1)}\big(\frac{\gamma_1}{3}-\frac{\epsilon}{\log e}\big)\bigg)\notag\\
&&\hskip0.1cm+\mathbb{P}\bigg(\frac{1}{n'}\sum_{\substack{l=0\\\textrm{$l$ odd}}}^{ n'-1}\boldsymbol{w}_{l}>\frac{1}{2(1+\gamma_1)}\big(\frac{\gamma_1}{3}-\frac{\epsilon}{\log e}\big)\bigg).\notag\\
\end{eqnarray}
At this point, similar to (\ref{platini}), one can apply Bernstein's inequality to conclude that each term on the right side of (\ref{platini3}) is bounded from above by $e^{-\Theta(n')}$.
   \item If $t_{1}-t$ is even, we need to partition the set of integers into two disjoint sets $\mathcal{I}$ and $\mathcal{J}$ such that the difference of any two element in each of these sets is not equal to $t_{1}-t$. Such a partition is given in Appendix~24B in \cite{gamkim} or Appendix~D in \cite{kam}. Then 
\begin{eqnarray}
\label{platini4}
\mathbb{P}\Big(\frac{1}{n'}\sum_{l=0}^{n'-1}\boldsymbol{w}_l>\frac{1}{1+\gamma_1}\big(\frac{\gamma_1}{3}-\frac{\epsilon}{\log e}\big)\Big)&\le&\mathbb{P}\bigg(\frac{1}{ n'}\sum_{\substack{l=0\\\textrm{$l\in \mathcal{I}$}}}^{n'-1}\boldsymbol{w}_{l}>\frac{1}{1+\gamma_1}\big(\frac{\gamma_1}{3}-\frac{\epsilon}{\log e}\big)\bigg)\notag\\&&\hskip0.1cm+\mathbb{P}\bigg(\frac{1}{n'}\sum_{\substack{l=0\\\textrm{$l\in\mathcal{J}$}}}^{n'-1}\boldsymbol{w}_{l}>\frac{1}{1+\gamma_1}\big(\frac{\gamma_1}{3}-\frac{\epsilon}{\log e}\big)\bigg).\notag\\
\end{eqnarray}
As each of $\sum_{\substack{\\\textrm{$l\in \mathcal{I}$}}}\boldsymbol{w}_{l}$ and $\sum_{\substack{\\\textrm{$l\in \mathcal{J}$}}}\boldsymbol{w}_{l}$ are sums of independent random variables, it follows that each term on the right side of (\ref{platini4}) is bounded from above by $e^{-\Theta(n')}$.
\end{itemize}
\end{itemize}
We conclude that whether $t\leq t_{1}-n'$ or $t_{1}-n'+1\leq  t\leq t_{1}-1$, each term in the first sum on the right side of (\ref{union1}) is bounded from above by $\Theta(1)e^{-\Theta(n')}$. Since there are $t_1$ terms in this sum, we get 
\begin{equation}
\label{uyh2}
\sum_{t=0}^{t_1-1}\mathbb{P}\left( \textrm{$\big((\boldsymbol{s}'_{1,l})_{l=0}^{n'-1},(\boldsymbol{y}_{1,l})_{l=t}^{t+n'-1}\big)\in A^{(n')}_{\epsilon}[p^{(1)}]$}\right)\leq t_1\Theta(1)e^{-\Theta(n')}.
\end{equation}
 \item The term $\sum_{t=0}^{t_1-1}\mathbb{P}\big( \textrm{$\big((\boldsymbol{s}'_{2,l})_{l=0}^{n'-1},(\boldsymbol{y}_{1,l})_{l=t}^{t+n'-1}\big)\in A^{(n')}_{\epsilon}[p^{(3)}]$}\big)$: The analysis in this case follows similar lines of reasoning in the previous case and is omitted. The result is that
\begin{equation}
\label{uyh3}
\sum_{t=0}^{t_1-1}\mathbb{P}\left( \textrm{$\big((\boldsymbol{s}'_{2,l})_{l=0}^{n'-1},(\boldsymbol{y}_{1,l})_{l=t}^{t+n'-1}\big)\in A^{(n')}_{\epsilon}[p^{(3)}]$}\right)\leq t_1\Theta(1)e^{-\Theta(n')}.
\end{equation}
\end{itemize}
By (\ref{uyh1}), (\ref{union1}), (\ref{uyh2}) and (\ref{uyh3}), 
\begin{eqnarray}
\label{uyh4}
\mathbb{P}(\,\hat{\boldsymbol{t}}_1\leq \boldsymbol{t}_1)\leq \Theta(1)e^{-\Theta(n')}\sum_{t_1\geq0}t_1\mathbb{P}(\boldsymbol{t}_1=t_1)=\Theta(1)\mathbb{E}[\boldsymbol{t}_1]e^{-\Theta(n')}.
\end{eqnarray}
But, 
\begin{eqnarray}
\label{uyh5}
\mathbb{E}[\boldsymbol{t}_1]&=&\mathbb{E}[\boldsymbol{\tau}_1^{(1)}+1]\notag\\&=&\mathbb{E}[\boldsymbol{\xi}_1^{(1)}]+\mathbb{E}[\lfloor n\boldsymbol{\nu}_1\rfloor]+1\notag\\
&\leq&\mathbb{E}[\boldsymbol{\xi}_1^{(1)}]+\mathbb{E}[ n\boldsymbol{\nu}_1]+1\notag\\
&=&\frac{\lfloor n\eta_1\rfloor}{\lambda_1}+\frac{n}{2}+1\notag\\
&=&\Theta(n)
\end{eqnarray}
By (\ref{uyh4}) and (\ref{uyh5}), 
\begin{eqnarray}
\mathbb{P}(\,\hat{\boldsymbol{t}}_1< \boldsymbol{t}_1)\leq \Theta(n)e^{-\Theta(n')},
\end{eqnarray}
as desired. 
\section*{Appendix~E; Proof of (\ref{canned_11})}
The proof follows similar lines of reasoning in (\ref{platini}). We explicitly write the constraint defining $A_{\epsilon}^{(m)}[p^{(3)}]$ in (\ref{smart1}) as the set of all $(\vec{x},\vec{y}\,)$ such that $\big|\frac{1}{m}\sum_{l=1}^{m}\big(\frac{1}{2\gamma_2}x_l^2+\frac{1}{2}\big(y_l-\sqrt{a_2}\,x_l\big)^2-1\big)\big|<\frac{\epsilon }{\log e}$. Then  
\small
 \begin{eqnarray}
 \label{}
&&\mathbb{P}\big(\, \big((\boldsymbol{s}'_{2,l})_{l=0}^{n'-1},(\boldsymbol{y}_{1,l})_{l=t_1}^{t_1+n'-1}\big)\in A_{\epsilon}^{(n')}[p^{(3)}]\big )\notag\\&&\stackrel{}{\leq} \mathbb{P}\Big(\Big|\frac{1}{n'}\sum_{l=0}^{n'-1}\Big(\frac{1}{2\gamma_2}|\boldsymbol{s}'_{2,l}|^2+\frac{1}{2}(\boldsymbol{y}_{1,l+t_1}-\sqrt{a_2}\,\boldsymbol{s}'_{2,l})^2-1\Big)\Big|<\frac{\epsilon}{\log e}\Big)\notag\\
&&\leq \mathbb{P}\Big(\frac{1}{n'}\sum_{l=0}^{n'-1}\Big(\frac{1}{2\gamma_2}|\boldsymbol{s}'_{2,l}|^2+\frac{1}{2}\big(\boldsymbol{y}_{1,l+t_1}-\sqrt{a_2}\,\boldsymbol{s}'_{2,l}\,\big)^2-1\Big)<\frac{\epsilon}{\log e}\Big)\notag\\
&&\stackrel{(a)}{=} \mathbb{P}\Big(\frac{1}{n'}\sum_{l=0}^{n'-1}\Big(\frac{1}{2\gamma_2}|\boldsymbol{s}'_{2,l}|^2+\frac{1}{2}\big(\boldsymbol{s}'_{1,l}+\boldsymbol{z}_{1,l+t_1}-\sqrt{a_2}\,\boldsymbol{s}'_{2,l}\,\big)^2-1\Big)<\frac{\epsilon}{\log e}\Big)\notag\\
&&\stackrel{(b)}{\leq} \mathbb{P}\Big(\frac{1}{n'}\sum_{l=0}^{n'-1}\Big(1+\frac{\gamma_1+a_2\gamma_2}{2}-\frac{1}{2\gamma_2}|\boldsymbol{s}'_{1,l}|^2-\frac{1}{2}\big(\boldsymbol{s}'_{1,l}+\boldsymbol{z}_{1,l+t_1}-\sqrt{a_2}\,\boldsymbol{s}'_{2,l}\,\big)^2\Big)>\frac{\gamma_1+a_2\gamma_2}{2}-\frac{\epsilon}{\log e}\Big)\notag\\
&&\stackrel{(c)}{=}\mathbb{P}\Big(\frac{1}{ n'}\sum_{l=0}^{n'-1}\Big(1-\frac{\frac{1}{2\gamma_2}|\boldsymbol{s}'_{2,l}|^2+\frac{1}{2}\big(\boldsymbol{s}'_{1,l}+\boldsymbol{z}_{1,l+t_1}-\sqrt{a_2}\,\boldsymbol{s}'_{2,l}\,\big)^2}{1+\frac{\gamma_1+a_2\gamma_2}{2}}\Big)>\frac{\frac{\gamma_1+a_2\gamma_2}{2}-\frac{\epsilon}{\log e}}{1+\frac{\gamma_1+a_2\gamma_2}{2}}\Big)\notag\\
&&\stackrel{(d)}{\leq}e^{-\Theta(n')},
\end{eqnarray}
\normalsize
where $(a)$ is due to the fact that  the signal at Rx~1 during time slots $t_1\leq t\leq t_1+n'-1$ is $\boldsymbol{y}_{t}=\boldsymbol{s}'_{1,t}+\boldsymbol{z}_{1,t}$, in $(b)$ we have added $\frac{\gamma_1+a_2\gamma_2}{2}$ to both sides of the inequality in $(a)$ after multiplying both sides of the inequality by $-1$, in $(c)$ both sides are divided by $1+\frac{\gamma_1+a_2\gamma_2}{2}$ and $(d)$ is due to Lemma~\ref{lem_2} (Bernstein's inequality)  where it is assumed that  $\epsilon<\frac{\gamma_1+a_2\gamma_2}{2}\log e$. In fact, the random variables $\boldsymbol{w}_l=1-\frac{\frac{1}{2\gamma_2}|\boldsymbol{s}'_{1,l}|^2+\frac{1}{2}(\boldsymbol{s}'_{1,l}+\boldsymbol{z}_{1,l+t_1}-\sqrt{a_2}\,\boldsymbol{s}'_{2,l}\,)^2}{1+\frac{\gamma_1+a_2\gamma_2}{2}}$ are independent with zero mean and finite variance and $\boldsymbol{w}_l\leq 1$. Therefore, Bernstein's inequality applies.
 \section*{Appendix~F; Proof of Proposition~\ref{prop_3}}
The $j_1^{th}$  burst of Tx~$1$ overlaps with the $j_2^{th}$ burst of Tx~$2$ if and only if one of the events
\begin{equation}
\label{shart1}
\mathcal{E}_n:=\big\{\boldsymbol{\tau}_1^{(j_1)}+1\leq \boldsymbol{\tau}^{(j_2)}_2+1\leq \boldsymbol{\tau}_1^{(j_1)}+n'+n_1\big\}
\end{equation}
or 
\begin{equation}
\label{shart2}
\mathcal{F}_n:=\big\{\boldsymbol{\tau}_1^{(j_1)}+1\leq \boldsymbol{\tau}^{(j_2)}_2+n'+n_2\leq \boldsymbol{\tau}_1^{(j_1)}+n'+n_1\big\}
\end{equation}
holds. Let us show that $\lim_{n\to\infty}\mathbb{P}(\mathcal{E}_n)=0$ if and only if $j_2\mu_2-j_1\mu_1+\nu_2-\nu_1\in(0,\theta_1)$. Define
\begin{equation}
\label{ }
\boldsymbol{\rho}_n=\frac{1}{n}(\boldsymbol{\tau}_{2}^{(j_2)}-\boldsymbol{\tau}_1^{(j_1)}),\,\,\,\boldsymbol{\rho}'_n=\frac{1}{n}(\boldsymbol{\tau}_1^{(j_1)}-\boldsymbol{\tau}_2^{(j_2)}+n'+n_1-1).
\end{equation}
Then
\begin{equation}
\label{ }
\mathcal{E}_n=\big\{\boldsymbol{\rho}_n\geq 0, \boldsymbol{\rho}'_n\geq 0\big\}=\big\{\min\{\boldsymbol{\rho}_n,\boldsymbol{\rho}'_n\}\geq 0\big\}.
\end{equation}
Following similar arguments made in Appendix~D, 
\begin{equation}
\label{ }
\lim_{n\to\infty}\min\{\boldsymbol{\rho}_n,\boldsymbol{\rho}'_n\}=\min\{j_2\mu_2-j_1\mu_1+\nu_2-\nu_1,j_1\mu_1-j_2\mu_2+\nu_1-\nu_2+\theta_1\}.
\end{equation}
If $j_2\mu_2-j_1\mu_1+\nu_2-\nu_1\in(0,\theta_1)$, then $\min\{j_2\mu_2-j_1\mu_1+\nu_2-\nu_1,j_1\mu_1-j_2\mu_2+\nu_1-\nu_2+\theta_1\}> 0$. Hence, $\lim_{n\to\infty}\mathbb{P}(\mathcal{E}_n)=1$ by Lemma~\ref{sim_lem}. Similarly, one can show that if $j_2\mu_2-j_1\mu_1+\nu_2-\nu_1\in(-\theta_2,\theta_1-\theta_2)$, then $\lim_{n\to\infty}\mathbb{P}(\mathcal{F}_n)=1$. It follows that if $j_2\mu_2-j_1\mu_1+\nu_2-\nu_1\in(0,\theta_1)\bigcup(-\theta_2,\theta_1-\theta_2)$, then $\lim_{n\to\infty}\mathbb{P}(\mathcal{E}_n\bigcup\mathcal{F}_n)=1$. Next, assume\footnote{Recall that by (\ref{res_11}), $j_2\mu_2-j_1\mu_1+\nu_2-\nu_1\notin\{0,\theta_1,-\theta_2,\theta_1-\theta_2\}$.} $j_2\mu_2-j_1\mu_1+\nu_2-\nu_1\notin[0,\theta_1]\bigcup[-\theta_2,\theta_1-\theta_2]$. Since $j_2\mu_2-j_1\mu_1+\nu_2-\nu_1\notin[0,\theta_1]$, then $\min\{j_2\mu_2-j_1\mu_1+\nu_2-\nu_1,j_1\mu_1-j_2\mu_2+\nu_1-\nu_2+\theta_1\}<0$ and we have $\lim_{n\to\infty}\mathbb{P}(\mathcal{E}_n)=0$ by Lemma~\ref{sim_lem}. Similarly, $j_2\mu_2-j_1\mu_1+\nu_2-\nu_1\notin[-\theta_2,\theta_1-\theta_2]$ results in $\lim_{n\to\infty}\mathbb{P}(\mathcal{F}_n)=0$. But, $\mathbb{P}(\mathcal{E}_n\bigcup\mathcal{F}_n)\leq \mathbb{P}(\mathcal{E}_n)+\mathbb{P}(\mathcal{F}_n)$ and we get $\lim_{n\to\infty}\mathbb{P}(\mathcal{E}_n\bigcup\mathcal{F}_n)=0$.
Finally, the probability of the $j_{2}^{th}$ burst of Tx~$2$ overlaping only with the preamble sequence in the $j_1^{th}$ burst of Tx~$1$ vanishes as $n$ grows due to Proposition~\ref{prop_2}. This completes the proof.  

\section*{Appendix~G; Proof of Proposition~\ref{prop_4}}
Fix $\epsilon>0$. We assume $i=1$. Given the index $j$ of the codeword of Tx~$1$, let both $\omega^{-}:=\omega^-_{1,j}$ and $\omega^{+}:=\omega^+_{1,j}$ be nonzero, $\omega^-\neq \omega^+$ and $\omega_{1,j}=0$.  The proof can be easily extended to the cases $\omega^-=\omega^+$ or $\omega^-\neq \omega^+,\omega_{1,j}\ge1$. Define the event $\mathcal{U}_n$ by
 \begin{eqnarray}
 \label{PTY11}
\mathcal{U}_n&:=&\left\{\boldsymbol{\tau}_{2}^{(\omega^{-})}+1<\boldsymbol{\tau}_1^{(j)}+n'+1\leq \boldsymbol{\tau}_{2}^{(\omega^{-})}+n'+n_2\right.\notag\\
&&\hskip5cm\left.\leq \boldsymbol{\tau}_{2}^{(\omega^{+})}+1\leq \boldsymbol{\tau}_{1}^{(j)}+n'+n_1+\boldsymbol{\tau}_{2}^{(\omega^{+})}+n'+n_2\right\}.
\end{eqnarray}
The probability of error in decoding the $j^{th}$ codeword of Tx~$1$ at Rx~1 is bounded as 
\begin{equation}
\label{ }
\mathbb{P}(\mathrm{error})\leq \mathbb{P}(\mathrm{error},\mathcal{U}_n)+\mathbb{P}(\mathcal{U}_n^c)\leq \mathbb{P}(\mathrm{error},\mathcal{U}_n)+\epsilon,
\end{equation}
where in the last step we have assumed $n$ is large enough so that $\mathbb{P}(\mathcal{U}_n^c)\leq \epsilon$. This is due to Proposition~\ref{prop_3} together with the fact that $\omega^-,\omega^+\neq 0$. Under the event $\mathcal{U}_n$, error can happen in two ways. The first case is when at least one of (\ref{mid1}), (\ref{mid2}) or (\ref{mid3}) is not satisfied for the actual transmitted codeword by Tx~$1$. We denote this error event by $\mathrm{error}_1$. The second case is when all of (\ref{mid1}), (\ref{mid2}) and (\ref{mid3}) are satisfied for a codeword that is different from the transmitted codeword by Tx~$1$. We denote this error event by $\mathrm{error}_2$. Then 
\begin{equation}
\label{errevr2}
\mathbb{P}(\mathrm{error},\mathcal{U}_n)\leq \mathbb{P}(\mathrm{error}_1,\mathcal{U}_n)+\mathbb{P}(\mathrm{error}_2,\mathcal{U}_n).
\end{equation}
Next, we address the two terms on the right side of (\ref{errevr2}) separately: 
\begin{itemize}
  \item The term $ \mathbb{P}(\mathrm{error}_1,\mathcal{U}_n)$: Here, we verify that (\ref{mid1}) occurs with high probability for the actual transmitted codeword in the asymptote of large $n$.  One can establish a similar result for (\ref{mid2}) and (\ref{mid3}) yielding $\lim_{n\to\infty}\mathbb{P}(\mathrm{error}_1,\mathcal{U}_n)=0$. Define the set $\widetilde{\mathcal{U}}_n$ by
  \begin{eqnarray}
  \label{mathb11}
\widetilde{\mathcal{U}}_n&:=&\{(t_1,t_2,t_3)\in \mathds{Z}^3: t_2+1<t_1+n'+1\leq t_2+n'+n_2\notag\\&&\hskip6.5cm\leq t_3+1\leq t_1+n'+n_1<t_3+n'+n_2\}.
\end{eqnarray}
Then 
\begin{eqnarray}
\mathcal{U}_n=\left\{ (\boldsymbol{\tau}_1^{(j)}, \boldsymbol{\tau}_{2}^{(\omega^{-})},\boldsymbol{\tau}_{2}^{(\omega^{+})})\in \widetilde{\mathcal{U}}_n\right\}.
\end{eqnarray}
 Assume $(\boldsymbol{s}_{1,l})_{l=0}^{n_1-1}$ is the $j^{th}$ codeword sent by Tx~$1$. The probability that (\ref{mid1}) does not occur for the actual transmitted codeword under $\mathcal{U}_n$ can be written as
     \begin{eqnarray}
     \label{yhnk11}
&&\mathbb{P}\left(\Big((\boldsymbol{s}_{1,l})_{l=0}^{\boldsymbol{\tau}_2^{(\omega^{-})}-\boldsymbol{\tau}_1^{(j)}+n_2-1},(\boldsymbol{y}_{1,l})_{l=\boldsymbol{\tau}_1^{(j)}+n'+1}^{\boldsymbol{\tau}_2^{(\omega^{-})}+n'+n_2}\Big)\notin A_{\epsilon}^{(\boldsymbol{\tau}_2^{(\omega^{-})}-\boldsymbol{\tau}_1^{(j)}+n_2)}[p^{(2)}],\,\mathcal{U}_n\right)\notag\\
&&=\sum_{(t_1,t_2,t_3)\in\,\widetilde{\mathcal{U}}_n}\mathbb{P}\left(\big((\boldsymbol{s}_{1,l})_{l=0}^{t_2-t_1+n_2-1},(\boldsymbol{y}_{1,l})_{l=t_1+n'+1}^{t_2+n'+n_2}\big)\notin A_{\epsilon}^{(t_2-t_1+n_2)}[p^{(2)}]\right)\notag\\
&&\hskip7cm \times \mathbb{P}\big((\boldsymbol{\tau}_1^{(j)}, \boldsymbol{\tau}_{2}^{(\omega^{-})},\boldsymbol{\tau}_{2}^{(\omega^{+})})=(t_1,t_2,t_3)\big).
\end{eqnarray}
For any $(t_1,t_2,t_3)\in \widetilde{\mathcal{U}}_n$, 
\begin{eqnarray}
\label{seq11}
&&\mathbb{P}\left(\big((\boldsymbol{s}_{1,l})_{l=0}^{t_2-t_1+n_2-1},(\boldsymbol{y}_{1,l})_{l=t_1+n'+1}^{t_2+n'+n_2}\big)\notin A_{\epsilon}^{(t_2-t_1+n_2)}[p^{(2)}]\right)\notag\\
&&\leq \mathbb{P}\bigg(\Big|\frac{1}{t_2-t_1+n_2}\sum_{l=0}^{t_2-t_1+n_2-1}\log p^{(2)}(\boldsymbol{s}_{1,l},\boldsymbol{y}_{1,l+t_1+n'+1})+h(p^{(2)})\Big|>\epsilon\bigg)\notag\\
&&\,\,\,\,\,\,\,+\mathbb{P}\bigg(\Big|\frac{1}{t_2-t_1+n_2}\sum_{l=0}^{t_2-t_1+n_2-1}\log p^{(2)}_{1}(\boldsymbol{s}_{1,l})+h(p^{(2)}_1)\Big|>\epsilon\bigg)\notag\\
&&\,\,\,\,\,\,\,+\mathbb{P}\bigg(\Big|\frac{1}{t_2-t_1+n_2}\sum_{l=0}^{t_2-t_1+n_2-1}\log p^{(2)}_{2}(\boldsymbol{y}_{1,l+t_1+n'+1})+h(p^{(2)}_2)\Big|>\epsilon\bigg),
\end{eqnarray}
where $p^{(2)}_1$ and $p^{(2)}_2$ are the first and second marginals of $p^{(2)}$, respectively. The three terms on the right side of (\ref{seq11}) can be treated similarly. Here, we only study the first term.  Let us write 
\begin{eqnarray}
\label{seq22}
&&\mathbb{P}\bigg(\Big|\frac{1}{t_2-t_1+n_2}\sum_{l=0}^{t_2-t_1+n_2-1}\log p^{(2)}(\boldsymbol{s}_{1,l},\boldsymbol{y}_{1,l+t_1+n'+1})+h(p^{(2)})\Big|>\epsilon\bigg)\notag\\
&&\leq \mathbb{P}\Big(\frac{1}{t_2-t_1+n_2}\sum_{l=0}^{t_2-t_1+n_2-1}\log p^{(2)}(\boldsymbol{s}_{1,l},\boldsymbol{y}_{1,l+t_1+n'+1})+h(p^{(2)})>\epsilon\Big)\notag\\
&&\,\,\,\,\,\,\,+\mathbb{P}\Big(\frac{1}{t_2-t_1+n_2}\sum_{l=0}^{t_2-t_1+n_2-1}\log p^{(2)}(\boldsymbol{s}_{1,l},\boldsymbol{y}_{1,l+t_1+n'+1})+h(p^{(2)})<-\epsilon\Big).
\end{eqnarray}
The random variables $\log p^{(2)}(\boldsymbol{s}_{1,l},\boldsymbol{y}_{1,l+t_1+n'+1})$ for $0\leq l\leq t_2-t_1+n_2-1$ are independent and identically distributed with expectation $-h(p^{(2)})$. Using Chernoff's bounding technique \cite{16} and for $r>0$, we can find an upper bound on the first term on the right side of (\ref{seq22}) as
 \begin{eqnarray}
 \label{siu11}
&&\mathbb{P}\Big(\frac{1}{t_2-t_1+n_2}\sum_{l=0}^{t_2-t_1+n_2-1}\log p^{(2)}(\boldsymbol{s}_{1,l},\boldsymbol{y}_{1,l+t_1+n'+1})+h(p^{(2)})>\epsilon\Big)\notag\\
&&\leq 2^{-r(t_2-t_1+n_2)(\epsilon-h(p^{(2)}))}\left(\mathds{E}\left[2^{r\log p^{(2)}(\boldsymbol{s}_{1,0},\boldsymbol{y}_{1,t_1+n'+1})}\right]\right)^{t_2-t_1+n_2}\notag\\
&&=2^{-r(t_2-t_1+n_2)(\epsilon-h(p^{(2)}))}\left(\mathds{E}\left[ \left(p^{(2)}(\boldsymbol{s}_{1,0},\boldsymbol{y}_{1,t_1+n'+1})\right)^r\right]\right)^{t_2-t_1+n_2}.
\end{eqnarray}
For notational simplicity and with a slight abuse of notation, let us write $\boldsymbol{y}_{1,t_1+n'+1}=\boldsymbol{s}_{1,0}+\sqrt{a_2}\,\boldsymbol{s}_{2}+\boldsymbol{z}_{1}$ where $\boldsymbol{s}_{2}\sim \mathrm{N}(0,\gamma_2)$ is a symbol of the $\omega^{-\,\,th}$ transmitted codeword by Tx~$2$ and $\boldsymbol{z}_1:=\boldsymbol{z}_{1,t_1+n'+1}\sim\mathrm{N}(0,1)$. We have 
\begin{eqnarray}
p^{(2)}(\boldsymbol{s}_{1,0},\boldsymbol{y}_{1,t_1+n'+1})&=&\mathrm{g}(\boldsymbol{s}_{1,0};\gamma_1)\mathrm{g}(\boldsymbol{y}_{1,t_1+n'+1}-\boldsymbol{s}_{1,0};1+a_2\gamma_2)\notag\\
&=&\mathrm{g}(\boldsymbol{s}_{1,0};\gamma_1)\mathrm{g}(\sqrt{a_2}\,\boldsymbol{s}_{2}+\boldsymbol{z}_{1};1+a_2\gamma_2)
\end{eqnarray}
and 
\begin{eqnarray}
\label{siu22}
\mathds{E}\left[ \left(p^{(2)}(\boldsymbol{s}_{1,0},\boldsymbol{y}_{1,t_1+n'+1})\right)^r\right]&=&\mathds{E}\left[\left(\mathrm{g}(\boldsymbol{s}_{1,0};\gamma_1)\right)^r\left(\mathrm{g}(\sqrt{a_2}\,\boldsymbol{s}_{2}+\boldsymbol{z}_{1};1+a_2\gamma_2)\right)^r\right]\notag\\
&=&\mathds{E}\left[\left(\mathrm{g}(\boldsymbol{s}_{1,0};\gamma_1)\right)^r\right]\mathds{E}\left[\left(\mathrm{g}(\sqrt{a_2}\,\boldsymbol{s}_{2}+\boldsymbol{z}_{1};1+a_2\gamma_2)\right)^r\right]\notag\\
&=&\frac{1}{(1+r)\left(2\pi\sqrt{\gamma_1(1+a_2\gamma_2)}\right)^r}\notag\\
&=&\frac{e^r}{(1+r)2^{rh(p^{(2)})}},
\end{eqnarray}
where the penultimate step is due to the fact that for $\boldsymbol{x}\sim\mathrm{N}(0,1)$ and any $u>-\frac{1}{2}$, we have $\mathds{E}[e^{-u\boldsymbol{x}^2}]=\frac{1}{\sqrt{1+2u}}$ and the last step is due to $h(p^{(2)})=\log(2\pi e\sqrt{\gamma_1(1+a_2\gamma_2)})$. By (\ref{siu11}) and (\ref{siu22}), 
\begin{eqnarray}
\label{siu333}
\mathbb{P}\Big(\frac{1}{t_2-t_1+n_2}\sum_{l=0}^{t_2-t_1+n_2-1}\log p^{(2)}(\boldsymbol{s}_{1,l},\boldsymbol{y}_{1,l+t_1+n'+1})+h(p^{(2)})>\epsilon\Big)\leq e^{-(t_2-t_1+n_2)u(r)},
\end{eqnarray}
where
\begin{equation}
\label{ }
u(r):=\ln(1+r)-r(1-\epsilon\ln 2),\,\,\,r>0.
\end{equation}
 Following a similar approach that led us to (\ref{siu333}), one can show that the second term on the right side of (\ref{seq22}) is bounded from above as 
\begin{eqnarray}
\label{siu444}
\mathbb{P}\Big(\frac{1}{t_2-t_1+n_2}\sum_{l=0}^{t_2-t_1+n_2-1}\log p^{(2)}(\boldsymbol{s}_{1,l},\boldsymbol{y}_{1,l+t_1+n'+1})+h(p^{(2)})<-\epsilon\Big)\leq e^{-(t_2-t_1+n_2)v(r)}
\end{eqnarray}
where
\begin{equation}
\label{ }
v(r):=\ln(1-r)+r(1+\epsilon\ln 2),\,\,\, 0<r<1.
\end{equation}
Regardless of the value of $\epsilon>0$, there always exists an $0<r_0<1$ such that $u(r), v(r)>0$ for $0<r<r_0$. It is understood that we take $r$ inside $(0,r_0)$. It is easy to see that for any $0<r<1$, $v(r)<u(r)$.\footnote{We have $v(r)<u(r)$ if and only if $w(r):=\ln(1+r)-\ln(1-r)-2r>0$. Note that $w(0)=0$ and $\frac{\mathrm{d}w}{\mathrm{d}r}=\frac{2r^2}{1-r^2}>0$ for any $0<r<1$. Therefore, $w(r)>0$ for any $0<r<1$ by the mean value theorem.} Using this fact together with (\ref{seq22}), (\ref{siu333}) and (\ref{siu444}),  
\begin{eqnarray}
\mathbb{P}\bigg(\Big|\frac{1}{t_2-t_1+n_2}\sum_{l=0}^{t_2-t_1+n_2-1}\log p^{(2)}(\boldsymbol{s}_{1,l},\boldsymbol{y}_{1,l+t_1+n'+1})+h(p^{(2)})\Big|>\epsilon\bigg)\leq 2e^{-(t_2-t_1+n_2)v(r)}.
\end{eqnarray}
 It can be shown similarly that the second and third terms on the right side of (\ref{seq11}) are bounded from above by $2e^{-(t_2-t_1+n_2)v(r)}$. Therefore, 
 \begin{eqnarray}
 \label{yhnk22}
\mathbb{P}\left(\Big((\boldsymbol{s}_{1,l})_{l=0}^{t_2-t_1+n_2-1},(\boldsymbol{y}_{1,l})_{l=t_1+n'+1}^{t_2+n'+n_2}\Big)\notin A_{\epsilon}^{(t_2-t_1+n_2)}[p^{(2)}]\right)\leq 6e^{-(t_2-t_1+n_2)v(r)}.
\end{eqnarray}
Using (\ref{yhnk22}) in (\ref{yhnk11}), 
\begin{eqnarray}
\label{bool11}
&&\mathbb{P}\Big(\big((\boldsymbol{s}_{1,l})_{l=0}^{\boldsymbol{\tau}_2^{(\omega^{-})}-\boldsymbol{\tau}_1^{(j)}+n_2-1},(\boldsymbol{y}_{1,l})_{l=\boldsymbol{\tau}_1^{(j)}+n'+1}^{\boldsymbol{\tau}_2^{(\omega^{-})}+n'+n_2}\big)\notin A_{\epsilon}^{(\boldsymbol{\tau}_2^{(\omega^{-})}-\boldsymbol{\tau}_1^{(j)}+n_2)}[p^{(2)}],\,\mathcal{U}_n\Big)\notag\\
&&\leq \sum_{(t_1,t_2,t_3)\in\,\widetilde{\mathcal{U}}_n}6e^{-(t_2-t_1+n_2)v(r)}\mathbb{P}((\boldsymbol{\tau}_1^{(j)}, \boldsymbol{\tau}_{2}^{(\omega^{-})},\boldsymbol{\tau}_{2}^{(\omega^{+})})=(t_1,t_2,t_3))\notag\\
&&\stackrel{(a)}{\leq}\sum_{(t_1,t_2,t_3)\in\mathds{Z}^3}6e^{-(t_2-t_1+n_2)v(r)}\mathbb{P}((\boldsymbol{\tau}_1^{(j)}, \boldsymbol{\tau}_{2}^{(\omega^{-})},\boldsymbol{\tau}_{2}^{(\omega^{+})})=(t_1,t_2,t_3))\notag\\
&&=6\mathds{E}\big[e^{-(\boldsymbol{\tau}_{2}^{(\omega^{-})}-\boldsymbol{\tau}_{1}^{(j)}+n_2)v(r)}\big]\notag\\
&&=6e^{-(\lfloor n \nu_2\rfloor-\lfloor n \nu_1\rfloor+n_2)v(r)}\mathds{E}\big[e^{-(\boldsymbol{\xi}_{2}^{(\omega^{-})}-\boldsymbol{\xi}_{1}^{(j)})v(r)}\big]\notag\\
&&=6e^{-(\lfloor n \nu_2\rfloor-\lfloor n \nu_1\rfloor+n_2)v(r)}\mathds{E}\big[e^{-v(r)\boldsymbol{\xi}_{2}^{(\omega^{-})}}\big]\mathds{E}\big[e^{v(r)\boldsymbol{\xi}_{1}^{(j)})}\big],
\end{eqnarray}
where in $(a)$ we have removed the constraint $(t_1,t_2,t_3)\in\widetilde{\mathcal{U}}_n$ and the last step is due to independence of $\boldsymbol{\xi}_{2}^{(\omega^{-})}$ and $\boldsymbol{\xi}_{1}^{(j)}$. Recalling the expression for the moment generating function of a negative Binomial  random variable\footnote{\label{ft11}The moment generating function of $\boldsymbol{x}\sim\mathrm{NB}(n,p)$ is given by $\mathds{E}[e^{t\boldsymbol{x}}]=\left(\frac{pe^t}{1-(1-p)e^t}\right)^n$ for $t<-\ln(1-p)$.}, we get 
\begin{equation}
\label{bool2}
\mathds{E}\Big[e^{-v(r)\boldsymbol{\xi}_{2}^{(\omega^{-})}}\Big]=\left(\frac{q_2e^{-v(r)}}{1-(1-q_2)e^{-v(r)}}\right)^{\frac{\omega^{-}\lfloor n\eta_2\rfloor}{k_2}}
\end{equation}
and
\begin{equation}
\label{bool3}
\mathds{E}\Big[e^{v(r)\boldsymbol{\xi}_{1}^{(j)}}\Big]=\left(\frac{q_1e^{v(r)}}{1-(1-q_1)e^{v(r)}}\right)^{\frac{j\lfloor n\eta_1\rfloor}{k_1}},
\end{equation}
where (\ref{bool3}) holds as long as $v(r)<-\ln(1-q)$. Since $\lim_{r\to0^+}v(r)=0$, one can make sure the constraint $v(r)<-\ln(1-q)$ holds by choosing $r$ small enough. By (\ref{bool11}), (\ref{bool2}) and (\ref{bool3}),
\begin{eqnarray}
\label{bool4}
&&\mathbb{P}\Big(\big((\boldsymbol{s}_{1,l})_{l=0}^{\boldsymbol{\tau}_2^{(\omega^{-})}-\boldsymbol{\tau}_1^{(j)}+n_2-1},(\boldsymbol{y}_{1,l})_{l=\boldsymbol{\tau}_1^{(j)}+n'+1}^{\boldsymbol{\tau}_2^{(\omega^{-})}+n'+n_2}\big)\notin A_{\epsilon}^{(\boldsymbol{\tau}_2^{(\omega^{-})}-\boldsymbol{\tau}_1^{(j)}+n_2)}[p^{(2)}],\,\mathcal{U}_n\Big)\notag\\
&&\leq 6e^{-(\lfloor n \nu_2\rfloor-\lfloor n \nu_1\rfloor+n_2)v(r)}\left(\frac{q_2e^{-v(r)}}{1-(1-q_2)e^{-v(r)}}\right)^{\frac{\omega^{-}\lfloor n\eta_2\rfloor}{k_2}}\left(\frac{q_1e^{v(r)}}{1-(1-q_1)e^{v(r)}}\right)^{\frac{j\lfloor n\eta_1\rfloor}{k_1}}.
\end{eqnarray}
Using the identity 
\begin{eqnarray}
\ln\frac{ae^{x}}{1-(1-a)e^{x}}=\frac{x}{a}+o(x)
\end{eqnarray}
for $0<a<1$, one can write (\ref{bool4}) as 
\begin{eqnarray}
\label{bool5}
&&\mathbb{P}\Big(\big((\boldsymbol{s}_{1,l})_{l=0}^{\boldsymbol{\tau}_2^{(\omega^{-})}-\boldsymbol{\tau}_1^{(j)}+n_2-1},(\boldsymbol{y}_{1,l})_{l=\boldsymbol{\tau}_1^{(j)}+n'+1}^{\boldsymbol{\tau}_2^{(\omega^{-})}+n'+n_2}\big)\notin A_{\epsilon}^{(\boldsymbol{\tau}_2^{(\omega^{-})}-\boldsymbol{\tau}_1^{(j)}+n_2)}[p^{(2)}],\,\mathcal{U}_n\Big)\notag\\
&&\leq 6e^{-(\lfloor n \nu_2\rfloor-\lfloor n \nu_1\rfloor+n_2)v(r)}e^{\frac{\omega^{-}\lfloor n\eta_2\rfloor}{k_2}(-\frac{v(r)}{q_2}+o(v(r)))}e^{\frac{j\lfloor n\eta_1\rfloor}{k_1}(\frac{v(r)}{q_1}+o(v(r)))}\notag\\
&&=e^{-nv(r)f(n)},
\end{eqnarray}
where
\begin{eqnarray}
f(n):=\frac{1}{n}\left(\frac{\omega^{-}\lfloor n\eta_2\rfloor}{\lambda_2}-\frac{j\lfloor n\eta_1\rfloor}{\lambda_1}+\lfloor n \nu_2\rfloor-\lfloor n \nu_1\rfloor+n_2-(\lfloor n\eta_1\rfloor+\lfloor n\eta_2\rfloor)\frac{o(v(r))}{v(r)}\right).
\end{eqnarray}
But, $\lim_{n\to\infty}f(n)=\omega^{-}\mu_2-j\mu_1+\nu_2-\nu_1+\theta_2+ (\eta_1+\eta_2)\frac{o(v(r))}{v(r)}$. By definition, $\omega^{-}$ satisfies $\omega^{-}\mu_2-j\mu_1+\nu_2-\nu_1+\theta_2>0$. Since $ \frac{o(v(r))}{v(r)}$ can be made arbitrarily small by choosing $r$ small enough, we conclude that $\lim_{n\to\infty}f(n)>0$. This together with (\ref{bool5}) implies that $e^{-nv(r)f(n)}$ decays exponentially with $n$ as desired. 
  \item The term $ \mathbb{P}(\mathrm{error}_2,\,\mathcal{U}_n)$: 
  Let $\delta>0$ and define 
  \begin{eqnarray}
\mathcal{V}_n:=\bigg\{\max\Big\{\Big|\frac{\boldsymbol{\tau}_{2}^{(\omega^-)}}{n}-(\omega^-\mu_2+\nu_2)\Big|, \Big|\frac{\boldsymbol{\tau}_{2}^{(\omega^+)}}{n}-(\omega^+\mu_2+\nu_2)\Big|\Big\}<\delta\bigg\}.
\end{eqnarray}
Also define $\widetilde{\mathcal{V}}_n$ by  
 \begin{eqnarray}
&&\widetilde{\mathcal{V}}_n:=\left\{(t_2,t_3)\in\mathbb{Z}^2: \max\Big\{\Big|\frac{t_2}{n}-(\omega^-\mu_2+\nu_2)\Big| , \Big|\frac{t_3}{n}-(\omega^+\mu_2+\nu_2)\Big|\Big\}<\delta\right\}.
\end{eqnarray}
We can write
\begin{eqnarray}
\mathbb{P}(\mathrm{error}_2,\mathcal{U}_n)\leq \mathbb{P}(\mathrm{error}_2,\,\mathcal{U}_n, \mathcal{V}_n)+\mathbb{P}(\mathcal{V}_n^{c}).
\end{eqnarray}
By SLLN, $\lim_{n\to\infty}\frac{\boldsymbol{\tau}_2^{(\omega^{-})}}{n}=\omega^-\mu_2+\nu_2$. Therefore, $\frac{\boldsymbol{\tau}_2^{(\omega^-)}}{n}$ also converges to $\omega^-\mu_2+\nu_2$ in probability and one can select $n$ large enough so that 
 \begin{eqnarray}
 \label{bial_11}
\mathbb{P}\Big(\Big|\frac{\boldsymbol{\tau}_2^{(\omega^-)}}{n}-(\omega^-\mu_2+\nu_2)\Big|\geq \delta\Big)<\epsilon/3.\end{eqnarray}
Similarly,  
\begin{eqnarray}
\mathbb{P}\Big(\Big|\frac{\boldsymbol{\tau}_2^{(\omega^+)}}{n}-(\omega^+\mu_2+\nu_2)\Big|\geq \delta\Big)<\epsilon/3
\end{eqnarray}
holds for large enough $n$. It follows that if $n$ is sufficiently large, then $\mathbb{P}(\mathcal{V}^c_n)<\epsilon$. To find an upper bound on $ \mathbb{P}(\mathrm{error}_2,\mathcal{U}_n, \mathcal{V}_n)$, let us label the messages of Tx~$1$ as message~$1$ to message~$2^{\lfloor n\eta_1\rfloor}$. Assume the $j^{th}$ transmitted message of Tx~$1$ is message~$1$ and $(\boldsymbol{s}_{1,l})_{l=0}^{n_1-1}$ is the codeword of user~1 corresponding to message~$2$. Recall that $ \mathbb{P}(\mathrm{error}_2,\mathcal{U}_n,\mathcal{V}_n)$ is the probability of the event that  a codeword different from the transmitted codeword satisfies (\ref{mid1}), (\ref{mid2}) and $(\ref{mid3})$. Then 
  \begin{eqnarray}
  \label{mas33}
\mathbb{P}(\mathrm{error}_2,\mathcal{U}_n,\mathcal{V}_n) \leq 2^{\lfloor n\eta_1\rfloor}\mathbb{P}\left(\mathcal{E}\bigcap\mathcal{F}\bigcap\mathcal{G}\bigcap \mathcal{U}_n\bigcap\mathcal{V}_n\right),
\end{eqnarray}
where 
\begin{eqnarray}
\label{mid111}
\mathcal{E}=\left\{\Big((\boldsymbol{s}_{1,l})_{l=0}^{\boldsymbol{\tau}_2^{(\omega^{-})}-\boldsymbol{\tau}_1^{(j)}+n_2-1},(\boldsymbol{y}_{1,l})_{l=\boldsymbol{\tau}_1^{(j)}+n'+1}^{\boldsymbol{\tau}_2^{(\omega^{-})}+n'+n_2}\Big)\in A_{\epsilon}^{(\boldsymbol{\tau}_2^{(\omega^{-})}-\boldsymbol{\tau}_1^{(j)}+n_2)}[p^{(2)}]\right\},
\end{eqnarray} 
\begin{eqnarray}
\label{mid222}
\mathcal{F}=\left\{\Big((\boldsymbol{s}_{1,l})_{l=\boldsymbol{\tau}_2^{(\omega^{-})}-\boldsymbol{\tau}_1^{(j)}+n_2}^{\boldsymbol{\tau}_2^{(\omega^{+})}-\boldsymbol{\tau}_1^{(j)}-n'-1},(\boldsymbol{y}_{1,l})_{l=\boldsymbol{\tau}_2^{(\omega^{-})}+n'+n_2+1}^{\boldsymbol{\tau}_2^{(\omega^{+})}}\Big)\in A_{\epsilon}^{(\boldsymbol{\tau}_2^{(\omega^{+})}-\boldsymbol{\tau}_2^{(\omega^{-})}-n'-n_2)}[p^{(1)}]\right\}
\end{eqnarray}
and 
\begin{eqnarray}
\label{mid333}
\mathcal{G}=\left\{\Big((\boldsymbol{s}_{1,l})_{l=\boldsymbol{\tau}_2^{(\omega^{+})}-\boldsymbol{\tau}_1^{(j)}-n'}^{n_1-1},(\boldsymbol{y}_{1,l})_{l=\boldsymbol{\tau}_2^{(\omega^{+})}+1}^{\boldsymbol{\tau}_1^{(j)}+n'+n_1}\Big)\in A_{\epsilon}^{(\boldsymbol{\tau}_1^{(j)}-\boldsymbol{\tau}_2^{(\omega^{+})}+n'+n_1)}[p^{(2)}]\right\}.
\end{eqnarray}
Recalling the definition of $\mathcal{V}_n$ in (\ref{mathb11}), 
\begin{eqnarray}
\label{pok}
&&\mathbb{P}\left(\mathcal{E}\bigcap\mathcal{F}\bigcap\mathcal{G}\bigcap \mathcal{U}_n\bigcap\mathcal{V}_n\right)\notag\\&&=\sum_{\substack{(t_1,t_2,t_3)\in\,\widetilde{\mathcal{U}}_n\\(t_2,t_3)\in\widetilde{\mathcal{V}}_n}}\mathbb{P}\left(\mathcal{E}\bigcap\mathcal{F}\bigcap\mathcal{G}\,\Big|\,(\boldsymbol{\tau}_1^{(j)},\boldsymbol{\tau}_2^{(\omega^{-})},\boldsymbol{\tau}_2^{(\omega^{+})})=(t_1,t_2,t_3)\right)\notag\\&&\hskip4cm\times\mathbb{P}\big((\boldsymbol{\tau}_1^{(j)},\boldsymbol{\tau}_2^{(\omega^{-})},\boldsymbol{\tau}_2^{(\omega^{+})})=(t_1,t_2,t_3)\big).
\end{eqnarray}
For any $(t_1,t_2,t_3)\in\widetilde{\mathcal{U}}_n$,
\begin{eqnarray}
\label{enuy11}
\mathbb{P}\left(\mathcal{E}\bigcap\mathcal{F}\bigcap\mathcal{G}\,\Big|\,(\boldsymbol{\tau}_1^{(j)},\boldsymbol{\tau}_2^{(\omega^{-})},\boldsymbol{\tau}_2^{(\omega^{+})})=(t_1,t_2,t_3)\right)=\mathbb{P}(\mathcal{E}_{t_1,t_2})\mathbb{P}(\mathcal{F}_{t_1,t_2,t_3})\mathbb{P}(\mathcal{G}_{t_1,t_3}),
\end{eqnarray}
where
\begin{eqnarray}
\label{mid111}
\mathcal{E}_{t_1,t_2}=\left\{\big((\boldsymbol{s}_{1,l})_{l=0}^{t_2-t_1+n_2-1},(\boldsymbol{y}_{1,l})_{l=t_1+n'+1}^{t_2+n'+n_2}\big)\in A_{\epsilon}^{(t_2-t_1+n_2)}[p^{(2)}]\right\},
\end{eqnarray} 
\begin{eqnarray}
\label{mid222}
\mathcal{F}_{t_1,t_2,t_3}=\left\{\big((\boldsymbol{s}_{1,l})_{l=t_2-t_1+n_2}^{t_3-t_1-n'-1},(\boldsymbol{y}_{1,l})_{l=t_2+n'+n_2+1}^{t_3}\big)\in A_{\epsilon}^{(t_3-t_2-n'-n_2)}[p^{(1)}]\right\}
\end{eqnarray}
and 
\begin{eqnarray}
\label{mid333}
\mathcal{G}_{t_1,t_3}=\left\{\big((\boldsymbol{s}_{1,l})_{l=t_3-t_1-n'}^{n_1-1},(\boldsymbol{y}_{1,l})_{l=t_3+1}^{t_1+n'+n_1}\big)\in A_{\epsilon}^{(t_1-t_3+n'+n_1)}[p^{(2)}]\right\}.
\end{eqnarray}
The reason behind (\ref{enuy11}) is that fixing $(\boldsymbol{\tau}_1^{(j)},\boldsymbol{\tau}_2^{(\omega^{-})},\boldsymbol{\tau}_2^{(\omega^{+})})=(t_1,t_2,t_3)$, the events $\mathcal{E}$, $\mathcal{F}$ and $\mathcal{G}$ are independent as they depend on non-overlapping segments of the sequences $(\boldsymbol{s}_{1,l})_{l=0}^{n_1-1}$ and $(\boldsymbol{y}_{1,l})_{l=t_1+n'+1}^{t_1+n'+n_1}$. Using the standard properties of jointly typical sequences \cite{16}, we have 
\begin{eqnarray}
\label{pok11}
\mathbb{P}(\mathcal{E}_{t_1,t_2})\leq 2^{-(t_2-t_1+n_2)( \psi_1-3\epsilon)},
\end{eqnarray}  
\begin{eqnarray}
\label{pok22}
\mathbb{P}(\mathcal{F}_{t_1,t_2,t_3})\leq 2^{-(t_3-t_2-n'-n_2)(\phi_1-3\epsilon)}
\end{eqnarray}
and
\begin{eqnarray}
\label{pok33}
\mathbb{P}(\mathcal{G}_{t_1,t_3})\leq 2^{-(t_1-t_3+n'+n_1)( \psi_1-3\epsilon)}.
\end{eqnarray}
By (\ref{pok}), (\ref{enuy11}), (\ref{pok11}), (\ref{pok22}) and (\ref{pok33}), 
\begin{eqnarray}
\label{mas22}
&&\mathbb{P}\left(\mathcal{E}\bigcap\mathcal{F}\bigcap\mathcal{G}\bigcap \mathcal{U}_n\bigcap\mathcal{V}_n\right)\notag\\&&\leq2^{-n'( \psi_1-\phi_1)}2^{-n_2(\psi_1-\phi_1)}2^{-n_1( \psi_1-3\epsilon)}\notag\\&&\hskip2cm\times \sum_{\substack{(t_1,t_2,t_3)\in\,\widetilde{\mathcal{U}}_n\\(t_2,t_3)\in\widetilde{\mathcal{V}}_n}}2^{-(t_3-t_2)(\phi_1-\psi_1)}\mathbb{P}\big((\boldsymbol{\tau}_1^{(j)},\boldsymbol{\tau}_2^{(\omega^{-})},\boldsymbol{\tau}_2^{(\omega^{+})})=(t_1,t_2,t_3)\big).
\end{eqnarray}
We can write
\begin{eqnarray}
\label{mas11}
&&\sum_{\substack{(t_1,t_2,t_3)\in\,\widetilde{\mathcal{U}}_n\\(t_2,t_3)\in\widetilde{\mathcal{V}}_n}}2^{-(t_3-t_2)(\phi_1-\psi_1)}\mathbb{P}\big((\boldsymbol{\tau}_1^{(j)},\boldsymbol{\tau}_2^{(\omega^{-})},\boldsymbol{\tau}_2^{(\omega^{+})})=(t_1,t_2,t_3)\big)\notag\\
&&\stackrel{(a)}{\leq}2^{-n(\omega^+\mu_2+\nu_2-\delta-\omega^-\mu_2-\nu_2-\delta)(\phi_1-\psi_1)} \sum_{\substack{(t_1,t_2,t_3)\in\,\widetilde{\mathcal{U}}_n\\(t_2,t_3)\in\widetilde{\mathcal{V}}_n}}\mathbb{P}\big((\boldsymbol{\tau}_1^{(j)},\boldsymbol{\tau}_2^{(\omega^{-})},\boldsymbol{\tau}_2^{(\omega^{+})})=(t_1,t_2,t_3)\big)\notag\\ 
&&\stackrel{(b)}{\leq}2^{-n(\mu_2-2\delta)(\phi_1-\psi_1)},
\end{eqnarray}
where $(a)$ is due to the fact that if $(t_2,t_3)\in \widetilde{\mathcal{V}}_n$, then $t_2\leq n(\omega^-\mu_2+\nu_2+\delta)$ and $t_3\geq n(\omega^+\mu_2+\nu_2-\delta)$ and $(b)$ is due to $\sum_{(t_1,t_2,t_3)\in\,\widetilde{\mathcal{U}}_n,\,\,(t_2,t_3)\in\widetilde{\mathcal{V}}_n}\mathbb{P}\big((\boldsymbol{\tau}_1^{(j)},\boldsymbol{\tau}_2^{(\omega^{-})},\boldsymbol{\tau}_2^{(\omega^{+})})=(t_1,t_2,t_3)\big)\leq 1$ and the fact that $\omega^+-\omega^-=1$. By (\ref{mas33}), (\ref{mas22}) and (\ref{mas11}), 
\begin{eqnarray}
\mathbb{P}(\mathrm{error}_2,\,\mathcal{U}_n,\mathcal{V}_n)\leq 2^{-nf(n)}, 
\end{eqnarray}
where
\begin{eqnarray}
f(n):=(\mu_2-2\delta)(\phi_1-\psi_1)+\frac{n_2}{n}(\psi_1-\phi_1)+\frac{n_1}{n}(\psi_1-3\epsilon)+\frac{n'}{n}(\psi_1-\phi_1)-\frac{\lfloor n\eta_1\rfloor}{n}.
\end{eqnarray}
We have $\lim_{n\to\infty}f(n)=\mu_2(\phi_1-\psi_1)+\theta_2( \psi_1-\phi_1)+\theta_1 \psi_1-\eta_1-2(\phi_1-\psi_1)\delta-3\theta_1\epsilon$. By~(\ref{help0}), $\mu_2(\phi_1-\psi_1)+\theta_2( \psi_1-\phi_1)+\theta_1 \psi_1-\eta_1>0$. Therefore, $\lim_{n\to\infty}f(n)>0$ for sufficiently small $\delta$~and~$\epsilon$  and $\mathbb{P}(\mathrm{error}_2,\,\mathcal{U}_n,\mathcal{V}_n)$ decays exponentially with $n$ as desired.
\end{itemize}
\section*{Appendix~H; Proof of Proposition~\ref{prop_5}}
Given the index $j$ of the codeword of Tx~$1$, we assume $\omega^{-}\neq 0$, $\omega^{+}=0$ and $\omega_{1,j}=0$. The proof can be easily extended for arbitrary $\omega_{1,j}\geq 1$. Define $\mathcal{U}_n$ by
 \begin{eqnarray}
\mathcal{U}_n:=\left\{\boldsymbol{\tau}_2^{(\omega^-)}+1<\boldsymbol{\tau}_1^{(j)}+n'+1\leq \boldsymbol{\tau}_{2}^{(\omega^{-})}+n'+n_2\leq \boldsymbol{\tau}_{1}^{(j)}+n'+n_1<\boldsymbol{\tau}_2^{\omega^-+1}+1\right\}.
\end{eqnarray}
Also, let 
 \begin{eqnarray}
\widetilde{\mathcal{U}}_n:=\left\{(t_1,t_2,t_3)\in\mathds{Z}^3: t_2+1<t_1+n'+1\leq t_2+n'+n_2\leq t_1+n'+n_1<t_3+1\right\}.
\end{eqnarray}
The probability of error in decoding the $j^{th}$ codeword of Tx~$1$ at Rx~1 is bounded as 
\begin{equation}
\label{ }
\mathbb{P}(\mathrm{error})\leq \mathbb{P}(\mathrm{error},\mathcal{U}_n)+\mathbb{P}(\mathcal{U}_n^c)\leq \mathbb{P}(\mathrm{error},\mathcal{U}_n)+\epsilon,
\end{equation}
where in the last step we have assumed $n$ is large enough so that $\mathbb{P}(\mathcal{U}_n^c)\leq \epsilon$. This follows by Proposition~\ref{prop_3} together with the facts that $\omega^-\neq 0$, $\omega^+=0$ and $\omega_{1,j}=0$. Under the event $\mathcal{U}_n$, error can happen in two possible ways. The first case is when at least one of (\ref{mid1_11}) or (\ref{mid2_22}) is not satisfied for the actual transmitted codeword by Tx~$1$. We denote this error event by $\mathrm{error}_1$. The second case is when both (\ref{mid1_11}) and (\ref{mid2_22}) are satisfied for a codeword that is different from the transmitted codeword by Tx~$1$. We denote this error event by $\mathrm{error}_2$. Then 
\begin{equation}
\label{errev2}
\mathbb{P}(\mathrm{error},\mathcal{U}_n)\leq \mathbb{P}(\mathrm{error}_1,\mathcal{U}_n)+\mathbb{P}(\mathrm{error}_2,\mathcal{U}_n).
\end{equation}
Analysis of $\mathbb{P}(\mathrm{error}_1,\mathcal{U}_n)$ is quite similar to the one offered in Appendix~G. Here, we only address $\mathbb{P}(\mathrm{error}_2,\mathcal{U}_n)$. Let $\delta>0$ and define 
  \begin{eqnarray}
\mathcal{V}_n:=\bigg\{\max\Big\{\Big|\frac{\boldsymbol{\tau}_1^{(j)}}{n}-(j\mu_1+\nu_1)\Big|,\Big|\frac{\boldsymbol{\tau}_{2}^{(\omega^-)}}{n}-(\omega^-\mu_2+\nu_2)\Big|\Big\}<\delta \bigg\}.\notag\\
\end{eqnarray}
Also, let
 \begin{eqnarray}
\widetilde{\mathcal{V}}_n:=\bigg\{(t_1,t_2)\in\mathds{Z}^2: \max\Big\{\Big|\frac{t_1}{n}-(j\mu_1+\nu_1)\Big|,\Big|\frac{t_2}{n}-(\omega^-\mu_2+\nu_2)\Big|\Big\}<\delta\bigg\}.
\end{eqnarray}
We can write
\begin{eqnarray}
\mathbb{P}(\mathrm{error}_2,\mathcal{U}_n)\leq \mathbb{P}(\mathrm{error}_2,\mathcal{U}_n, \mathcal{V}_n)+\mathbb{P}(\mathcal{V}_n^{c})\leq \mathbb{P}(\mathrm{error}_2,\mathcal{U}_n, \mathcal{V}_n)+\epsilon,
\end{eqnarray}
where the last step we are assuming that $n$ is large enough such that $\mathbb{P}(\mathcal{V}_n^{c})\leq \epsilon$ following a similar reasoning in (\ref{bial_11}) in Appendix~F. Let us label the messages of Tx~$1$ as message~$1$ to message~$2^{\lfloor n\eta_1\rfloor}$. Assume the $j^{th}$ transmitted message of Tx~$1$ is message~$1$ and $(\boldsymbol{s}_{1,l})_{l=0}^{n_1-1}$ is the codeword corresponding to message~$2$. Then 
  \begin{eqnarray}
  \label{pok_1122}
\mathbb{P}(\mathrm{error}_2,\mathcal{U}_n,\mathcal{V}_n) \leq 2^{\lfloor n\eta_1\rfloor}\mathbb{P}\left(\mathcal{E}\bigcap\mathcal{F}\bigcap \mathcal{U}_n\bigcap\mathcal{V}_n\right),
\end{eqnarray}
where 
\begin{equation}
\label{ }
\mathcal{E}=\left\{\Big((\boldsymbol{s}_{1,l})_{l=0}^{\boldsymbol{\tau}_2^{(\omega^{-})}-\boldsymbol{\tau}_1^{(j)}+n_2-1},(\boldsymbol{y}_{1,l})_{l=\boldsymbol{\tau}_1^{(j)}+n'+1}^{\boldsymbol{\tau}_2^{(\omega^{-})}+n'+n_2}\Big)\in A_{\epsilon}^{(\boldsymbol{\tau}_2^{(\omega^{-})}-\boldsymbol{\tau}_1^{(j)}+n_2)}[p^{(2)}]\right\}
\end{equation}
and 
\begin{eqnarray}
\mathcal{F}=\left\{\Big((\boldsymbol{s}_{1,l})_{l=\boldsymbol{\tau}_2^{(\omega^{-})}-\boldsymbol{\tau}_1^{(j)}+n_2}^{n_1-1},(\boldsymbol{y}_{1,l})_{l=\boldsymbol{\tau}_2^{(\omega^{-})}+n'+n_2+1}^{\boldsymbol{\tau}_1^{(j)}+n'+n_1}\Big)\in A_{\epsilon}^{(\boldsymbol{\tau}_1^{(j)}-\boldsymbol{\tau}_2^{(\omega^{-})}+n_1-n_2)}[p^{(1)}]\right\}.
\end{eqnarray}
Then 
\begin{eqnarray}
\label{miss11}
\mathbb{P}\left(\mathcal{E}\bigcap\mathcal{F}\bigcap \mathcal{U}_n\bigcap\mathcal{V}_n\right)&=&\sum_{\substack{(t_1,t_2,t_3)\in\,\widetilde{\mathcal{U}}_n\\(t_1,t_2)\in\widetilde{\mathcal{V}}_n}}\mathbb{P}\left(\mathcal{E}\bigcap\mathcal{F}\,\big|\,(\boldsymbol{\tau}_1^{(j)},\boldsymbol{\tau}_2^{(\omega^{-})},\boldsymbol{\tau}_2^{(\omega^{-}+1)})=(t_1,t_2,t_3)\right)\notag\\
&&\hskip2.5cm\times\mathbb{P}\big((\boldsymbol{\tau}_1^{(j)},\boldsymbol{\tau}_2^{(\omega^{-})},\boldsymbol{\tau}_2^{(\omega^{-}+1)})=(t_1,t_2,t_3)\big).
\end{eqnarray}
For any $(t_1,t_2,t_3)\in\,\widetilde{\mathcal{U}}_n$,
\begin{equation}
\label{miss111}
\mathbb{P}\Big(\mathcal{E}\bigcap\mathcal{F}\,\big|\,(\boldsymbol{\tau}_1^{(j)},\boldsymbol{\tau}_2^{(\omega^{-})},\boldsymbol{\tau}_2^{(\omega^{-}+1)})=(t_1,t_2,t_3)\Big)=\mathbb{P}(\mathcal{E}_{t_1,t_2})\mathbb{P}(\mathcal{F}_{t_1,t_2}),
\end{equation}
where
\begin{equation}
\label{ }
\mathcal{E}_{t_1,t_2}=\left\{\Big((\boldsymbol{s}_{1,l})_{l=0}^{t_2-t_1+n_2-1},(\boldsymbol{y}_{1,l})_{l=t_1+n'+1}^{t_2+n'+n_2}\Big)\in A_{\epsilon}^{(t_2-t_1+n_2)}[p^{(2)}]\right\}
\end{equation}
and
\begin{eqnarray}
\mathcal{F}=\left\{\Big((\boldsymbol{s}_{1,l})_{l=t_2-t_1+n_2}^{n_1-1},(\boldsymbol{y}_{1,l})_{l=t_2+n'+n_2+1}^{t_1+n'+n_1}\Big)\in A_{\epsilon}^{(t_1-t_2+n_1-n_2)}[p^{(1)}]\right\}.
\end{eqnarray}
Using the standard properties of jointly typical sequences \cite{16}, we have
\begin{eqnarray}
\label{miss22}
\mathbb{P}(\mathcal{E}_{t_1,t_2})\leq 2^{-(t_2-t_1+n_2)( \psi_1-3\epsilon)}
\end{eqnarray} 
and 
\begin{equation}
\label{miss33}
\mathbb{P}(\mathcal{F}_{t_1,t_2})\leq 2^{-(t_1-t_2+n_1-n_2)( \phi_1-3\epsilon)}.
\end{equation}
By (\ref{miss11}), (\ref{miss111}), (\ref{miss22}) and (\ref{miss33}), 
\begin{eqnarray}
\label{miss444}
&&\mathbb{P}\left(\mathcal{E}\bigcap\mathcal{F}\bigcap \mathcal{U}_n\bigcap\mathcal{V}_n\right)\leq  2^{-n_1( \phi_1-3\epsilon)}2^{-n_2(  \psi_1-\phi_1)}\notag\\&&\hskip2.5cm\times\sum_{\substack{(t_1,t_2,t_3)\in\,\widetilde{\mathcal{U}}_n\\(t_1,t_2)\in\widetilde{\mathcal{V}}_n}}2^{-(t_1-t_2)( \phi_1-  \psi_1)}\mathbb{P}\big((\boldsymbol{\tau}_1^{(j)},\boldsymbol{\tau}_2^{(j^{-})},\boldsymbol{\tau}_2^{(\omega^{-}+1)})=(t_1,t_2,t_3)\big).
\end{eqnarray}
 For any $(t_1,t_2)\in \widetilde{\mathcal{V}}_n$, we have $t_1\geq n(j\mu_1+\nu_1-\delta)$ and $t_2\leq n(\omega^-\mu_2+\nu_2+\delta)$. Using these bounds in (\ref{miss444}), we get 
\begin{eqnarray}
\label{miss44}
\mathbb{P}\left(\mathcal{E}\bigcap\mathcal{F}\bigcap \mathcal{U}_n\bigcap\mathcal{V}_n\right)\leq 2^{-n_1( \phi_1-3\epsilon)}2^{-n_2(\psi_1-\phi_1)}2^{-n(j\mu_1-\omega^-\mu_2+\nu_1-\nu_2-2\delta)(\phi_1-\psi_1)}.
\end{eqnarray}
By (\ref{pok_1122}) and (\ref{miss44}), 
\begin{eqnarray}
\mathbb{P}(\mathrm{error}_2,\,\mathcal{U}_n,\mathcal{V}_n)\leq 2^{-nf(n)},
\end{eqnarray}
 where 
 \begin{eqnarray}
f(n)=(j\mu_1-\omega^-\mu_2+\nu_1-\nu_2-2\delta)(\phi_1-\psi_1)+\frac{n_1}{n}( \phi_1-3\epsilon)+\frac{n_2}{n}(\psi_1-\phi_1)-\frac{\lfloor n\eta_1\rfloor}{n}.
\end{eqnarray}
We have $\lim_{n\to\infty}f(n)=(j\mu_1-\omega^-\mu_2+\nu_1-\nu_2)(\phi_1-\psi_1)+\theta_1\phi_1+\theta_2(\psi_1-\phi_1)-\eta_1-2(\phi_1-\psi_1)\delta-3\theta_1\epsilon$. By~(\ref{help1}), $(j\mu_1-\omega^-\mu_2+\nu_1-\nu_2)(\phi_1-\psi_1)+\theta_1\phi_1+\theta_2(\psi_1-\phi_1)-\eta_1>0$. Therefore, $\lim_{n\to\infty}f(n)>0$ for sufficiently small $\delta$~and~$\epsilon$  and $\mathbb{P}(\mathrm{error}_2,\,\mathcal{U}_n,\mathcal{V}_n)$ decays exponentially with $n$ as desired.
\section*{Appendix~I}
We need the following lemma:
\begin{lem}
\label{lem_7}
Let $m,n$ be positive integers. The number of non-decreasing sequences of length $n$ whose entries are among the numbers $1,2,\cdots, m$ is $\binom{m+n-1}{n}$.
\end{lem}
\begin{proof}
Define the sets $\mathcal{A}$ and $\mathcal{B}$  by 
\begin{eqnarray}
\mathcal{A}:=\Big\{(x_1,\cdots, x_{n+1}): \textrm{$x_1,\cdots, x_{n+1}$ are integers}, x_1\geq 1, x_2,\cdots, x_{n+1}\geq 0,  x_1+\cdots+x_{n+1}=m\Big\}\notag\\
\end{eqnarray}
and
 \begin{eqnarray}
\mathcal{B}:=\Big\{(y_1,\cdots, y_n): \textrm{$y_1,\cdots, y_n$ are integers}, 1\leq y_1\leq y_2\leq \cdots \le y_n\leq m\Big\},
\end{eqnarray}
respectively. Define the map 
\begin{eqnarray}
\begin{array}{c}
      f:\mathcal{A}\to \mathcal{B}    \\
         (x_1,\cdots, x_{n+1})\mapsto (x_1,x_1+x_2, x_1+x_2+x_3,\cdots, x_1+\cdots+x_n).
\end{array}
\end{eqnarray}
The codomain of the map $f$ is as promised because for any $(x_1,\cdots, x_{n+1})\in\mathcal{A}$, $1\leq x_1\leq x_1+x_2\leq x_1+x_2+x_3\leq \cdots\leq x_1+\cdots+x_n\leq m$. We make the following observations:
\begin{itemize}
  \item The map $f$ is one-to-one. In fact, let $(x_1,\cdots, x_{n+1}), (x'_1,\cdots, x'_{n+1})\in\mathcal{A}$ and $f(x_1,\cdots, x_{n+1})=f(x'_1,\cdots, x'_{n+1})$. Then 
\begin{eqnarray}
\begin{array}{c}
      x_1=x'_1    \\
      x_1+x_2=x'_1+x'_2\\
      \vdots\\
      x_1+\cdots+x_n=x'_1+\cdots+x'_n
\end{array}.
\end{eqnarray}
It follows that $x_i=x'_i$ for any $1\leq i\leq n$ and hence, $x_{n+1}=m-\sum_{i=1}^n x_i=m-\sum_{i=1}^n x'_i=x'_{n+1}$. We conclude that $(x_1,\cdots, x_{n+1})=(x'_1,\cdots, x'_{n+1})$.
  \item The map $f$ is onto. To see this, let $(y_1,\cdots, y_n)\in \mathcal{B}$. Define $x_1:=y_1$, $x_{i}:=y_i-y_{i-1}$ for any $2\leq i\leq n$ and $x_{n+1}=m-\sum_{i=1}^{n} x_i$. Then it is easy to see that $(x_1,\cdots, x_{n+1})\in \mathcal{A}$ and $f(x_1,\cdots, x_{n+1})=(y_1,\cdots, y_n)$.
\end{itemize} 
It follows that $f$ is a bijection between $\mathcal{A}$ and $\mathcal{B}$ and hence, $|\mathcal{A}|=|\mathcal{B}|$. But, $|\mathcal{A}|$ is exactly the number of solutions for the tuples $(z_1,\cdots, z_{n+1})$ where $z_1:=x_1, z_2:=x_2+1\cdots, z_{n+1}:=x_{n+1}+1$ are positive integers and $z_1+\cdots+z_{n+1}=m+n$. This number is known to be $\binom{m+n-1}{n}$.
\end{proof}
Note that $\{(j,u_j,v_j): 1\leq j\leq N_2\}$ is a state if and only if 
\begin{equation}
\label{ }
1\leq i_1\leq j_1\leq i_2\leq j_2\leq \cdots\leq i_{N_2-1}\leq j_{N_2-1}\leq i_{N_2}\leq j_{N_2}\leq 2N_1+1,
\end{equation} 
i.e., $(i_1,j_1,i_2,j_2,\cdots, i_{N_2},j_{N_2})$ is a non-decreasing sequence of integers whose entries are among the numbers $1,2,\cdots,2N_1+1$. By Lemma~\ref{lem_7}, the number of such sequences is $\binom{2N_1+2N_2}{2N_2}$. 
\section*{Appendix~J; Proof of Proposition~\ref{prop_8}}
As shown earlier in (\ref{dim_11}), we only need to assume $j^*=1$. Then  $\mathscr{A}_{1}=\{R>\lambda: \alpha<\mu<\alpha+\theta\}=\left(\lambda,\big(1+\frac{\alpha}{\theta}\big)\lambda\right)$, $\mathscr{C}_{1}=\{R>\lambda: \mu>\alpha+\theta\}=\left(\big(1+\frac{\alpha}{\theta}\big)\lambda,\infty\right)$ and $\mathscr{B}_{1}=\mathscr{D}_{1}=\emptyset$. For notational simplicity, we show $\phi_\gamma$ and $\psi_{\gamma}$ by $\kappa$ and $\kappa'$, respectively. It is beneficial to  our presentation to write $\mathcal{P}(x,y;\gamma)$ in (\ref{vart4}) as  
\begin{eqnarray}
\mathcal{P}(x,y;\gamma)=\widetilde{\mathcal{P}}(x,y;\gamma)\bigcap\mathcal{U}(\gamma),
\end{eqnarray}
where
\begin{eqnarray}
\widetilde{\mathcal{P}}(x,y;\gamma):=\left\{R_c>\lambda: \big(1-\frac{x}{\lambda}(\kappa-\kappa')\big)R_c<\kappa'-(\kappa-\kappa')\frac{y}{\theta} \right\}
\end{eqnarray}
and
\begin{eqnarray}
\mathcal{U}(\gamma):=\left\{R_c>\lambda: \gamma\leq \big(\frac{1}{N}+\frac{R_c}{\lambda}\big)P \right\}.
\end{eqnarray}
 We have
  \begin{equation}
\label{huber_0}
\mathcal{R}_1=(\mathcal{R}_{\mathrm{sym}}\bigcap\mathscr{A}_{1})\bigcup(\mathcal{R}_{\mathrm{sym}}\bigcap\mathscr{C}_{1}),
\end{equation}
where
\begin{equation}
\label{huber_1}
\mathcal{R}_{\mathrm{sym}}\bigcap\mathscr{A}_{1}=\bigcup_{\gamma\geq 0}\left(\mathcal{U}(\gamma)\bigcap\widetilde{\mathcal{P}}(1,\theta;\gamma)\bigcap\widetilde{\mathcal{P}}(0,-\alpha;\gamma)\bigcap\mathscr{A}_{1}\right)
\end{equation}
and
\begin{eqnarray}
\label{huber_2}
\mathcal{R}_{\mathrm{sym}}\bigcap \mathscr{C}_{1}&=&\bigcup_{\gamma\ge0}\left(\mathcal{U}(\gamma)\bigcap\widetilde{\mathcal{P}}(0,-\alpha;\gamma)\bigcap\mathscr{C}_{1}\right),
\end{eqnarray}
by (\ref{PERT1}) and (\ref{PERT3}), respectively.  We have
\begin{eqnarray}
\label{huber_012}
\widetilde{\mathcal{P}}(0,-\alpha;\gamma)\bigcap\mathscr{C}_{1}=\Big(\big(1+\frac{\alpha}{\theta}\big)\lambda,\kappa'+\frac{\alpha}{\theta}(\kappa-\kappa')\Big).
\end{eqnarray}
To compute the right side of (\ref{huber_1}), it is enough to note that 
\begin{eqnarray}
\label{huber_5}
\widetilde{\mathcal{P}}(1,\theta;\gamma)&=&\Big\{R>\lambda: \big(1-\frac{1}{\lambda}(\kappa-\kappa')\big)R<2\kappa'-\kappa\Big\}\notag\\&=&\left\{\begin{array}{cc}
     \big(\lambda,\frac{2\kappa'-\kappa}{1-\frac{1}{\lambda}(\kappa-\kappa')}\big)   &  \kappa-\kappa'<\lambda<\kappa'\\
      \Big(\max\left\{\lambda,\frac{2\kappa'-\kappa}{1-\frac{1}{\lambda}(\kappa-\kappa')}\right\},\infty\Big)   &  \kappa\ge2\kappa',\,\,\kappa-\kappa'>\lambda\\
     (\lambda,\infty) & \kappa\le2\kappa',\,\,\,\kappa-\kappa'>\lambda\\
     \emptyset &\textrm{otherwise}
\end{array}\right.
\end{eqnarray}
and
\begin{eqnarray}
\label{Code_2}
\widetilde{\mathcal{P}}(0,-\alpha;\gamma)\bigcap\mathscr{A}_{1}=\Big(\lambda, \min\Big\{\big(1+\frac{\alpha}{\theta}\big)\lambda, \kappa'+\frac{\alpha}{\theta}(\kappa-\kappa')\Big\}\Big).
\end{eqnarray}
Having~(\ref{huber_5}) and (\ref{Code_2}), we can find $\mathcal{R}_{\mathrm{sym}}\bigcap\mathscr{A}_{1}$ in (\ref{huber_1}) which together with $\mathcal{R}_{\mathrm{sym}}\bigcap \mathscr{C}_{1}$ in (\ref{huber_2}) and (\ref{huber_012}) complete the description of $\mathcal{R}_1$ in (\ref{huber_0}). To simplify our computations, let us consider two cases:
\begin{enumerate}
  \item Assume 
  \begin{equation}
\label{code_1111}
\big(1+\frac{\alpha}{\theta}\big)\lambda<\kappa'+\frac{\alpha}{\theta}(\kappa-\kappa').
\end{equation}
Using this in (\ref{Code_2}), we see that $\widetilde{\mathcal{P}}(0,-\alpha;\gamma)\bigcap\mathscr{A}_{1}=\big(\lambda, \big(1+\frac{\alpha}{\theta}\big)\lambda\big)$ and we get
\begin{eqnarray}
\label{liver_1}
\widetilde{\mathcal{P}}(1,\theta;\gamma)\bigcap\widetilde{\mathcal{P}}(0,-\alpha;\gamma)\bigcap\mathscr{A}_{1}=\left\{\begin{array}{cc}
     \big(\lambda,\min\big\{(1+\frac{\alpha}{\theta})\lambda,\frac{2\kappa'-\kappa}{1-\frac{1}{\lambda}(\kappa-\kappa')}\big)\big\}   &  \kappa-\kappa'<\lambda<\kappa'\\
      \Big(\max\left\{\lambda,\frac{2\kappa'-\kappa}{1-\frac{1}{\lambda}(\kappa-\kappa')}\right\},(1+\frac{\alpha}{\theta})\lambda\Big)   &  \kappa\ge2\kappa',\,\,\kappa-\kappa'>\lambda\\
     (\lambda,(1+\frac{\alpha}{\theta})\lambda) & \kappa\le2\kappa',\,\,\,\kappa-\kappa'>\lambda\\
     \emptyset &\textrm{otherwise}
\end{array}\right..\notag\\
\end{eqnarray}
Simple algebra shows that 
\begin{eqnarray}
\label{sard_1}
&&\big(\lambda-(\kappa-\kappa')\big)\Big(\kappa'+\frac{\alpha}{\theta}(\kappa-\kappa')-\frac{2\kappa'-\kappa}{1-\frac{1}{\lambda}(\kappa-\kappa')}\Big)\notag\\&&=(\kappa-\kappa')\big((1+\frac{\alpha}{\theta})\lambda-(\kappa'+\frac{\alpha}{\theta}(\kappa-\kappa'))\big).
\end{eqnarray}
By (\ref{code_1111}) the right side of (\ref{sard_1}) is negative. Therefore, 
\begin{eqnarray}
\kappa-\kappa'<\lambda\implies \kappa'+\frac{\alpha}{\theta}(\kappa-\kappa')<\frac{2\kappa'-\kappa}{1-\frac{1}{\lambda}(\kappa-\kappa')}\implies \big(1+\frac{\alpha}{\theta}\big)\lambda <\frac{2\kappa'-\kappa}{1-\frac{1}{\lambda}(\kappa-\kappa')},
\end{eqnarray}
where in the last step we use (\ref{code_1111}). Therefore, the interval in the first row of (\ref{liver_1}) becomes $(\lambda,(1+\frac{\alpha}{\theta})\lambda)$, i.e., 
\begin{eqnarray}
\label{liver_2}
\widetilde{\mathcal{P}}(1,\theta;\gamma)\bigcap\widetilde{\mathcal{P}}(0,-\alpha;\gamma)\bigcap\mathscr{A}_{1}=\left\{\begin{array}{cc}
     \big(\lambda,(1+\frac{\alpha}{\theta})\lambda)  &  \kappa-\kappa'<\lambda<\kappa'\\
      \Big(\max\left\{\lambda,\frac{2\kappa'-\kappa}{1-\frac{1}{\lambda}(\kappa-\kappa')}\right\},(1+\frac{\alpha}{\theta})\lambda\Big)   &  \kappa\ge2\kappa',\,\,\kappa-\kappa'>\lambda\\
     (\lambda,(1+\frac{\alpha}{\theta})\lambda) & \kappa\le2\kappa',\,\,\,\kappa-\kappa'>\lambda\\
     \emptyset &\textrm{otherwise}
\end{array}\right..\notag\\
\end{eqnarray}
It is notable that in the second line in (\ref{liver_2}) it is always true that $\frac{2\kappa'-\kappa}{1-\frac{1}{\lambda}(\kappa-\kappa')}<(1+\frac{\alpha}{\theta})\lambda$ provided  $\kappa-\kappa'>\lambda$ and hence, the interval $\big(\max\big\{\lambda,\frac{2\kappa'-\kappa}{1-\frac{1}{\lambda}(\kappa-\kappa')}\big\},(1+\frac{\alpha}{\theta})\lambda\big)$ is nonempty.\footnote{Multiplying both sides of $\frac{2\kappa'-\kappa}{1-\frac{1}{\lambda}(\kappa-\kappa')}<(1+\frac{\alpha}{\theta})\lambda$ by the negative number $1-\frac{1}{\lambda}(\kappa-\kappa')$ yields $\big(1+\frac{\alpha}{\theta}\big)\lambda<\kappa'+\frac{\alpha}{\theta}(\kappa-\kappa')$ which is our assumption in (\ref{code_1111}).}
  \item Assume
   \begin{equation}
\label{rent_11}
\big(1+\frac{\alpha}{\theta}\big)\lambda\ge\kappa'+\frac{\alpha}{\theta}(\kappa-\kappa').
\end{equation}
Then 
\begin{equation}
\label{ }
\widetilde{\mathcal{P}}(0,-\alpha;\gamma)\bigcap\mathscr{C}_{1}=\emptyset.
\end{equation}
 By (\ref{Code_2}) and (\ref{rent_11}), $\widetilde{\mathcal{P}}(0,-\alpha;\gamma)\bigcap\mathscr{A}_{1}=\big(\lambda, \kappa'+\frac{\alpha}{\theta}(\kappa-\kappa')\big)$. Using this together with (\ref{huber_1}) and (\ref{huber_5}),   
\begin{eqnarray} 
\label{Sard_1111}
&&\widetilde{\mathcal{P}}(1,\theta;\gamma)\bigcap\widetilde{\mathcal{P}}(0,-\alpha;\gamma)\bigcap\mathscr{A}_{1}\notag\\&&=\left\{\begin{array}{cc}
    \Big(\lambda,\min\big\{\frac{2\kappa'-\kappa}{1-\frac{1}{\lambda}(\kappa-\kappa')}, \kappa'+\frac{\alpha}{\theta}(\kappa-\kappa')\big\}\Big)   &   \kappa-\kappa'<\lambda<\kappa'  \\
    \Big(\max\big\{\lambda,\frac{2\kappa'-\kappa}{1-\frac{1}{\lambda}(\kappa-\kappa')}\big\},\kappa'+\frac{\alpha}{\theta}(\kappa-\kappa')\Big)  & \kappa\ge2\kappa',\,\,\kappa-\kappa'>\lambda, \\
(\lambda,\kappa'+\frac{\alpha}{\theta}(\kappa-\kappa'))& \kappa\le2\kappa',\,\,\,\kappa-\kappa'>\lambda\\
\emptyset &\textrm{otherwise}
\end{array}\right..
\end{eqnarray}
The right side of (\ref{sard_1}) is nonnegative due to (\ref{rent_11}).  Therefore,
\begin{equation}
\label{soar_1}
\kappa-\kappa'<\lambda\implies\frac{2\kappa'-\kappa}{1-\frac{1}{\lambda}(\kappa-\kappa')}\leq\kappa'+\frac{\alpha}{\theta}(\kappa-\kappa')
\end{equation}
and
\begin{equation}
\label{soar_2}
\kappa-\kappa'>\lambda\implies\frac{2\kappa'-\kappa}{1-\frac{1}{\lambda}(\kappa-\kappa')}\geq\kappa'+\frac{\alpha}{\theta}(\kappa-\kappa').
\end{equation}
Let us make the following observations: 
\begin{itemize}
  \item By (\ref{soar_1}), the interval in the first row of the definition of $\mathcal{R}_{\mathrm{sym}}\bigcap\mathscr{A}_1$ in (\ref{Sard_1111}) becomes $\big(\lambda,\frac{2\kappa'-\kappa}{1-\frac{1}{\lambda}(\kappa-\kappa')}\big)$.
  \item By (\ref{soar_2}), the interval in the second row of  (\ref{Sard_1111}) is empty.
  \item Combining $\kappa-\kappa'>\lambda$ with (\ref{rent_11}), we get $\lambda>\kappa'$ and hence, $\kappa>2\kappa'$. This shows that the constraints in the third row of (\ref{Sard_1111}) are not compatible with (\ref{rent_11}).
\end{itemize}
Based on these observations, we can simplify (\ref{Sard_1111}) as 
 \begin{eqnarray}
\label{Sard_11}
\widetilde{\mathcal{P}}(1,\theta;\gamma)\bigcap\widetilde{\mathcal{P}}(0,-\alpha;\gamma)\bigcap\mathscr{A}_{1}=\left\{\begin{array}{cc}
    \big(\lambda,\frac{2\kappa'-\kappa}{1-\frac{1}{\lambda}(\kappa-\kappa')}\big)   &  \kappa-\kappa'<\lambda<\kappa'  \\
\emptyset &\textrm{otherwise}
\end{array}\right..
\end{eqnarray}
\end{enumerate} 

     \begin{figure}[t]
\centering
\subfigure[]{
\includegraphics[scale=0.67]{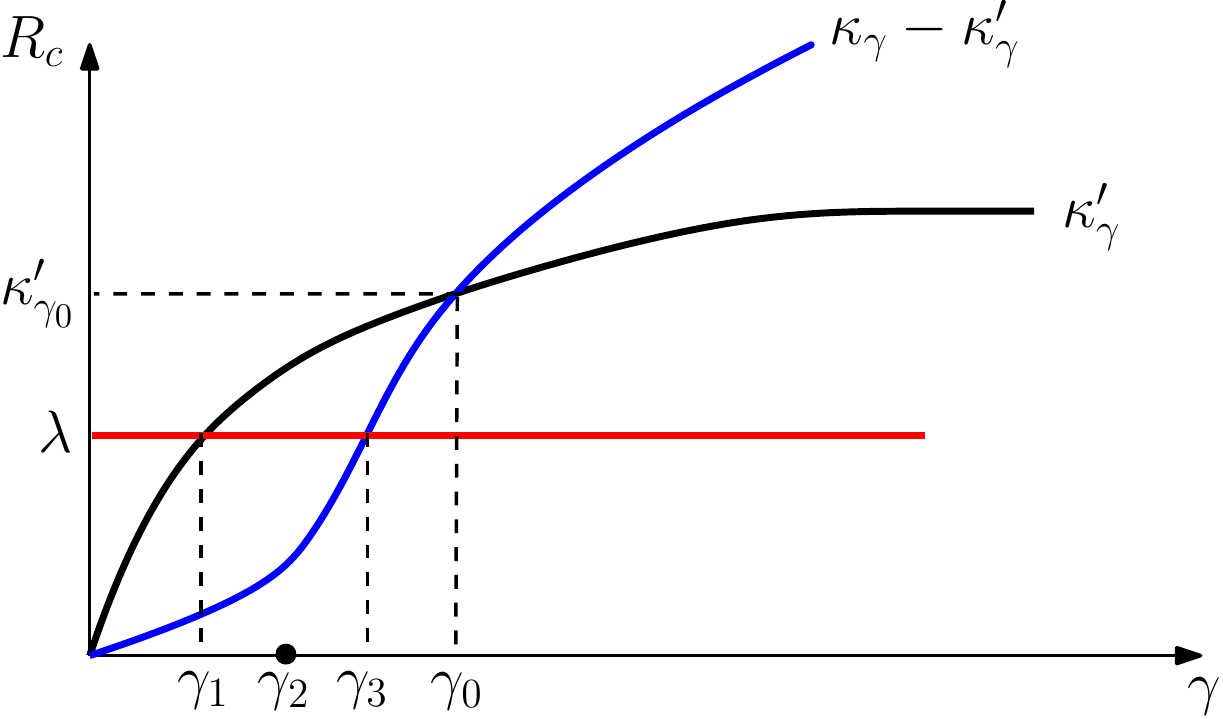}
\label{ref1}
}
\subfigure[]{
\includegraphics[scale=0.67]{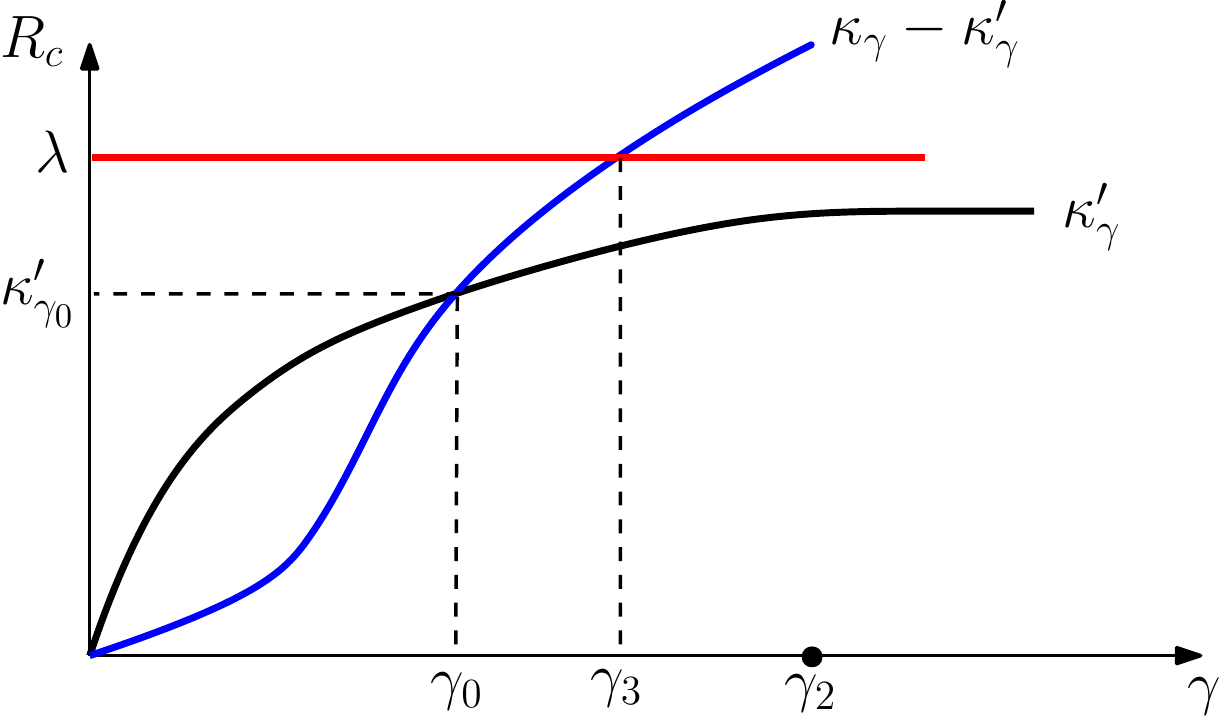}
\label{ref2}
}
\label{fig:subfigureExample}
\caption[Optional caption for list of figures]{Plots of $\psi_\gamma$ (black curve) and $\phi_{\gamma}-\psi_{\gamma}$ (blue curve) as function of $\gamma$. The constant level $\lambda$ is shown in red. The value of $\gamma$ for which $\psi_{\gamma}=\phi_{\gamma}-\psi_{\gamma}$ is denoted by $\gamma_0$. Panel~(a) shows a scenario where $\lambda<\psi_{\gamma_0}$. The solutions for $\gamma$ in $\psi_{\gamma}=\lambda$ and $\phi_{\gamma}-\psi_{\gamma}=\lambda$ are shown by $\gamma_1$ and $\gamma_3$, respectively. It is only for values of $\gamma$ in the interval $(\gamma_1,\gamma_3)$ that $\phi_{\gamma}-\psi_{\gamma}<\lambda<\psi_{\gamma}$. Moreover, the equation $\psi_\gamma+\frac{\alpha}{\theta}\phi_{\gamma}-\psi_{\gamma}=(1+\frac{\alpha}{\theta})$ is solved for $\gamma=\gamma_2$ where $\gamma_1<\gamma_2<\gamma_3$. Panel~(b) shows a scenario where $\lambda>\psi_{\gamma_0}$.  The conditions $\phi_{\gamma}-\psi_{\gamma}<\lambda<\psi_{\gamma}$ no longer hold.   }
\label{pict_101}
\end{figure}  

Define $\gamma_0$ as the value of $\gamma$ that solves $\kappa'=\kappa-\kappa'$. We consider two cases:
\begin{enumerate}
  \item Let $\lambda<\psi_{\gamma_0}$. This situation is shown in panel~(a) of Fig.~\ref{pict_101}. The solutions for $\gamma$ in $\psi_{\gamma}=\lambda$ and $\phi_{\gamma}-\psi_{\gamma}=\lambda$ are shown by $\gamma_1$ and $\gamma_3$, respectively. It is only for values of $\gamma$ in the interval $(\gamma_1,\gamma_3)$ that $\phi_{\gamma}-\psi_{\gamma}<\lambda<\psi_{\gamma}$. Moreover, the equation $\psi_\gamma+\frac{\alpha}{\theta}(\phi_{\gamma}-\psi_{\gamma})=(1+\frac{\alpha}{\theta})\lambda$ is solved for $\gamma=\gamma_2$ where $\gamma_1<\gamma_2<\gamma_3$.\footnote{if $\gamma<\gamma_1$, then $\kappa'$ and $\kappa-\kappa'$ are both smaller than $\lambda$ and if $\gamma>\gamma_3$, then $\kappa'$ and $\kappa-\kappa'$ are both larger than $\lambda$. In either case, $\psi_\gamma+\frac{\alpha}{\theta}(\phi_{\gamma}-\psi_{\gamma})\neq(1+\frac{\alpha}{\theta})\lambda$. } We have
  \begin{eqnarray}
  \label{hi_444}
\widetilde{\mathcal{P}}(0,-\alpha;\gamma)\bigcap\mathscr{C}_{1}=\left\{\begin{array}{cc}
    \emptyset  & \gamma\leq\gamma_2   \\
    \big(\big(1+\frac{\alpha}{\theta}\big)\lambda,\kappa'+\frac{\alpha}{\theta}(\kappa-\kappa')\big)  &   \gamma>\gamma_2
\end{array}\right.
\end{eqnarray}
and
 \begin{eqnarray}
 \label{hi_555}
\widetilde{\mathcal{P}}(1,\theta;\gamma)\bigcap\widetilde{\mathcal{P}}(0,-\alpha;\gamma)\bigcap\mathscr{A}_{1}=\left\{\begin{array}{cc}
    \emptyset  & \gamma\leq\gamma_1   \\
    \big(\lambda,\frac{2\kappa'-\kappa}{1-\frac{1}{\lambda}(\kappa-\kappa')}\big)  &   \gamma_1<\gamma<\gamma_2\\
   \big(\lambda,(1+\frac{\alpha}{\theta})\lambda\big) & \gamma>\gamma_2
\end{array}\right.
\end{eqnarray}
  \item Let $\lambda\geq \psi_{\gamma_0}$. This situation is shown in panel~(b) of Fig.~\ref{pict_101}. The conditions $\kappa-\kappa'<\lambda<\kappa'$ no longer hold. Moreover,  $\gamma_3<\gamma_2<\gamma_1$ where we let $\gamma_1=\infty$ if it does not exist. The set $\widetilde{\mathcal{P}}(0,-\alpha;\gamma)\bigcap\mathscr{C}_{1}$ is given by (\ref{hi_444}) and
  \begin{eqnarray}
  \label{hi_222}
\widetilde{\mathcal{P}}(1,\theta;\gamma)\bigcap\widetilde{\mathcal{P}}(0,-\alpha;\gamma)\bigcap\mathscr{A}_{1}=\left\{\begin{array}{cc}
    \emptyset  & \gamma\leq\gamma_2   \\
   \big(\frac{2\kappa'-\kappa}{1-\frac{1}{\lambda}(\kappa-\kappa')},(1+\frac{\alpha}{\theta})\lambda\big) & \gamma_2<\gamma\leq\gamma_1\\
   \big(\lambda,(1+\frac{\alpha}{\theta})\lambda\big)& \gamma> \gamma_1
\end{array}\right..
\end{eqnarray}
\end{enumerate}
By (\ref{huber_0}), (\ref{huber_1}) and (\ref{huber_2}), 
\begin{eqnarray}
\mathcal{R}_{\mathrm{sym}}=\bigcup_{\gamma\geq 0}\bigg(\mathcal{U}(\gamma)\bigcap\Big(\big(\widetilde{\mathcal{P}}(1,\theta;\gamma)\bigcap\widetilde{\mathcal{P}}(0,-\alpha;\gamma)\bigcap\mathscr{A}_{1}\big)\bigcup\big(\widetilde{\mathcal{P}}(0,-\alpha;\gamma)\bigcap\mathscr{C}_{1}\big)\Big)\bigg).
\end{eqnarray}
Using the expressions in (\ref{hi_444}), (\ref{hi_555}) and (\ref{hi_222}), 
\begin{eqnarray}
\big(\widetilde{\mathcal{P}}(1,\theta;\gamma)\bigcap\widetilde{\mathcal{P}}(0,-\alpha;\gamma)\bigcap\mathscr{A}_{1}\big)\bigcup\big(\widetilde{\mathcal{P}}(0,-\alpha;\gamma)\bigcap\mathscr{C}_{1}\big)=\big(f(\gamma),g(\gamma)\big),
\end{eqnarray}
where $f$ and $g$ are defined in (\ref{f_1111}) and (\ref{g_1111}) under the assumption that $\lambda<\psi_{\gamma_0}$ and $\lambda\ge\psi_{\gamma_0}$, respectively.



\begin{thebibliography}{10}
\providecommand{\url}[1]{#1} \csname url@rmstyle\endcsname
\providecommand{\newblock}{\relax} \providecommand{\bibinfo}[2]{#2}
\providecommand\BIBentrySTDinterwordspacing{\spaceskip=0pt\relax}
\providecommand\BIBentryALTinterwordstretchfactor{4}
\providecommand\BIBentryALTinterwordspacing{\spaceskip=\fontdimen2\font
plus \BIBentryALTinterwordstretchfactor\fontdimen3\font minus
  \fontdimen4\font\relax}
\providecommand\BIBforeignlanguage[2]{{%
e^{}andafter\ifx\csname l@#1\endcsname\relax
\typeout{** WARNING: IEEEtran.bst: No hyphenation pattern has been}%
\typeout{** loaded for the language `#1'. Using the pattern for}%
\typeout{** the default language instead.}%
\else \language=\csname l@#1\endcsname \fi #2}}


\bibitem{gamkim}
A.~El~Gamal and Y-H.~Kim, ``Network Information Theory'', \emph{Cambridge University Press},~2012. 


\bibitem{et1}
R. H. Etkin, D. N. C. Tse and H. Wang, ``Gaussian interference channel capacity to within one bit'', \emph{ IEEE Trans. Inf. Theory}, vol.~54, no.~12, pp.~5534-5562, Dec.~2008.


 
  
  \bibitem{HK}
  T. S. Han and K. Kobayashi, ``A new achievable rate region for the
interference channel'', \emph{IEEE Trans. Inf. Theory}, vol.~27, no.~1, pp.~49Ð60, Jan.~1981. 

\bibitem{eph-haj}
A.~Ephremides and B.~Hajek, ``Information theory and communication networks: An unconsummated union'',  \emph{IEEE Trans. on Inf. Theory}, vol.~44, no.~6, pp.~2416-2434, October~1998. 


  
  
  
  
  
 



\bibitem{sim}
O.~Simeone, Y.~Bar-Ness and U.~Spagnolini, ``Stable throughput of cognitive radios with and without relaying capability'', \emph{IEEE Transactions on Communications}, vol.~55, no.~12, pp.~2351-2360, Dec.~2007.

\bibitem{dev}
N. Devroye, P. Mitran and V. Tarokh, ``Achievable rates in cognitive radio channels'', \emph{IEEE Transactions  on Information Theory}, vol.~52, no.~5, pp.~1813-1827, May~2006.

\bibitem{dash}
D.~Dash and A.~Sabharval, ``Paranoid Secondary: Waterfilling in a cognitive interference channel with partial knowledge'', \emph{IEEE Transactions on Wireless Communications}, vol.~11, no.~3, pp.~1045-1055, March~2012.

\bibitem{khude}
N.~Khude, V.~Prabhakaran and P.~Viswanath, ``Harnessing bursty interference'', \emph{Information Theory Workshop}, ITW~2009, Volos, Greece. 

\bibitem{minero}
P.~Minero, M.~Franceschetti and D.~N.~C.~Tse, ``Random Access: An information-theoretic perspective'', \emph{IEEE Trans. on Inf. Theory}, vol.~58, no.~2, pp.~909-930, February~2012.

\bibitem{wang1}
I-H.~Wang and S.~Diggavi, ``Interference channels with bursty traffic and delayed feedback'', \emph{IEEE Workshop on Signal Processing Advances in Wireless Communications}, SPAWC~2013.

\bibitem{kim}
S.~Kim, I-H.~Wang and C.~Suh, ``A relay can increase degrees of freedom in bursty MIMO interference networks'', \emph{Int. Symp. on Inf. Theory}, ISIT~2015, Hong~Kong.



\bibitem{cov}
T. M. Cover, R. J. Mceliece and E.C. Posner, ``Asynchronous multiple-access channel capacity'', \emph{IEEE Trans. on Inf. Theory}, vol.~27, no.~4, pp.~409-413, July~1981.


\bibitem{hui}
J. Y. N. Hui and P. A. Humblet, ``The capacity region of the  totally asynchronous multiple access channel'', \emph{IEEE Trans. on Inf. Theory}, vol~31, no.~2, pp.~207-216,  March~1985. 





\bibitem{massey1}
J.~L.~Massey and P.~Mathys, ``The collision channel without feedback'', \emph{IEEE Trans. on Inf. Theory}, vol.~31, no.~2, pp.~192-204, March~1985.


\bibitem{calvo}
E.~Calvo, J.~R.~Fonollosa and J.~Vidal, ``On the totally asynchronous interference channel with single-user receivers'',  \emph{Int. Symp. on Inf. Theory}, ISIT~2009, Seoul, Korea, June~2009.

\bibitem{verdu}
S. Verd\'u and T.S. Han, ``A general formula for channel capacity'', \emph{IEEE Trans. on Inf. Theory}, vol.~40, no.~7, pp.~1147-1157, July~1994.












\bibitem{Chandar11}
V.~Chandar, A.~Tchamkerten and G.~Wornell, ``Optimal sequential frame synchronization'',  \emph{ IEEE Transactions on Information TheoTrans. on Inf. Theoryry}, vol.~54, no.~8, pp.~3725-3728, August~2008.






\bibitem{kam}
K.~Moshksar and A.~K.~Khandani, ``Decentralized wireless networks with asynchronous users and burst transmissions'', \emph{IEEE Trans. Inf. Theory}, vol.~61, no.~7, pp.~3851-3880, July~2015.


  \bibitem{16}
T. M. Cover and J. A. Thomas, ``Elements of information theory", \emph{John Wiley and Sons, Inc.}, 1991.

\bibitem{upfal}
M.~Mitzenmacher and E.~Upfal, ``Probability and computing: Randomized algorithms and probabilistic analysis'', \emph{Cambridge University Press}, 2005.



\bibitem{lugosi}
S. Boucheron, G. Lugosi and P. Massart, ``Concentration inequalities: 
A nonasymptotic theory of independence'', \emph{Oxford University Press}, 2013. 


\bibitem{166}
R. M. Dudley, ``Real analysis and probability'', \emph{Cambridge University Press}, 2002.







 


 \end{thebibliography}
\end{document}